\documentclass[aps,pra,twocolumn,groupedaddress,showpacs,superscriptaddress,scrartcl,longbibliography]{revtex4-2}
\usepackage{placeins,float,xcolor,array,multirow,float,natbib,comment,amssymb,amsmath,graphicx,epstopdf,hyperref,amsthm,dsfont,enumerate,bm,color}

\newcolumntype{P}[1]{>{\centering\arraybackslash}p{#1}}
\newcolumntype{M}[1]{>{\centering\arraybackslash}m{#1}}

\newtheorem{definition}{Definition}
\newtheorem{theorem}{Theorem}
\newtheorem{lemma}{Lemma}
\newtheorem{corollary}{Corollary}

\newtheorem*{theorem*}{Theorem 1}
\newtheorem*{lemma1*}{Lemma 1}
\newtheorem*{lemma2*}{Lemma 2}
\newtheorem*{lemma3*}{Lemma 3}
\newtheorem*{lemma4*}{Lemma 4}
\newtheorem*{lemma5*}{Lemma 5}
\newtheorem*{lemma6*}{Lemma 6}
\newtheorem*{lemma7*}{Lemma 7}

\newcommand{\dpg}[1]{{\color{black}{#1}}}	
\newcommand{\dpgg}[1]{{\color{black}{#1}}}	
\newcommand{\dpggg}[1]{{\color{black}{#1}}}	

\begin{document}
\interfootnotelinepenalty=10000

\title{Multiphoton and side-channel attacks in mistrustful quantum cryptography}

\author{Mathieu Bozzio}
\affiliation{Faculty of Physics, University of Vienna, VCQ, Boltzmanngasse 5, 1090 Vienna, Austria}
\affiliation{Sorbonne Universit\'e, CNRS, LIP6, 4 Place Jussieu, F-75005 Paris, France}

\author{Adrien Cavaill\`es}
\affiliation{Sorbonne Universit\'e, CNRS, LIP6, 4 Place Jussieu, F-75005 Paris, France}

\author{Eleni Diamanti}
\affiliation{Sorbonne Universit\'e, CNRS, LIP6, 4 Place Jussieu, F-75005 Paris, France}

\author{Adrian Kent}
%\email[]{apak@damtp.cam.ac.uk}
\affiliation{Centre for Quantum Information and Foundations, DAMTP,
Centre for Mathematical Sciences, University of Cambridge,
Wilberforce Road, Cambridge, CB3 0WA, United Kingdom}
\affiliation{Perimeter Institute for Theoretical Physics, 31 Caroline Street North, Waterloo, ON N2L 2Y5, Canada}

\author{Dami\'an Pital\'ua-Garc\'ia}
\email[]{D.Pitalua-Garcia@damtp.cam.ac.uk}
\affiliation{Centre for Quantum Information and Foundations, DAMTP,
Centre for Mathematical Sciences, University of Cambridge,
Wilberforce Road, Cambridge, CB3 0WA, United Kingdom}

\date{\today}

\begin{abstract}

Mistrustful cryptography includes important tasks like bit commitment, oblivious transfer, coin flipping, secure computations, position authentication, digital signatures and secure unforgeable tokens. Practical quantum implementations presently use photonic setups. In many such implementations, Alice sends photon pulses encoding quantum states and Bob chooses measurements on these states. In practice, Bob generally uses single photon threshold detectors, which cannot distinguish the number of photons in detected pulses. Also, losses and other imperfections require Bob to report the detected pulses. Thus, malicious Alice can send and track multiphoton pulses and thereby gain information about Bob's measurement choices, violating the protocols' security. Here, we provide a theoretical framework for analysing such multiphoton attacks, and present known and new attacks.  We illustrate the power of these attacks with an experiment, and study their application to earlier experimental demonstrations of mistrustful quantum cryptography. We analyse countermeasures based on selective reporting and prove them inadequate. We also discuss side-channel attacks where Alice controls further degrees of freedom or sends other physical systems.
\end{abstract}

\maketitle

%\eee{COMMENT: I suggest to erase the blue text.}

\section{Introduction}
Quantum cryptography promises that cryptographic
tasks can be implemented with provable security, assuming
only the validity of quantum theory.   
As with most guarantees, though, one needs to 
study the small print.   Security proofs are based
on idealized models of quantum cryptosystems, which
do not necessarily characterize the behaviour of
real world equipment.   Hence
apparently faithful implementations of provably
secure protocols can be vulnerable to attacks.   
A wide range of attacks on practical quantum key distribution 
systems have been noted \dpgg{\cite{GFKZR06,MAS06,ZFQCL08,FTLM09,M09,LWWESM10,LWHJCYLZLZGBH11,SRKBPMLM15,JSKEML16,QHWCYGH18,QHWCYGH19,HQWCYGH20,XMZLP20}}, along with countermeasures.   
Less attention has so far been paid to attacks on practical implementations
of quantum protocols for mistrustful 
cryptographic tasks.   We consider such attacks here.  

In mistrustful cryptography, two or more parties collaborate to implement a cryptographic task without trusting each other. Significant mistrustful cryptographic 
tasks include bit commitment \cite{K99,K05.2,K11.2,K12,LKBHTKGWZ13,LCCLWCLLSLZZCPZCP14,LKBHTWZ15,CCL15,AK15.1,AK15.2,VMHBBZ16}, various types of oblivious transfer  \cite{R81,EGL85,C88,PG15.1,PGK18,PG19,ASMMDWA20}, coin flipping \cite{B83,K99.2,M07NATURE,BBBGST11,PJLCLTKD14,ACGKM16,BCKD20,PG21}, secure computations \cite{Y82}, position authentication \cite{KMS11,pbqc,K11.1}, digital signature schemes \cite{WDKA15,AWKA16} and money
or secure unforgeable tokens \cite{wiesner1983conjugate,G12,PYLLC12,GK15,BOVZKD18,KSmoney,BDG19,KPG20,KLPGR21}. Some of these tasks cannot be implemented with unconditional security using standard quantum cryptography \cite{M97,LC97,LC98,Kitaev02,ABDR04}, but
can be when relativistic signalling constraints are taken into
account \cite{K99,K99.2,K05.2,K11.2,K12,LKBHTKGWZ13,LCCLWCLLSLZZCPZCP14,LKBHTWZ15,CCL15,AK15.1,AK15.2,VMHBBZ16,PG21}. For others, even relativistic quantum protocols cannot
provide unconditional security \cite{L97,R02NATURE,CK06,C07,pbqc,BCS12}, but there may be  protocols that are provably unbreakable when realistic practical assumptions are made on the amounts of entanglement \cite{KMS11,pbqc}, quantum memory \cite{DFSS08,WST08,NJMKW12,ENGLWW14}, or other technologically challenging resources. Some tasks are only properly defined in a relativistic setting \cite{KMS11,pbqc,K11.3,K12.1,K13,PG15.1,HM16,AK15.3,PGK18,PG19}.

At present, quantum implementations of mistrustful cryptography
generally use photonic systems. 
A common step in these protocols is for (say) Alice to send
Bob quantum states encoding some secret data of hers
in photonic degrees of freedom (usually polarization), and for Bob
to apply quantum measurements on the received states,
chosen from a predetermined set, where his choice encodes
secret data of his.  
Idealized protocols often assume that Alice has ideal 
single photon sources, the channel is lossless, and Bob has perfectly efficient
ideal detectors. In practice, there are preparation and measurement errors and 
losses.   Moreover, most implementations use weak coherent pulses
rather than near-perfect single photon sources. 
Another issue is that Bob generally uses threshold 
photon detectors, which cannot
distinguish the number of photons in pulses activating a
detection and which are not close to perfectly efficient.
Because of losses and imperfect detectors, 
realistic implementations tend to require Bob to report during the protocol the labels of pulses 
activating a detection.

Realistic security analyses need to take all these points into account.
The full range of attacks that they allow in mistrustful quantum cryptographic scenarios seems not to have been appreciated. In such scenarios, Bob cannot assume that Alice sends independent
light pulses with identically randomly distributed
photon numbers, or the pulses have similar
frequencies, or any variations are statistical fluctuations
that Alice has no more information about than he does. 
If Alice can advantageously vary the distributions in 
a controlled way or obtain information about individual 
pulses, then she might.   An unconditionally secure real world
implementation must allow for these possibilities and
still provide security guarantees. Bob also cannot assume
that his detectors have precisely equal efficiencies, nor that
Alice has no information about their efficiencies. 
Even if the latter were true at the start of a protocol,
Alice can learn information about their efficiencies during
the protocol. In this respect she has advantages over
Bob, since she knows the states sent, while he does not, and
can send states other than those prescribed by the protocol.

Here we analyse multiphoton attacks, in which Alice controls the number of photons in the transmitted pulses and uses Bob's message reporting the successful measurements to obtain information about his measurement bases. A related attack in quantum key distribution (QKD) comprises Eve sending multiphoton pulses to Bob to obtain information about the key generated by Alice and Bob. Squashing models \cite{BML08}\dpggg{, analyses of double click rates \cite{ZCWLL21,T20},} and measurement-device-independent protocols \cite{LCQ12} have been proposed as countermeasures against  these attacks in QKD. However, security analyses in mistrustful quantum cryptography are not in general equivalent to QKD security analyses. An important reason for this is that Alice and Bob trust each other in QKD, while in mistrustful quantum cryptography they do not.

We analyse various strategies of Bob to report the labels of pulses
activating a detection, some of which were considered before, and discuss their vulnerability to multiphoton attacks. The analysed reporting strategies fit within a broad class of probabilistic reporting strategies introduced here in which Bob reports the label of a pulse with a probability that depends on which of his detectors are activated. Our main result (Theorem \ref{lemma3}) states that the only reporting strategy within this class that provides perfect protection against arbitrary multiphoton attacks when Bob's detection efficiencies are different in standard setups with single photon threshold detectors is a trivial strategy in which Bob reports all detection events with the same probability. This implies that the strategy of symmetrization of losses \cite{NJMKW12}, which is commonly used (e.g. \cite{NJMKW12,LKBHTKGWZ13,ENGLWW14,PJLCLTKD14}), does not protect against arbitrary multiphoton attacks. We discuss how multiphoton attacks apply to the experimental demonstrations of mistrustful quantum cryptography of Refs. \cite{NJMKW12,LKBHTKGWZ13,LCCLWCLLSLZZCPZCP14,ENGLWW14,PJLCLTKD14}. We report an experiment suggesting that multiphoton attacks can be implemented in practice. We also discuss side-channel attacks in mistrustful quantum cryptography, where Alice controls degrees of freedom not previously agreed with Bob. %We propose some countermeasures to investigate against multiphoton and side-channel attacks.
We discuss possible countermeasures against multiphoton and side-channel attacks, including the use of photon-number-resolving detectors, measurement device independent protocols, fully device independent protocols, the use of teleportation to filter received pulses, and using near-perfect sources and near-ideal detectors.   All of these options either fail to guarantee security or are practically challenging; our analyses suggest that the last is the most promising option at
present and that teleportation could provide a good solution in the future.

%\section{Results}

\section{Private measurement of an unknown qubit state}
\label{PMQS}
Many interesting protocols in mistrustful quantum cryptography (e.g. \cite{K12,NJMKW12,LKBHTKGWZ13,LCCLWCLLSLZZCPZCP14,ENGLWW14,PJLCLTKD14,PGK18,KSmoney,PG19,KPG20,KLPGR21}) use
  some version of a task we call \emph{private measurement of an
    unknown qubit state}. An
  \emph{ideal protocol} to implement this task is the following:
\begin{enumerate}
\item Alice prepares a qubit state $\lvert \psi\rangle$ randomly from a set $\mathcal{S}=\{\lvert \psi_{ij}\rangle\}_{(i,j)\in\{0,1\}^2}$ and sends it to Bob.
\item Bob generates a random bit $\beta\in\{0,1\}$ privately and measures $\lvert\psi\rangle$ in a qubit orthogonal  basis $\mathcal{B}_{\beta}=\{\lvert\psi_{0\beta}\rangle,\lvert\psi_{1\beta}\rangle\}$.
\end{enumerate}
Commonly,
  $\mathcal{B}_{0}$ and $\mathcal{B}_{1}$ are the computational and
  Hadamard bases, respectively, and $\mathcal{S}$ is the set
of Bennett-Brassard 1984 \cite{BB84} (BB84) states.  We consider here the more
  general situation in which $\mathcal{B}_{0}$ and $\mathcal{B}_{1}$
  are arbitrary distinct qubit orthogonal bases. We assume that Alice and Bob know these bases precisely.

We are primarily interested in the security attainable against
Alice in various realistic implementations of the task.
This is parametrised by Alice's probability $P_{\text{guess}}$
to guess Bob's chosen bit $\beta$, assuming Bob honestly follows
the version of the protocol defined for the given implementation.
We say the protocol is $\epsilon_{\text{guess}}$-secure against Alice if
\begin{equation}
\label{security}
P_{\text{guess}}\leq \frac{1}{2}+\epsilon_{\text{guess}},
\end{equation}
for any possible strategy (not necessarily honestly following
the protocol) of Alice. 
We say it is {\it secure} against Alice if $\epsilon_{\text{guess}}
  \rightarrow 0$ as some protocol security parameter is increased. 
In general, a dishonest Alice may deviate in any way from 
the protocol.  For example, she may send Bob quantum states
that are not only outside the agreed set $\mathcal{S}$ but 
outside its Hilbert space.

Different experimental setups correspond to different
  protocols to implement versions of this task. Here we consider setups and
  protocols with photonic systems, where Alice encodes quantum states
  in degrees of freedom of photons, for example in polarization, and
  Bob measures quantum states using single photon detectors.

We focus here on attacks by Alice, and assume that Bob
honestly follows the given protocols.   However, we have in mind 
applications in which these are sub-protocols for mistrustful cryptographic
tasks in which 
cheating by Bob is equally relevant.   
These applications motivate two \emph{correctness} criteria:

\begin{enumerate}
\item If Alice and Bob follow the protocol, the pulse sent by Alice must produce a measurement outcome with probability $P_{\text{det}}$ satisfying
\begin{equation}
\label{new13}
P_{\text{det}}\geq \delta_{\text{det}},
\end{equation}
for some $\delta_{\text{det}}>0$ predetermined by Alice and Bob. 
\item If Alice and Bob follow the protocol, Bob measures
  the received qubit in the basis of Alice's prepared state, and 
Bob gets a measurement outcome, then the outcome is the state $\lvert \psi\rangle$ prepared by Alice,
  with probability $1-P_{\text{error}}$, where
\begin{equation}
\label{new10}
P_{\text{error}}\leq \delta_{\text{error}},
\end{equation}
for some $\delta_{\text{error}} \geq 0$ predetermined by Alice and Bob.
\end{enumerate}

In an \emph{ideal setup}, Alice's and Bob's laboratories are
  perfectly secure, their preparation and measurement devices are
  perfect, Alice has a perfect single photon source, the probability
  that a transmitted photon is lost in the quantum channel is zero,
  Bob has a perfect random number generator and single photon
  detectors with unit detection efficiency and without dark counts. 
  
Since Bob's detectors are ideal and the quantum channel is lossless, Bob obtains a measurement outcome
with unit probability if Alice sends a single photon: precisely one of
his detectors clicks for each photon sent. 
Thus, $P_{\text{det}}=1$ and condition (\ref{new13}) is trivially satisfied.
Since the preparation and measurement devices are perfect, 
$P_\text{error}=0$ and (\ref{new10}) is also trivially satisfied.

In the setups we consider Bob has at least two detectors. If Alice sends something other than a single photon state,
Bob may get zero, two or more clicks.  The number of clicks
may depend statistically on his measurement basis.   However, since
he does not report anomalous results, his laboratory is secure, 
and his basis choices are perfectly random, Alice still learns no
information about his chosen basis.   
Thus, the ideal setup allows us to effectively implement
the task of private measurement of an unknown qubit state, with
perfect security against Alice ($\epsilon_\text{guess}=0$) and perfect correctness ($P_\text{det}=1$ and $P_\text{error}=0$).

However, in practical implementations,  
the ideal protocol and setup have to be altered to allow for imperfect
sources, channels and measuring devices. We show in this paper that this allows attacks by Alice in practical
implementations of the task using the standard setups described below with single
photon threshold detectors.
%We show in this paper that this allows attacks by Alice \eee{that make perfect security impossible} in practical implementations of the task using known technology.   This is hence also true of higher-level cryptographic tasks that use the considered task as a sub-routine.

\section{Practical implementations}

\subsection{A practical protocol}
\label{APP}
 
We  make the standard cryptographic assumption
that Alice's and Bob's laboratories are secure. 
We consider the realistic case in which their preparation and
measurement devices are imperfect, the 
quantum channel is lossy, Bob's
single photon detectors have nonunit efficiencies and
nonzero dark count probabilities. We assume that Bob's random number generator is perfect, as the attacks we discuss do not depend on Bob having an imperfect random number generator.

We assume Alice has what is usually called a single photon source, although in fact only approximates to one, i.e. the source emits a pulse with $k$ photons with
probability $p_k$, for $k\in\{0,1,2,\ldots\}$, where
$p_1>>\sum_{k=2}^\infty p_k$. 
For the moment we suppose that Alice and Bob know the probabilities
$p_k$ but neither of them has any further information
about the number of photons $k$ in any given pulse. 
Most commonly, in implementations to date, Alice uses a weak coherent
source (e.g. \cite{LKBHTKGWZ13,LCCLWCLLSLZZCPZCP14,PJLCLTKD14,BOVZKD18}), which emits a pulse of $k$ photons with probability
$p_k=e^{-\mu}\mu^k/k!$, where $\mu$ is the average photon number,
chosen by Alice and agreed with Bob, with
$0<\mu<<1$.   This emits an empty pulse
with probability $p_0=e^{-\mu}\approx 1$, 
but nonempty pulses are likely to be single photons since
$p_1>>\sum_{k=2}^\infty p_k$. 
Another possibility is a source of pairs of entangled photons in
which Alice measures one of the photons and sends the other one to
Bob, with transmissions considered valid if Alice obtains a
measurement outcome (e.g. \cite{NJMKW12,ENGLWW14}).  This gives $p_0<<1$, 
$p_1\approx 1$ and so $p_1>>\sum_{k=2}^\infty p_k$. We note that quantum-dot single-photon
 sources also satisfy $p_1\approx 1$ \cite{Michlerbook}, with $p_2$ reaching a few percent of $p_1$
in practice \cite{SSW17}.

In either case there is a nonzero probability that none of Bob's detectors
  click when a nonempty pulse is sent, because they are not perfectly
efficient and the quantum channel is lossy. 
There is also a nonzero probability that more than one of Bob's
detectors click, because their dark counts are nonzero and the
transmitted pulse may have more than one photon.
As we explain below, a realistic protocol generally requires an algorithm 
determining a message $m$ as a function (which may be probabilistic)
of Bob's detection results. This algorithm may in general
  depend on various experimental parameters known to Bob, which may include his detectors' efficiencies, for instance. Bob sends Alice $m=1$ ($m=0$) 
to report a successful (unsuccessful) measurement. 
As usual in cryptography, we assume the full protocol, and 
hence the algorithm, is known to both parties.  
It defines Bob's \emph{reporting strategy}.

The parameter $\delta_\text{det}$ must be sufficiently small to satisfy
(\ref{new13}). Bob's detection efficiencies are often small. For example, 
  Refs. \cite{NJMKW12,LKBHTKGWZ13,ENGLWW14,PJLCLTKD14}
  report detection efficiencies of the order of  $0.06$,
  $0.08$, $0.13$ and $0.015$, respectively. 
Ref. \cite{LCCLWCLLSLZZCPZCP14} reports considerably higher detection efficiencies, of the order of $0.45$. Note that in this paper, we include in the term `detection efficiency' the transmission efficiency of the quantum channel and the quantum efficiency of the detectors.

However, more complex protocols that have
our task as a subroutine may also require that
$\delta_{\text{det}}$ be not too small in order to give security
against Bob. For example, a $N-$parallel
repetition of the task in which Bob reports to Alice that $\ll N
\delta_{\text{det}}$ pulses produce a measurement outcome could
allow Bob to choose to report an
appropriate subset of pulses producing measurement outcomes
advantageous to him in some cheating strategy.
Alice thus stipulates a minimum value of $\delta_{\text{det}}$, and
Bob must ensure his technology allows this value to be attained. 
Alice aborts if Bob reports less than
$N\delta_{\text{det}}(1-\delta_{\text{det}}^{\text{dev}})$,
where $\delta_{\text{det}}^{\text{dev}}>0$
is predetermined by Alice and agreed by Bob as the maximum
tolerable deviation from the expected value of reported pulses. 

Since the preparation and measurement devices are not perfect, there
is a probability $P_\text{error}>0$ that Bob obtains a measurement
outcome different to Alice's target state $\lvert\psi\rangle$ when Bob
measures in the basis of preparation by Alice. Thus,
$\delta_\text{error}$ must be chosen large enough to
guarantee the condition (\ref{new10}).

However, more complex protocols that have the task as a subroutine
generally require Bob to be able to report correct measurement
outcomes with reasonably high probability. 
This requires $\delta_\text{error}$ not to be too large.  
 
These constraints require that Bob identifies some subset of purportedly successful
measurement outcomes to Alice, in which the proportion of actually
successful measurement outcomes will be relatively high if
Alice honestly followed the protocol. 
They motivate a \emph{practical protocol} for private measurement of
an unknown qubit state, with the above practical setup:

\begin{enumerate}
\item Alice prepares and sends Bob a photon pulse with an approximate
  single photon source, where each photon in the pulse encodes the
  same qubit state $\lvert \psi\rangle$, and where
  $\lvert \psi\rangle$ is chosen randomly by Alice from the set
  $\mathcal{S}$.
\item Bob generates a random bit $\beta\in\{0,1\}$ and measures the pulse in the qubit basis $\mathcal{B}_{\beta}$.
\item Bob sends a message $m\in\{0,1\}$ to Alice reporting whether a
  measurement outcome was produced ($m=1$) or not ($m=0$),
following an agreed reporting strategy.  
\end{enumerate}

Unless otherwise stated, in the definition of security
  against Alice for the practical protocol, $P_\text{guess}$ is taken
  as Alice's probability to guess $\beta$, independently of the value
  of $m$. This makes sense with the setup described below in
  extensions of this protocol in which $N>1$ photon pulses are
  produced by Alice and all pulses are measured by Bob in the same
  basis $\mathcal{B}_\beta$. Examples of these type of protocols are
  those of Refs. \cite{LKBHTKGWZ13,LCCLWCLLSLZZCPZCP14}, which we
  discuss in Appendix \ref{appD}.

  Alternatively, we could define $P_\text{guess}$ as Alice's
  probability to guess $\beta$ conditioned on Bob reporting $m=1$.
  This makes more sense in extensions of the practical protocol in
  which Alice sends Bob $N>1$ photon pulses and Bob measures each
  pulse randomly in one of the two bases, $\mathcal{B}_0$ and
  $\mathcal{B}_1$, or in protocols using a setup with four single
  photon detectors (setup II defined in Appendix \ref{NJMKW12})
  instead of two as in the setup below. This is because in these scenarios Alice might only care to learn Bob's measurement bases for the pulses that he reported as being successful. In Appendix \ref{NJMKW12}, we use this definition to present attacks to the protocols of Refs. \cite{NJMKW12,ENGLWW14}.

Some version of the practical protocol defined above is
commonly used in experimental demonstrations of mistrustful
cryptography \cite{NJMKW12,LKBHTKGWZ13,LCCLWCLLSLZZCPZCP14,ENGLWW14,PJLCLTKD14}. Various reporting strategies have been
used.   We show below that they have subtle weaknesses, allowing
attacks by a dishonest Alice. We show moreover that this is 
true of {\it any} reporting strategy of the broad class described here.  

\subsection{Details of setup}

We consider a basic setup to implement the practical
  protocol, denoted as \emph{setup I} (see Fig. \ref{fig1}). Another
setup is discussed in Appendix \ref{NJMKW12}. Alice encodes a
qubit state in some degrees of freedom of a photon pulse, using an
approximate single photon source and a state modulator. Alice and Bob
agree in advance on these degrees of freedom. These typically consist
in the polarization (e.g. \cite{NJMKW12,LKBHTKGWZ13,LCCLWCLLSLZZCPZCP14,ENGLWW14,PJLCLTKD14,BOVZKD18}), but can also be time bin (e.g. \cite{BBBGST11}) or others, for
example. For definiteness we focus on polarization coding in
this paper, but our results apply to any choice.

A \emph{dishonest Alice} may deviate arbitrarily from the
protocol, limited only by her technological capabilities. Importantly,
when proving the unconditional security of a protocol, we need to
assume that the technology available to dishonest Alice is only
limited by the laws of physics. Dishonest Alice may, for
  example, replace the photon source agreed with Bob for the protocol
  with another one that has different statistics. More generally, Alice may send Bob an arbitrary quantum
state $\rho$ encoded in the polarization and other degrees of freedom.
  
In most of this paper we focus on  \emph{multiphoton
  attacks} performed by a dishonest Alice, in which Alice sends Bob a pulse of
$k$ photons encoding an arbitrary quantum state $\rho$ of her choice,
which may be pure or mixed and may be entangled with an
ancilla held by Alice, and where each photon in the pulse encodes a qubit
in its polarization. 

We emphasize that the multiphoton attacks we consider here are different from the photon number splitting attacks previously considered in the quantum cryptographic literature \cite{HIGM95}. In the latter, a dishonest Bob exploits the fact that honest Alice's realistic photon source emits multiphoton pulses with non-zero probability and, in principle, perfectly learns Alice's prepared states for the multiphoton pulses. However, in this paper we consider multiphoton attacks by a dishonest Alice, who uses multiphoton pulses to gain information about Bob's measurement choices. For this reason, decoy state protocols \cite{WY03,WCSL10} and other countermeasures against photon number splitting attacks do not concern us here.

%We clarify that the multiphoton attacks we consider here do not include the photon number splitting attacks \cite{HIGM95}. In the latter, a dishonest Bob exploits the fact that honest Alice's realistic photon source emits multiphoton pulses with non-zero probability and, in principle, perfectly learns Alice's prepared states for the multiphoton pulses. However, in this paper we only consider attacks by a dishonest Alice. For this reason, decoy state protocols \cite{WY03,WCSL10} and other countermeasures against photon number splitting attacks do not concern us here.}

We assume from here on that Bob honestly follows the agreed protocol. 
In setup I (see Fig. \ref{fig1}), Bob measures the polarization of the received
pulse with a wave plate, set in one of two positions
$\beta\in\{0,1\}$, followed by a polarizing beam splitter and two
single photon detectors, D$_0$ and D$_1$.   These are threshold
detectors: they do not distinguish the number of photons of a
pulse producing a detection. If $\beta=0$ ($\beta=1$), Bob sets the
wave plate in its first (second) position, corresponding to a
measurement in the basis $\mathcal{B}_0$ ($\mathcal{B}_1$). Let
$\eta_{i\beta}$ and $d_{i\beta}$ be the detection efficiency and the
dark count probability of detector D$_i$ when Bob applies the
measurement $\mathcal{B}_\beta$, where $0<\eta_{i\beta}<1$ and
$0<d_{i\beta}<<1$, for $i,\beta\in\{0,1\}$. In our model, dark counts
and each photo-detection are independent random events. To the best of
our knowledge, this is a valid assumption.

\begin{figure}
\includegraphics[scale=0.59]{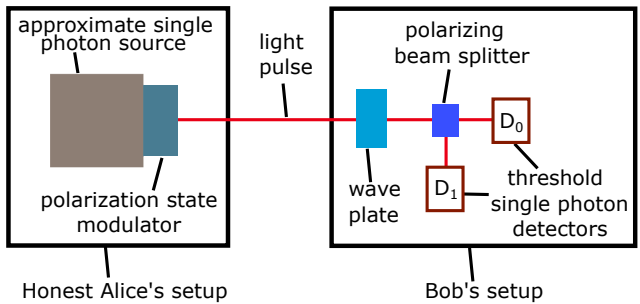}
\caption{\label{fig1} \textbf{Setup I.} Honest Alice's setup consists in an
  approximate single photon source and a polarization state modulator. The setup for Bob, who is assumed honest, comprises a wave plate set in one of two positions, according to
  $\beta\in\{0,1\}$, a polarizing beam splitter and two threshold
  single photon detectors D$_0$ and D$_1$.}
\end{figure}

The pair of detectors produces a \emph{detection event}
$(c_0,c_1)\in\{0,1\}^2$ with four possible values: $(c_0,c_1)=(1,0)$
if D$_0$ clicks and D$_1$ does not click, $(c_0,c_1)=(0,1)$ if D$_0$
does not click and D$_1$ clicks, $(c_0,c_1)=(0,0)$ if no detector
clicks, and $(c_0,c_1)=(1,1)$ if both detectors click.

Independently of whether Alice is honest or not, we define
$P_{\text{det}}(c_0,c_1\vert \beta,\rho ,k)$ and
$P_{\text{report}}(m\lvert \beta, \rho, k)$ to be the probabilities
that a detection event $(c_0,c_1)\in\{0,1\}^2$ occurs, and that Bob
reports to Alice the message $m\in\{0,1\}$ when Alice sends Bob a
pulse of $k$ photons encoding the state $\rho$ and Bob measures in the
basis $\mathcal{B}_\beta$, respectively, for
$c_0,c_1,m,\beta\in\{0,1\}$ and $k\in\{0,1,2,\ldots\}$. We define
$P_\text{det}(c_0,c_1\vert \beta)=\sum_{k=0}^\infty p_k
P_{\text{det}}(c_0,c_1\vert \beta,\rho ,k)$, for
$c_0,c_1,\beta\in\{0,1\}$, when Alice and Bob follow the protocol honestly.

\dpg{We note that although setup I is designed to work when the quantum states are encoded in the photons' polarization, straightforward variations can be implemented if the quantum states are encoded in other photonic degrees of freedom, particularly in the time-bin. For example, the experimental demonstration of quantum coin flipping of Ref. \cite{BBBGST11} encodes the quantum states in the time-bin of photons and, as in setup I, also requires Bob to use two threshold single photon detectors. In other setups, also encoding quantum states in the photons' time-bin, only one threshold single photon detector D, working at two different time intervals $\tau_0$ and $\tau_1$, could be required. Our analyses in this paper apply straightforwardly to these setups by identifying the detector D$_i$ of setup I with the detector D working in the time interval $\tau_i$, for $i\in\{0,1\}$.

%Moreover, although 
Setup I is very commonly used to implement the task of private measurement of an unknown qubit state considered here, used for example in Refs. \cite{LKBHTKGWZ13,LCCLWCLLSLZZCPZCP14,PJLCLTKD14}. Another common setup (setup II), in which Bob has four threshold single photon detectors and which is used in Refs. \cite{NJMKW12,ENGLWW14}, is discussed in detailed in Appendix \ref{NJMKW12}. Variations of setup I are also discussed in section \ref{exotic}.
In principle, other setups could be devised; it is beyond our scope here to consider multiphoton attacks on these, which would need to be analysed case by case.   
%, other setups could be used instead. 
%We do not attempt to analyse the task, and the multiphoton attacks arising, for all possible setups. However, another common setup (setup II), in which Bob has four threshold single photon detectors and which is used in Refs. \cite{NJMKW12,ENGLWW14}, is discussed in detailed in Appendix \ref{NJMKW12}. Variations of setup I are also discussed in section \ref{exotic}.
}

\section{Bob's reporting strategies and Alice's multiphoton attacks}

Tables \ref{tablemaintext} -- \ref{table3} summarize the multiphoton attacks, the main suggested countermeasures against them, and the reporting strategies  discussed below, as well as the application of multiphoton attacks to previous experimental demonstrations of mistrustful quantum cryptography.

\subsection{Reporting only single clicks}

\label{singleclicks}

We define \emph{reporting strategy I} by $m=1$ if
$(c_0,c_1)\in\{(1,0),(0,1)\}$, and $m=0$ otherwise. \dpgg{That is, \emph{Bob reports only single clicks}
 in the sense that he tells Alice that he obtained a measurement outcome, by sending her the message $m=1$, if exactly one of his two detectors click. This implies that if none or both of Bob's detectors click then he tells Alice that he did not obtain a measurement outcome, by sending Alice the message $m=0$.} This might seem a natural strategy. However, as Liu et
al. discuss \cite{LCCLWCLLSLZZCPZCP14}, if Bob 
uses reporting strategy I, Alice can gain information about $\beta$ with the following attack, which we call \emph{multiphoton attack I}.

To illustrate this attack, consider a setup in which
Alice's polarization preparation devices and Bob's polarizers are 
precisely aligned. Alice sends a
pulse with a large number of photons $k$ in the same polarization
state chosen from $\mathcal{S}$; for example, 
$\rho=(\lvert \psi_{00}\rangle\langle \psi_{00}\rvert)^{\otimes k}$.
If Bob measures the pulse in the basis $\mathcal{B}_0$ then the
detection event $(c_0,c_1)=(1,0)$ occurs with high probability,
giving $m=1$. If Bob measures in the basis $\mathcal{B}_1$
then the detection event $(c_0,c_1)=(1,1)$ occurs with high 
probability, giving $m=0$. Thus, given $m$, Alice can learn $\beta$
with high probability.

In a different version of attack I, Alice's pulse is
  prepared with a coherent source with average photon number
  $\mu>>1$. We show below (%see \hyperref[methods]{Methods}
  see section \ref{experiment}) that in this case, if $\mathcal{B}_0$ and
  $\mathcal{B}_1$ are the computational and Hadamard bases, Bob's detectors have equal efficiencies $\eta\in(0,1)$ and zero dark count probabilities, then Alice's probability
  $P_\text{guess}^\text{cs}(\mu)$ to guess Bob's bit $\beta$ as a function of $\mu$ is given by
\begin{equation}
\label{attackexmain}
P_\text{guess}^\text{cs}(\mu)=1-\frac{1}{2}\Bigl(2e^{-\frac{\mu\eta}{2}}-e^{-\mu\eta}\Bigr)
\, .
\end{equation}

When the devices are closely but not precisely aligned, Alice can still learn
significant information about $\beta$ from $m$ with an
appropriate choice of $k$ or $\mu$. An experimental
simulation of this attack is presented below (see
%\hyperref[methods]{Methods}
section \ref{experiment}). 
The results are given in Fig. \ref{guessingexp} and show that Alice's probability to guess Bob's bit $\beta$ is very well approximated by (\ref{attackexmain}), and can be very close to unity if $\mu$ is sufficiently large.
Different versions of attack I apply to the experimental protocols of
Refs. \cite{NJMKW12,LKBHTKGWZ13,ENGLWW14,PJLCLTKD14} if they are
implemented
with Bob using reporting strategy I (see
  Appendix \ref{appD}).
  
We note that multiphoton attack I still applies in a variation of reporting strategy I in which Bob also sets $m=1$ with some non-zero but small probability when $(c_0,c_1)=(1,1)$, i.e. when there is a double click.
  
\begin{figure}
\includegraphics[scale=0.63]{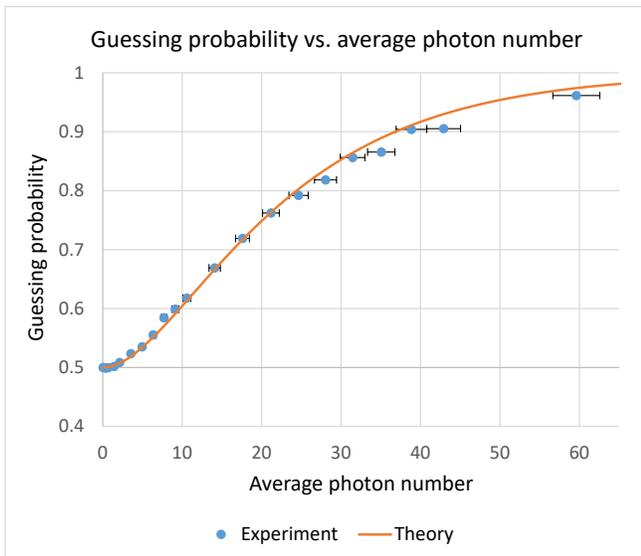}
\caption{\label{guessingexp} \textbf{Alice's guessing
      probability in multiphoton attack I.} Implementing setup I, we simulated multiphoton attack I by varying the attenuation of Alice's coherent source in order to scan the range for the average photon number $\mu$. Frequencies for zero, single and double clicks in the pair of Bob's detectors were registered in order to then compute the experimental estimation of Alice's
    guessing probability $P_\text{guess}^\text{cs}(\mu)$ (blue dots) when Bob uses reporting strategy I. We propose reporting strategy II (defined below) as a countermeasure against multiphoton attack I.  The orange solid curve
    represents the theoretical prediction given by
    (\ref{attackexmain}), assuming the values $d_{i\beta}=0$ and
    $\eta_{i\beta}=\eta\in(0,1)$, for $i,\beta\in\{0,1\}$. The
    measured value of the detection efficiency, including the
    transmission probability through the quantum channel, is
    $\eta=0.12$. The horizontal and vertical uncertainty bars are included. The vertical uncertainty bars are so small that they are unnoticeable and lie within the corresponding markers.}
\end{figure}

\subsection{Reporting if at least one detector clicks}
\label{reportdoublecounts}

A better reporting strategy is to set $m=1$ if at
least one detector clicks and $m=0$ if no detector clicks. We call this \emph{reporting strategy II}. 

This reporting strategy has been considered in quantum key distribution: squashing models map a multiphoton quantum state to a single-photon state by randomly assigning the measurement outcome of a double click to a single click \cite{BML08}. This reporting strategy has also been implemented in the experimental demonstrations of mistrustful quantum cryptography of Refs. \cite{BBBGST11,LCCLWCLLSLZZCPZCP14}.

As the following lemma shows, if Bob's detectors have exactly
equal efficiencies and their dark count probabilities are independent
of his measurement basis, Alice cannot learn any information about
$\beta$ from the message $m$. %See Appendix \ref{appE} for proofs of the lemmas and theorem.

\begin{lemma}
\label{lemma0}
Suppose that $\eta_{i\beta}=\eta$ and $d_{i\beta}=d_i$, for $i,\beta\in\{0,1\}$.
If Bob uses reporting strategy II,  then for an arbitrary quantum state $\rho$ encoded in a pulse of $k$ photons that Alice sends Bob, and for arbitrary qubit orthogonal bases $\mathcal{B}_0$ and $\mathcal{B}_1$, it holds that
\begin{eqnarray}
\label{a1}
P_{\text{report}}(1\lvert \beta, \rho, k)&=&1-(1-d_0)(1-d_1)(1-\eta)^k,
\end{eqnarray}
for $\beta\in\{0,1\}$ and $k\in\{0,1,2,\ldots\}$. 
\end{lemma}

\begin{proof}
By definition, if Bob follows reporting strategy II then
\begin{equation}
\label{wwwa1}
P_{\text{report}}(1\lvert \beta, \rho, k)=1-P_{\text{det}}(0,0\vert \beta,\rho, k),
\end{equation}
for $\beta\in\{0,1\}$ and $k\in\{0,1,2,\ldots\}$. Let $k_0$ be the number of photons that go to detector D$_0$. The number of photons that go to detector D$_1$ is $k_1=k-k_0$. 
The probability
$P(k_0\vert \beta, \rho)$ that $k_0$ photons go to $D_0$ and $k-k_0$ photons go to
$D_1$ depends both on the quantum state $\rho$ and on the quantum measurement $\mathcal{B}_\beta$ implemented by Bob, via the Born rule. 
As previously stated, in our model, the dark counts and each photo-detection are independent random events. We define $P_i(0\vert \rho, k_0)$ to be the probability 
that $D_i$ does not click, which is independent of whether $D_{\bar{i}}$
clicks, and which is independent of Bob's measurement basis $\mathcal{B}_\beta$, for $i,\beta\in\{0,1\}$.  
We have
\begin{eqnarray}
\label{newnew1}
P_0(0\vert \rho, k_0)&=& (1-d_0)(1-\eta_0)^{k_0},\nonumber\\
P_1(0\vert \rho,k_0)&=& (1-d_1)(1-\eta_1)^{k-k_0},
\end{eqnarray}
for $k_0\in\{0,1,\ldots,k\}$ and $k\in\{0,1,2,\ldots\}$. Thus, for a $k-$photon pulse, the probability
$P_{\text{det}}(0,0\vert \beta,\rho, k)$ that no detector clicks when Bob measures in the basis $\mathcal{B}_\beta$ is given by
\begin{eqnarray}
\label{z1.1}
P_{\text{det}}(0,0\vert \beta,\rho, k)&=&\!\sum_{k_0=0}^k\!\! P(k_0\vert \beta, \rho) P_0(0\vert \rho,k_0)P_1(0\vert \rho,k_0)\nonumber\\
&=&\!\sum_{k_0=0}^k \!\!P(k_0\vert \beta, \rho) (1-d_0)(1-d_1)\times\nonumber\\
&&\qquad\quad\times(1-\eta_0)^{k_0}(1-\eta_1)^{k-k_0}\nonumber\\
&=&\!\sum_{k_0=0}^k\!\! P(k_0\vert \beta,\rho) (1-d_0)(1-d_1)(1-\eta)^k\nonumber\\
&=&\!(1-d_0)(1-d_1)(1-\eta)^k,
\end{eqnarray}
for $\beta\in\{0,1\}$, where in the second line we used (\ref{newnew1}), in the third line we used the assumption that
$\eta_0=\eta_1=\eta$, and in last line we used
$\sum_{k_0=0}^k P(k_0\vert \beta, \rho) =1$, for $\beta\in\{0,1\}$. Thus, (\ref{a1}) follows from (\ref{wwwa1}) and (\ref{z1.1}).
\end{proof}

In practice, Bob cannot guarantee that the efficiencies of his
detectors are exactly equal. In this general case, the detection
probabilities $P_{\text{det}}(c_0,c_1\vert \beta,\rho,k)$ depend
nontrivially on $\beta$, as illustrated in Lemma \ref{lemma2} (see
%\hyperref[methods]{Methods}
section \ref{methods}). We call \emph{multiphoton attack II} any strategy implemented by Alice that allows her to exploit the difference of Bob's detection efficiencies to obtain information about $\beta$, when Bob reports double clicks with unit probability, as in reporting strategy II, or with high probability. Particularly, in this attack, we assume that Alice knows the efficiencies of
Bob's detectors and can control the number of photons and the states
of her pulses. For example, in an extension of the task to $N>>1$ pulses in which Bob measures all pulses in the basis $\mathcal{B}_0$ or in the basis $\mathcal{B}_1$, Alice can prepare a subset of pulses in specific states and with specific number of photons to maximize her probability to guess $\beta$ from Bob's messages reporting whether these pulses produced successful measurements or not. This and other versions of this attack apply (for example)
to the experimental demonstrations of
Refs. \cite{LKBHTKGWZ13,LCCLWCLLSLZZCPZCP14}.

In Appendix \ref{appD}, we discuss how
  multiphoton attacks I and II apply to some experimental
  demonstrations of mistrustful quantum cryptography. We summarize
  this in Table \ref{tablemaintext} below.

\begin{table*}
\centering
\begin{center}
\begin{tabular}{ | c| c | c | c | c | c |}
\hline
\multirow{2}{*}{\textbf{Reference}} & \multirow{2}{*}{\textbf{Task}}  & \multirow{2}{*}{\textbf{Setup}} & \textbf{Reports multiple} &\textbf{Covered against} &\textbf{Covered against}\\
& & &\textbf{clicks?} &\textbf{attack I?} & \textbf{attack II?}\\
\hline
\multirow{2}{*}{\cite{LCCLWCLLSLZZCPZCP14}} & Relativistic quantum  & \multirow{2}{*}{I} & \multirow{2}{*}{Yes} &\multirow{2}{*}{Yes} &\multirow{2}{*}{No}\\
& bit commitment& & & &\\
\hline
\multirow{2}{*}{\cite{LKBHTKGWZ13}} & Relativistic quantum  & \multirow{2}{*}{I} & \multirow{2}{*}{Does not say} &Yes (no) if double clicks &\multirow{2}{*}{No}\\
& bit commitment& & &are (are not) reported &\\
   \hline
   \multirow{2}{*}{\cite{PJLCLTKD14}} & \multirow{2}{*}{Quantum coin flipping}  & \multirow{2}{*}{I} & \multirow{2}{*}{Does not say} &%No, even if double 
   Yes (no) if double clicks &\multirow{2}{*}{No}\\
& & & &%clicks are reported
are (are not) reported&\\
\hline
   \multirow{2}{*}{\cite{ENGLWW14}} & Quantum 1-out-of-2 oblivious transfer  & \multirow{2}{*}{II} & \multirow{2}{*}{Does not say} &Yes (no) if multiple clicks &\multirow{2}{*}{No}\\
& in the noisy storage model & & &are (are not) reported &\\
\hline
   \multirow{2}{*}{\cite{NJMKW12}} & Quantum bit commitment  & \multirow{2}{*}{II} & \multirow{2}{*}{No} &\multirow{2}{*}{No} &\multirow{2}{*}{No}\\
& in the noisy storage model & & & &\\
\hline
\end{tabular}
\caption{\textbf{Application of multiphoton attacks I and II to
    experimental demonstrations of mistrustful quantum cryptography.}
  The multiphoton attacks I and II are briefly described in the main
  text. The table indicates whether the protocols of
  Refs. \cite{NJMKW12,LKBHTKGWZ13,LCCLWCLLSLZZCPZCP14,ENGLWW14,PJLCLTKD14}
  are covered against variations of these attacks, which are discussed
  in detail in Appendix \ref{appD}. Setup I is described in Fig. \ref{fig1},
  while setup II is a variation of Setup I with four detectors, which
  is described in Appendix \ref{NJMKW12}. }
\label{tablemaintext}
\end{center}
\end{table*}

%\dddd{DAMIAN'S COMMENT:For the last two "no", the justification is Lemma 13. For the first "no", the justification is that by definition, multiphoton attack I applies when Bob does not report double clicks, and multiphoton attack II applies when Bob reports double clicks. By reporting double clicks, the analyzed attack for the coin flipping protocol still applies, although with smaller probability   (but still with probability close to $7/8$). So, there is no new analysis for the first "no", only realizing the observation above (red text in page 31).}

We show below that there is no reporting strategy using setup I that
guarantees perfect security against Alice when the detection
efficiencies are different\dpgg{, apart from the trivial reporting strategy in which Bob does not report any losses to Alice}. However, we derive in Appendix \ref{appB} an upper
bound on the amount of information that Alice can learn about $\beta$
from Bob's message $m$, which approaches zero when the difference of
the detection efficiencies tends to zero.

\subsection{Symmetrization of losses}

The \emph{symmetrization of losses} strategy was introduced in
  Ref. \cite{NJMKW12} as follows. Bob tests his setup by preparing
  and measuring states as in the practical protocol, a large number of times $N$ in parallel. Then, for $c,\beta\in\{0,1\}$, Bob computes the frequency $F_{\text{det}}(c,\bar{c}\vert \beta)$ of detection events $(c,\bar{c})$, which provides a good estimate of the probability $P_{\text{det}}(c,\bar{c}\vert \beta)$ if $NF_{\text{det}}(c,\bar{c}\vert \beta)>>1$. Bob then computes the numbers $S_{c\bar{c}\beta}\in(0,1]$ satisfying
\begin{equation}
\label{new1}
S_{c\bar{c}\beta}F_{\text{det}}(c,\bar{c}\vert \beta)= F_{\text{det}}^{\text{min}},
\end{equation}
for $c,\beta\in\{0,1\}$,
where 
\begin{equation}
\label{new1.0}
F_{\text{det}}^{\text{min}}=\min_{c,\beta\in\{0,1\}}\{F_{\text{det}}(c,\bar{c}\vert
\beta)\} \, .
\end{equation}

Then, in the implementation of the protocol with
Alice, 
Bob reports a detection event $(c,\bar{c})$ with probability
$S_{c\bar{c}\beta}$, for $c,\beta\in\{0,1\}$. 
Ref. \cite{NJMKW12} explicitly states that the detection events $(0,0)$ and $(1,1)$ are not reported by Bob, for $\beta\in\{0,1\}$. 

Symmetrization of losses aims to effectively make the detection probabilities of Bob's detectors equal. As shown below in Lemma \ref{lemma1}, a more precise definition of this strategy (reporting strategy III), achieves this if Alice's pulse has zero or one photons. However, dishonest Alice can send Bob a pulse with an arbitrary number of photons $k$. Because the detection probabilities are not linear functions of $k$ (see %\hyperref[methods]{Methods}
section \ref{methods}) and Bob does not know $k$, symmetrization of losses fails in effectively making the detection probabilities of Bob's detectors equal. This is proved by Theorem \ref{lemma3} below.

Symmetrization of losses has been implemented in at least four experimental
demonstrations of mistrustful quantum cryptography protocols
\cite{NJMKW12,LKBHTKGWZ13,PJLCLTKD14,ENGLWW14}.  
These implementations used setups and protocols that are slight
variations of ours.  We discuss them, 
and show that they are vulnerable to multiphoton attacks by Alice, in
Appendix \ref{appD}.

\subsection{Generalization of symmetrization of losses to
    double click events}
    
In this paper, we introduce the following generalization of symmetrization of losses and call it \emph{reporting strategy III}. Let
$\eta_{\text{min}}=\min\{\eta_{00},\eta_{01},\eta_{10},\eta_{11}\}$. If Bob obtains a detection event $(c_0,c_1)$, he sets $m=1$ with probability $S_{c_0c_1\beta}$, for $c_0,c_1,\beta\in\{0,1\}$, where
\begin{eqnarray}
\label{a5}
S_{00\beta}&=&0,\nonumber\\
S_{11\beta}&\in&[0,1],\nonumber\\
S_{01\beta}&=&\frac{\eta_{\text{min}}}{\eta_{1\beta}},\nonumber\\
S_{10\beta}&=&\frac{\eta_{\text{min}}}{\eta_{0\beta}},
\end{eqnarray}
for $\beta\in\{0,1\}$. Note that Bob needs to know the
efficiencies of his detectors to apply this reporting strategy
and that the choices of $S_{110}$ and $S_{111}$ are left free.

As the following lemma shows, this reporting strategy guarantees that Alice cannot
obtain any information about $\beta$ if Alice's pulse does not have
more than one photon and $d_{i\beta}=0$, for arbitrary
$\eta_{i\beta}\in(0,1)$ and for $i,\beta\in\{0,1\}$. Furthermore, it
guarantees that Alice cannot obtain much information about $\beta$ if
Alice's pulse does not have more than one photon and
$0<d_{i\beta}\leq \delta$, for $0<\delta<< 1$ and
$i,\beta\in\{0,1\}$.

\begin{lemma}
\label{lemma1}
Let
$d_{i\beta}\leq \delta$, for some $0\leq \delta<1$ and for
$i,\beta\in\{0,1\}$. Consider the practical protocol with
  $p_0+p_1=1$ in which Alice sends Bob a single photon pulse ($k=1$)
in arbitrary qubit state $\rho$ or an empty pulse ($k=0$).  Suppose that Bob applies
reporting strategy III. Then
\begin{equation}
\label{a3}
\bigl\lvert P_{\text{report}}(1\lvert 1, \rho, k) -P_{\text{report}}(1\lvert 0, \rho, k)\bigr\rvert \leq B_\text{III}^k,
\end{equation}
for $k\in\{0,1\}$, where
\begin{eqnarray}
\label{newa3}
B_\text{III}^0&=&2\delta,\nonumber\\
B_\text{III}^1&=&6\delta+2\delta^2+S_{11}^{\text{max}}(5\delta
+\delta^2),
\end{eqnarray}
and where $S_{11}^{\text{max}}=\max\{S_{110},S_{111}\}$.
\end{lemma}

\begin{proof}
If Alice sends Bob a pulse of $k$ photons encoding a state $\rho$ and Bob measures in the basis $\mathcal{B}_\beta$, the probability that Bob reports the message $m=1$ to Alice is given by
\begin{equation}
\label{newa4}
P_{\text{report}}(1\lvert \beta, \rho, k)=\sum_{c_0=0}^1\sum_{c_1=0}^1S_{c_0c_1\beta}P_{\text{det}}(c_0,c_1\lvert \beta, \rho, k),
\end{equation}
for $\beta\in\{0,1\}$ and $k\in\{0,1,2,\ldots\}$. From (\ref{a5}) and (\ref{newa4}), we have
\begin{eqnarray}
\label{a6}
&&P_{\text{report}}(1\lvert \beta, \rho, k) \nonumber\\
&&\qquad = \eta_{\text{min}}\biggl(\frac{P_{\text{det}}(0,1\vert \beta, \rho, k )}{\eta_{1\beta}}+\frac{P_{\text{det}}(1,0\vert \beta, \rho, k )}{\eta_{0\beta}}\biggr)\nonumber\\
&&\qquad\qquad\qquad\qquad+S_{11\beta}P_\text{det}(1,1\vert \beta,\rho,k),
\end{eqnarray}
for $\beta\in\{0,1\}$ and $k\in\{0,1,2,\ldots\}$.

We consider the case $k=0$. From (\ref{a6}) and from Lemma \ref{lemma2}, we have
\begin{eqnarray}
\label{a7}
P_{\text{report}}(1\lvert \beta, \rho, 0) &=&\eta_{\text{min}}\biggl(\frac{(1-d_{0\beta})d_{1\beta}}{\eta_{1\beta}}+\frac{(1-d_{1\beta})d_{0\beta}}{\eta_{0\beta}}\biggr)\nonumber\\
&&\qquad\qquad+S_{11\beta}d_{0\beta}d_{1\beta},
\end{eqnarray}
for $\beta\in\{0,1\}$. Since $0<\eta_{\text{min}}\leq \eta_{i\beta}$, $0\leq d_{i\beta}\leq \delta<1$ and $S_{11\beta}\leq 1$, for $i,\beta\in\{0,1\}$, it follows straightforwardly from (\ref{a7}) that
\begin{equation}
\label{a8}
\bigl\lvert P_{\text{report}}(1\lvert 1, \rho, 0) - P_{\text{report}}(1\lvert 0, \rho, 0)\bigr\rvert \leq 2\delta.
\end{equation}

We consider the case $k=1$. From (\ref{a6}) and from Lemma \ref{lemma2}, \dpgg{i.e. substituting (\ref{xab2}) with $k=1$ in (\ref{a6}) and after simple algebra,} we have
\begin{eqnarray}
\label{a9}
&&P_{\text{report}}(1\lvert \beta, \rho, 1) \nonumber\\
&&\quad=\frac{\eta_{\text{min}}}{\eta_{1\beta}}\bigl[\eta_{1\beta}(1-q_\beta)-d_{0\beta}(1-q_\beta\eta_{0\beta})\nonumber\\
&&\quad\qquad+(d_{0\beta}+d_{1\beta}-d_{0\beta}d_{1\beta})\bigl(q_\beta(1-\eta_{0\beta})\nonumber\\
&&\quad\qquad+(1-q_\beta)(1-\eta_{1\beta})\bigr)\bigr]\nonumber\\
&&\qquad+\frac{\eta_{\text{min}}}{\eta_{0\beta}}\bigl[\eta_{0\beta}q_\beta-d_{1\beta}(1-(1-q_\beta)\eta_{1\beta})\nonumber\\
&&\qquad\qquad+(d_{0\beta}+d_{1\beta}-d_{0\beta}d_{1\beta})\bigl(q_\beta(1-\eta_{0\beta})\nonumber\\
&&\qquad\qquad+(1-q_\beta)(1-\eta_{1\beta})\bigr)\bigr]\nonumber\\
&&\qquad+S_{11\beta}\bigl[1+(1-d_{0\beta})(1-d_{1\beta})\bigl(q_\beta(1-\eta_{0\beta})\nonumber\\
&&\qquad\quad+(1-q_\beta)(1-\eta_{1\beta})\bigr)\nonumber\\
&&\qquad\quad-(1-d_{0\beta})(1-q_\beta\eta_{0\beta})\nonumber\\
&&\qquad\qquad\qquad-(1-d_{1\beta})\bigl(1-(1-q_\beta)\eta_{1\beta}\bigr)\bigr],
\end{eqnarray}
for $\beta\in\{0,1\}$. Since $0<\eta_{\text{min}}\leq \eta_{i\beta}$, $0\leq d_{i\beta}\leq \delta<1$ and $0\leq S_{11\beta}\leq S_{11}^\text{max}$, for $i,\beta\in\{0,1\}$, it follows straightforwardly from (\ref{a9}) that
\begin{eqnarray}
\label{a10}
&&\bigl\lvert P_{\text{report}}(1\lvert 1, \rho, 1) - P_{\text{report}}(1\lvert 0, \rho, 1)\bigr\rvert \nonumber\\
&&\qquad\qquad\qquad\leq 6\delta+2\delta^2+S_{11}^\text{max}(5\delta+\delta^2).
\end{eqnarray}
Thus, from (\ref{newa3}), (\ref{a8}) and (\ref{a10}), the claimed result (\ref{a3}) follows.
\end{proof}

\subsection{Probabilistic reporting strategies\label{gprs}}

In this paper, we introduce \emph{probabilistic reporting strategies}, which generalize the reporting strategies considered above. If a
  detection event $(c_0,c_1)$ occurs when Bob measures in the basis
  $\mathcal{B}_\beta$, Bob reports the message $m=1$ to Alice with
  some probability $S_{c_0c_1\beta}$, for
  $c_0,c_1,\beta\in\{0,1\}$. Thus, if Alice sends Bob a pulse of $k$
photons encoding a state $\rho$ and Bob measures in the basis
$\mathcal{B}_\beta$, the probability that Bob reports the message
$m=1$ is given by
\begin{equation}
\label{a4}
P_{\text{report}}(1\lvert \beta, \rho, k)=\sum_{c_0=0}^1\sum_{c_1=0}^1S_{c_0c_1\beta}P_{\text{det}}(c_0,c_1\lvert \beta, \rho, k),
\end{equation}
for $\beta\in\{0,1\}$ and
$k\in\{0,1,2,\ldots\}$. Note that the previous
  strategies, including symmetrization of losses, are special cases of
  probabilistic reporting strategies.

We define a \emph{trivial strategy} as a probabilistic reporting strategy satisfying
\begin{equation}
\label{new6}
S_{c_0c_1\beta}=S,
\end{equation}
for $c_0,c_1,\beta\in\{0,1\}$ and for some $S\in(0,1]$: we
  exclude $S=0$, which implies that Bob never reports a successful measurement.
A trivial strategy guarantees to Bob that Alice cannot learn any information about $\beta$ from his message $m$. Indeed, if (\ref{new6}) holds then from (\ref{a4}) we have
\begin{eqnarray}
\label{new7}
P_{\text{report}}(1\lvert \beta, \rho, k)&=&S\sum_{c_0=0}^1\sum_{c_1=0}^1P_{\text{det}}(c_0,c_1\lvert \beta, \rho, k)\nonumber\\
&=&S,
\end{eqnarray}
for $\beta\in\{0,1\}$, for any $k-$qubit state $\rho$ and for any $k\in\{0,1,2,\ldots\}$. 

If Bob uses a trivial strategy, he may as well take $S=1$, 
sending $m=1$ to Alice with unit probability.  As shown in Appendix \ref{appA}, this
  satisfies the correctness properties (\ref{new13}) and (\ref{new10}) if losses are low and Bob's detectors have high efficiency. 
However, a trivial strategy cannot satisfy (\ref{new13}) and (\ref{new10}) for a class of common experimental setups (see details in Appendix \ref{appA}).

In the following theorem we show that the only probabilistic reporting strategy with setup I guaranteeing to Bob perfect security against Alice is the trivial strategy (\ref{new6}). This is shown explicitly for the BB84 case and
numerically for the case of general bases. It follows that symmetrization of losses, being a particular probabilistic reporting strategy, does not guarantee security against Alice. See Appendix \ref{appE} for the proof of the theorem.

\begin{theorem}
\label{lemma3}
Suppose that 
\begin{eqnarray}
\label{lemmaa14}
d_{i\beta}&=&0,\\
\label{lemmaa15}
\eta_{i\beta}&=&\eta_i\in(0,1),\\
\label{lemmanew2}
\eta_0&\neq& \eta_1,
\end{eqnarray}
for $i,\beta\in\{0,1\}$, and $\mathcal{B}_0$ and $\mathcal{B}_1$ are
arbitrary distinct qubit orthogonal bases. If Alice sends Bob a pulse
of $k$ photons encoding a state $\rho$, with $k\in\{0,1,2\}$
chosen by Alice and unknown to Bob, then the only
  probabilistic reporting strategy that guarantees to Bob that Alice
  cannot obtain any information about $\beta$ from his message $m$ is
  the trivial strategy (\ref{new6}).
\end{theorem}

\begin{table*}
\centering
\begin{center}
\begin{tabular}{ | c | c | c | }
\hline
\textbf{multiphoton} & \multirow{2}{*}{\textbf{Applies if}}  & \multirow{2}{*}{\textbf{Proposed countermeasures}} \\
\textbf{attack} & & \\
\hline
\multirow{3}{*}{I} & \multirow{3}{*}{$S_{11\beta}=0$ or $S_{11\beta}$ is small enough}  &Reporting double clicks with unit or high probability,  \\
&  & i.e. setting $S_{11\beta}\approx 1$; or, using the trivial reporting \\
&  &   strategy if consistent with experimental parameters\\
\hline
\multirow{3}{*}{II} & $S_{11\beta}=1$ or $S_{11\beta}$ is large enough &Making Bob's detection efficiencies exactly equal, or \\
& and Bob's detection efficiencies are &  sufficiently close; alternatively, using the trivial reporting\\
&  different, or not close enough & strategy  if consistent with experimental parameters \\
\hline
\end{tabular}
\caption{\textbf{Summary of multiphoton attacks I and II.}
  The multiphoton attacks I and II are briefly described in the main
  text. The table indicates for which parameters $S_{11\beta}$ of the probabilistic reporting strategies introduced in the main text the attacks apply and the main proposed countermeasures discussed in the main text. The value of $\beta$ runs over $\{0,1\}$.}
\label{table2}
\end{center}
\end{table*}

\begin{table*}
\centering
\begin{center}
\begin{tabular}{ | c | c | c | c | c | c | c |}
\hline
\textbf{Reporting} & \multirow{2}{*}{\textbf{Definition}}  & \textbf{Brief description} & \textbf{Considered}& \textbf{Subject to} &\textbf{Subject to} &\textbf{Known security}\\
\textbf{strategy}& & \textbf{and comments}& \textbf{before?} &\textbf{attack I?} &\textbf{attack II?}& \textbf{guarantees for Bob} \\
\hline
\multirow{2}{*}{I} & $S_{00\beta}=S_{11\beta}=0$, & Bob only reports & \multirow{2}{*}{Yes} & \multirow{2}{*}{Yes} & \multirow{2}{*}{No} & \multirow{2}{*}{None} \\
 & $S_{01\beta}=S_{10\beta}=1$  & single clicks & & & & \\
\hline
\multirow{6}{*}{II} &  &    &  & \multirow{6}{*}{No} & Yes, if & Perfect security \\
 & $S_{00\beta}=0$,  &  &  & & Bob's & against arbitrary\\
  & $S_{01\beta}=1,$  & Bob reports single & Yes (e.g. in &  &detection &  multiphoton attacks\\
   & $S_{10\beta}=1,$  & and double clicks & Refs. \cite{BBBGST11,LCCLWCLLSLZZCPZCP14}) & &efficiencies & if Bob's detection\\
      &  $S_{11\beta}=1$  &  & & & are & efficiencies are exactly\\
            &   &  & & & different & equal (see Lemma \ref{lemma0})\\
\hline
\multirow{6}{*}{III} &  & \multirow{6}{*}{Defined by (\ref{a5})}   &  &  &  & Good security if Bob \\
 &  $S_{00\beta}=0$,  &  & No, case & Yes, if & Yes, if & guarantees receiving only\\
  &  $S_{01\beta}=\frac{\eta_{\text{min}}}{\eta_{1\beta}}$, &  &  $S_{11\beta}>0$ &  $S_{11\beta}$ is & $S_{11\beta}$ is &  pulses with zero or one\\
   & $S_{10\beta}=\frac{\eta_{\text{min}}}{\eta_{0\beta}}$, & & introduced & small & large & photons and having small\\
      & $S_{11\beta}\in[0,1]$   &  & here & enough & enough & dark count probabilities \\
            &   &  & & & &  (see Lemma \ref{lemma1})\\
\hline
 & Ideally defined &  It aims to  &   & \multirow{8}{*}{Yes} & \multirow{8}{*}{No} &  \multirow{8}{*}{See reporting strategy III}\\
 &  as special case &  symmetrize & Yes, &  &  & \\
 &  of reporting & the detection &  introduced &   &  & \\
 Symmetrization  & strategy  III   &  probabilities  & by Ref. \cite{NJMKW12} &  &  & \\
     of losses  &  with $S_{11\beta}=0$,  & of Bob's   & and used  &  &  &  \\
            &  but precisely  & detectors, but & in Refs. & & &  \\
                        &   defined  by &  fails in general & \cite{NJMKW12,LKBHTKGWZ13,PJLCLTKD14,ENGLWW14} & & &  \\
                                    &    (\ref{new1}) and (\ref{new1.0}) &  (see Theorem \ref{lemma3})  & & & &  \\
\hline
 \multirow{13}{*}{Trivial} &  & & &  \multirow{13}{*}{No} &  \multirow{13}{*}{No} & It is the only known\\
 &   & &  & & & reporting strategy that\\
 &   &  &  & & & guarantees Bob perfect \\
 &   & Bob reports all   & & & &  security against arbitrary  \\
  & $S_{00\beta}=S$,   & detection events  & Yes, in ideal & & & multiphoton attacks (see  \\
  & $S_{01\beta}=S$,   & with the same  & protocols & & & Theorem \ref{lemma3}), but is not \\
    & $S_{10\beta}=S$, & probability $S$; & without & & & compatible with the\\
    & $S_{11\beta}=S$, & where normally & losses, & & & correctness criteria (\ref{new13}) \\
     & $S\in(0,1]$ & $S=1$, but we & for example & & &  and (\ref{new10}) for common\\
        &  & allow $S\in(0,1]$ &  & & & experimental setups\\
                & & &  & & & (see details in\\
                                & & &  & & & Appendix \ref{appA})\\
\hline
   \multirow{5}{*}{Probabilistic} &   & Includes as & No, the & Yes, if  & Yes, if  & \\
 & $S_{c_0c_1\beta}\in[0,1]$   & special cases & general case & $S_{11\beta}$  & $S_{11\beta}$ & See previous\\
  & for  & the reporting & was  & is &  is & reporting \\
  &  $(c_0,c_1)\in\{0,1\}^2$   & strategies &  introduced & small &  large & strategies\\
    &    & discussed above &  here & enough &  enough & \\
\hline
\end{tabular}
\caption{\textbf{Summary of reporting strategies discussed in the main text and their relationships with multiphoton attacks I and II.} In the table, $S_{c_0c_1\beta}$ indicates that Bob sends Alice the message $m=1$ with probability $S_{c_0c_1\beta}$ if the pulse received from her activates a detection event $(c_0,c_1)\in\{0,1\}^2$, or $m=0$ otherwise, when he applies the measurement $\mathcal{B}_\beta$, where $c_i=1$ ($c_i=0$) indicates that the detector D$_i$ clicks (does not click), for $i\in\{0,1\}$ (see setup I illustrated in Fig. \ref{fig1}). The value of $\beta$ runs over $\{0,1\}$.}
\label{table3}
\end{center}
\end{table*}

%\section{Methods}
%\label{methods}

\section{Detection probabilities in multiphoton attacks}
\label{methods}
\begin{lemma}
\label{lemma2}
Suppose that Alice sends Bob a pulse of $k$ photons, each encoding a qubit state $\rho_{\text{qubit}}$. That is, the state encoded by the pulse is $\rho=\rho_{\text{qubit}}^{\otimes k}$. We define
\begin{equation}
\label{xab1main}
q_\beta=\langle \psi_{0\beta}\rvert \rho_{\text{qubit}}\lvert \psi_{0\beta}\rangle,
\end{equation}
for $\beta\in\{0,1\}$. Then \begin{eqnarray}
\label{xab2}
&&P_{\text{det}}(0,0\vert \beta, \rho ,k)\nonumber\\
&&\quad=(1\!-\!d_{0\beta})(1\!-\!d_{1\beta})\bigl[q_{\beta}(1\!-\!\eta_{0\beta})\!+\!(1\!-\!q_\beta)(1\!-\!\eta_{1\beta})\bigr]^k,\nonumber\\
&&P_{\text{det}}(0,1\vert \beta, \rho ,k)\nonumber\\
&&\quad=(1\!-\!d_{0\beta})\bigl(1\!-\!q_{\beta}\eta_{0\beta}\bigr)^k\nonumber\\
&&\quad\quad-(1\!-\!d_{0\beta})(1\!-\!d_{1\beta})\bigl[q_{\beta}(1\!-\!\eta_{0\beta})\!+\!(1\!-\!q_\beta)(1\!-\!\eta_{1\beta})\bigr]^k,\nonumber\\
&&P_{\text{det}}(1,0\vert \beta, \rho ,k)\nonumber\\
&&\quad=(1\!-\!d_{1\beta})\bigl(1\!-\!(1\!-\!q_{\beta})\eta_{1\beta}\bigr)^k\nonumber\\
&&\quad\quad-(1\!-\!d_{0\beta})(1\!-\!d_{1\beta})\bigl[q_{\beta}(1\!-\!\eta_{0\beta})\!+\!(1\!-\!q_\beta)(1\!-\!\eta_{1\beta})\bigr]^k,\nonumber\\
&&P_{\text{det}}(1,1\vert \beta, \rho ,k)\nonumber\\
&&\quad =1+(1\!-\!d_{0\beta})(1\!-\!d_{1\beta})\bigl[q_{\beta}(1\!-\!\eta_{0\beta})\!+\!(1\!-\!q_\beta)(1\!-\!\eta_{1\beta})\bigr]^k\nonumber\\
&&\quad\quad-(1\!-\!d_{0\beta})\bigl(1\!-\!q_{\beta}\eta_{0\beta}\bigr)^k\nonumber\\
&&\quad\quad-(1\!-\!d_{1\beta})\bigl(1\!-\!(1\!-\!q_{\beta})\eta_{1\beta}\bigr)^k,
\end{eqnarray}
for $\beta\in\{0,1\}$ and $k\in\{0,1,2\ldots\}$. Now consider that Alice prepares a pulse of $k$ photons with the probability distribution of a coherent source with average photon number $\mu>0$, for $k\in\{0,1,2\ldots\}$. Then the probability $P_{\text{det}}^{\text{cs}}(c_0,c_1\vert \beta,\rho,\mu)$ of the detection event $(c_0,c_1)$ is given by
\begin{eqnarray}
\label{ccc6lemma}
&&P_{\text{det}}^{\text{cs}}(0,0\vert \beta,\rho,\mu)\nonumber\\
&&\qquad=(1-d_{0\beta})(1-d_{1\beta})e^{-\mu\bigl[q_\beta\eta_{0\beta}+(1-q_\beta)\eta_{1\beta}\bigr]},\nonumber\\
&&P_{\text{det}}^{\text{cs}}(0,1\vert \beta,\rho,\mu)\nonumber\\
&&\qquad=(1-d_{0\beta})e^{-\mu q_\beta\eta_{0\beta}}\nonumber\\
&&\qquad\qquad-(1-d_{0\beta})(1-d_{1\beta})e^{-\mu\bigl[q_\beta\eta_{0\beta}+(1-q_\beta)\eta_{1\beta}\bigr]},\nonumber\\
&&P_{\text{det}}^{\text{cs}}(1,0\vert \beta,\rho,\mu)\nonumber\\
&&\qquad =(1-d_{1\beta})e^{-\mu(1-q_\beta)\eta_{1\beta}}\nonumber\\
&&\qquad\qquad-(1-d_{0\beta})(1-d_{1\beta})e^{-\mu\bigl[q_\beta\eta_{0\beta}+(1-q_\beta)\eta_{1\beta}\bigr]},\nonumber\\
&&P_{\text{det}}^{\text{cs}}(1,1\vert \beta,\rho,\mu)\nonumber\\
&&\qquad=1+(1-d_{0\beta})(1-d_{1\beta})e^{-\mu\bigl[q_\beta\eta_{0\beta}+(1-q_\beta)\eta_{1\beta}\bigr]}\nonumber\\
&&\qquad \qquad-(1-d_{0\beta})e^{-\mu q_\beta\eta_{0\beta}}-(1-d_{1\beta})e^{-\mu(1-q_\beta)\eta_{1\beta}},\nonumber\\
\end{eqnarray}
for $c_0,c_1,\beta\in\{0,1\}$ and $\mu>0$.
\end{lemma}

\begin{proof}
Let $k_0$ be the number of photons that go to detector D$_0$. The number of photons that go to detector D$_1$ is $k_1=k-k_0$.

As previously stated, in our model, the dark counts and each photo-detection are independent random events. It follows straightforwardly from (\ref{xab1main}) that:
\begin{eqnarray}
\label{xaccc1}
&&P_{\text{det}}(0,0\vert \beta, \rho ,k)\nonumber\\
&&\quad =(1-d_{0\beta})(1-d_{1\beta})\sum_{k_0=0}^k \biggl[\begin{pmatrix}
k \\ k_0
\end{pmatrix} (q_\beta)^{k_0}(1-q_\beta)^{k-k_0}\times\nonumber\\
&&\qquad\qquad\times(1-\eta_{0\beta})^{k_0}(1-\eta_{1\beta})^{k-k_0}\biggr],\nonumber\\
&&P_{\text{det}}(0,1\vert \beta, \rho ,k)\nonumber\\
&&\quad=\sum_{k_0=0}^k \biggl\{\begin{pmatrix}
k \\ k_0
\end{pmatrix} (q_\beta)^{k_0}(1-q_\beta)^{k-k_0}\times\nonumber\\
&&\qquad\qquad\times\Bigl[(1-d_{0\beta})(1-\eta_{0\beta})^{k_0}\bigl(1-(1-\eta_{1\beta})^{k-k_0}\bigr)\nonumber\\
&&\qquad\qquad\qquad+d_{1\beta}(1-d_{0\beta})(1-\eta_{0\beta})^{k_0}(1-\eta_{1\beta})^{k-k_0}\Bigr]\biggr\},\nonumber\\
&&P_{\text{det}}(1,0\vert \beta, \rho ,k)\nonumber\\
&&\quad=\sum_{k_0=0}^k \biggl\{\begin{pmatrix}
k \\ k_0
\end{pmatrix} (q_\beta)^{k_0}(1-q_\beta)^{k-k_0}\times\nonumber\\
&&\qquad\qquad\times\Bigl[(1-d_{1\beta})\bigl(1-(1-\eta_{0\beta})^{k_0}\bigr)(1-\eta_{1\beta})^{k-k_0}\nonumber\\
&&\qquad\qquad\qquad+d_{0\beta}(1-d_{1\beta})(1-\eta_{0\beta})^{k_0}(1-\eta_{1\beta})^{k-k_0}\Bigr]\biggr\},\nonumber\\
&&P_{\text{det}}(1,1\vert \beta, \rho ,k)\nonumber\\
&&\quad=\sum_{k_0=0}^k \biggl\{\begin{pmatrix}
k \\ k_0
\end{pmatrix} (q_\beta)^{k_0}(1-q_\beta)^{k-k_0}\times\nonumber\\
&&\qquad\qquad\times\Bigl[\bigl(1-(1-\eta_{0\beta})^{k_0}\bigr)\bigl(1-(1-\eta_{1\beta})^{k-k_0})\nonumber\\
&&\qquad\qquad\qquad+d_{1\beta}\bigl(1-(1-\eta_{0\beta})^{k_0}\bigr)(1-\eta_{1\beta})^{k-k_0}\nonumber\\
&&\qquad\qquad\qquad+d_{0\beta}(1-\eta_{0\beta})^{k_0}\bigl(1-(1-\eta_{1\beta})^{k-k_0}\bigr)\nonumber\\
&&\qquad\qquad\qquad+d_{0\beta}d_{1\beta}(1-\eta_{0\beta})^{k_0}(1-\eta_{1\beta})^{k-k_0}\Bigr]\biggr\},
\end{eqnarray}
for $\beta\in\{0,1\}$ and $k\in\{0,1,2,\dots\}$. Using the binomial theorem, (\ref{xab2}) follows straightforwardly from (\ref{xaccc1}), as claimed.

Now we suppose that Alice generates photon pulses with a coherent source of average photon number $\mu$. The probability that a pulse has $k$ photons is $p_k=e^{-\mu}\mu^k/k!$, hence, the probability $P_{\text{det}}^{\text{cs}}(c_0,c_1\vert\beta,\rho,\mu)$ that, for a pulse, the detector $\text{D}_0$ clicks if $c_0=1$ (does not click if $c_0=0$) and the detector $\text{D}_1$ clicks if $c_1=1$ (does not click if $c_1=0$) is given by
\begin{eqnarray}
\label{wcseq}
P_{\text{det}}^{\text{cs}}(c_0,c_1\vert\beta,\rho,\mu)&=&\sum_{k=0}^{\infty}p_kP_{\text{det}}(c_0,c_1\vert \beta, \rho ,k)\nonumber\\
&=&\sum_{k=0}^{\infty}\frac{e^{-\mu}\mu^k}{k!}P_{\text{det}}(c_0,c_1\vert \beta, \rho ,k),\nonumber\\
\end{eqnarray}
for $c_0,c_1,\beta\in\{0,1\}$ and $\mu>0$. Thus, from (\ref{xaccc1}) and (\ref{wcseq}), and using $e^x=\sum_{k=0}^\infty \frac{x^k}{k!}$ for $x>0$, we obtain (\ref{ccc6lemma}), as claimed.
\end{proof}

\section{Experiment}
\label{experiment}
\subsection{Experimental evaluation of the detection probabilities in a multiphoton attack}

\begin{figure}
\includegraphics[scale=0.57]{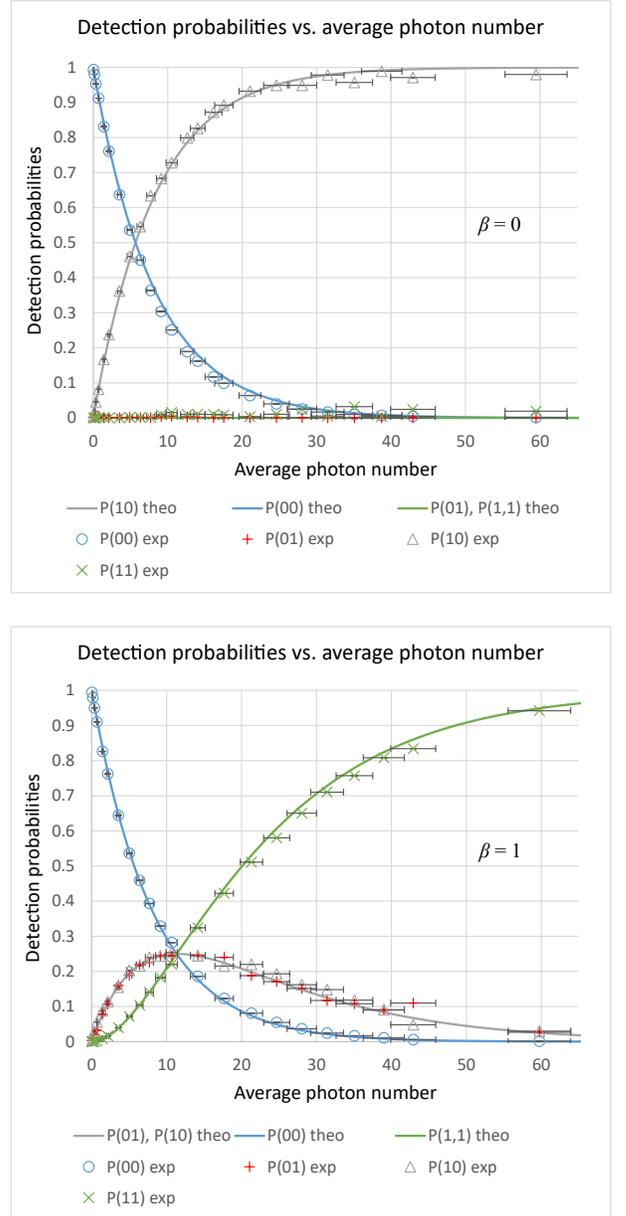}
\caption{\label{newplots} \textbf{Detection probabilities.} We plot the experimental estimates obtained for the detection probabilities $P_{\text{det}}^{\text{cs}}(c_0,c_1\vert \beta,\rho,\mu)$, indicated with the labels ``$P(c_0c_1)$ exp'', and their theoretical predictions given by (\ref{first}), indicated with the labels ``$P(c_0c_1)$ theo'', with $\eta=0.12$, for various values of the average photon number $\mu$, for the qubit state $\rho_\text{qubit}=\lvert 0\rangle\langle 0\rvert$, and for $c_0,c_1,\beta\in\{0,1\}$. The horizontal and vertical uncertainty bars are included. The vertical uncertainty bars are so small that they are unnoticeable and lie within the corresponding markers.}
\end{figure}

We experimentally evaluated the validity of the theoretical
predictions for the detection probabilities given by
(\ref{ccc6lemma}), in Lemma \ref{lemma2}. For this, Alice prepared
the state with vertical polarization using a weak coherent photon source and varied the average photon number $\mu$. This polarization corresponds to a qubit state
$\rho_\text{qubit}=\lvert 0\rangle\langle 0\rvert$, which has Bloch vector along the $z$ axis, given by $\vec{r}=(0,0,1)$.
That is, as above, each pulse encodes in the polarization a
$k-$qubit
state $\rho=\rho_\text{qubit}^{\otimes
  k}$, where
$k$
has a probability distribution $p_k=\frac{e^{-\mu}\mu^k}{k!}$,
for $k\in\{0,1,2,\ldots\}$.
For $c_0,c_1,\beta\in\{0,1\}$,
we fixed Bob's measurement basis to $\mathcal{B}_\beta$
and computed the frequency of the detection events
$(c_0,c_1)$
for a large number of pulses $N=4\times 10^6$, estimating in this way the detection probability
$P_\text{det}^\text{cs}(c_0,c_1\vert\beta,\rho,\mu)$.
As above, Bob's measurement bases $\mathcal{B}_0$
and $\mathcal{B}_1$
are the computational and Hadamard bases, respectively.

In practice, this was done as follows. Pulses of 1 $\mu$s with vertical polarization were generated by Alice at a repetition rate of 200 kHz for $20$ s using two fibered amplitude modulators connected in series. An unbalanced coupler was then used for power calibration in order to fix the value of $\mu$. The estimated uncertainty for the coupler's splitting ratio is $\pm5$ \%, and the uncertainty on the power measurement is also $\pm5$ \% from the relevant data sheet. Combining these uncertainties in quadrature, this gives a total uncertainty on $\mu$ of $\pm7$ \%. One arm transmitted the signal, while the other was sent to a photodiode to calibrate the power. Bob's measurement basis was selected using a half wave plate, with angle uncertainty of $\pm 0.25 ^{\circ}$, hence contributing with $\pm 0.5^{\circ}$ and $\pm 1^{\circ}$ uncertainties in the physical polarization angle and in the Bloch sphere angle, respectively. The polarization was then projected using a fibered polarization beam splitter, and detected with two id230 free-running avalanche detectors, at Bob's site (see Fig. \ref{fig1}).

In Fig. \ref{newplots}, we have plotted the obtained experimental estimates
of $P_\text{det}^\text{cs}(c_0,c_1\vert\beta,\rho,\mu)$,
for $c_0,c_1,\beta\in\{0,1\}$ and
for various values of $\mu$. Fig. \ref{newplots} also shows the plots for the theoretical
predictions given by (\ref{ccc6lemma}). In our theoretical
predictions, we assume that $d_{i\beta}=0$
and $\eta_{i\beta}=\eta\in(0,1)$,
for $i,\beta\in\{0,1\}$.
We measured a value of $\eta=0.12$. Note that this value includes Bob's detection efficiencies and the probability of loss in the quantum channel. Using the definition (\ref{xab1main}), we have  that $q_0=\langle
0\rvert \rho_\text{qubit}\rvert 0\rangle=\frac{1}{2}(1+r_z)$
and $q_1=\langle +\rvert \rho_\text{qubit}\rvert
+\rangle=\frac{1}{2}(1+r_x)$, where $r_z$ and $r_x$ are the
$z$
and the $x$
components of the Bloch vector $\vec{r}$, and where $\lvert+\rangle=\frac{1}{\sqrt{2}}\bigl(\lvert 0\rangle+\lvert 1\rangle\bigr)$. Thus, we have $q_0=1$
and $q_1=\frac{1}{2}$. By substituting these values in
(\ref{ccc6lemma}), we obtain the plotted theoretical curves. These are given by
\begin{eqnarray}
\label{first}
P_{\text{det}}^{\text{cs}}(0,0\vert \beta,\rho,\mu) &=& e^{-\mu\eta}, \nonumber\\
P_{\text{det}}^{\text{cs}}(0,1\vert 0,\rho,\mu)&=&0,\nonumber\\
P_{\text{det}}^{\text{cs}}(1,0\vert 0,\rho,\mu)&=&1-e^{-\mu\eta},\nonumber\\
P_{\text{det}}^{\text{cs}}(1,1\vert 0,\rho,\mu)&=&0,\nonumber\\
P_{\text{det}}^{\text{cs}}(c,\bar{c}\vert 1,\rho,\mu)&=&e^{-\frac{\mu \eta}{2}}-e^{-\mu\eta},\nonumber\\
P_{\text{det}}^{\text{cs}}(1,1\vert 1,\rho,\mu)&=&1+e^{-\mu\eta}-2e^{-\frac{\mu \eta}{2}},\nonumber\\
\end{eqnarray}
 for $\beta,c\in\{0,1\}$. In Fig. \ref{newplots}, we observe a good agreement between the experimental estimates of the detection probabilities $P_{\text{det}}^{\text{cs}}(c_0,c_1\vert \beta,\rho,\mu)$ and the theoretical predictions given by (\ref{first}), for $c_0,c_1,\beta\in\{0,1\}$. The slight disagreements observed for a few experimental points with the theoretical predictions of (\ref{first}) are possibly explained by the uncertainties on some experimental quantities, like the detection efficiency and the measurement bases, for instance, in addition to the uncertainties on the average photon number $\mu$ given by the horizontal uncertainty bars.

\subsection{Experimental simulation of multiphoton attack I}

We illustrate a version of multiphoton attack I in which Alice sends Bob a photon pulse prepared with a coherent source, with average photon number $\mu>>1$ chosen by Alice. We assume that Bob applies reporting strategy I, i.e. he only sets valid measurement outcomes, and sends Alice the message $m=1$, for pulses that generate a click in only one of his detectors. We consider the case that $\mathcal{B}_0$ and $\mathcal{B}_1$ are the computational and Hadamard bases, respectively. Each photon in the pulse encodes a qubit state $\rho_{\text{qubit}}=\lvert \psi_{00}\rangle\langle \psi_{00}\rvert=\lvert 0\rangle\langle 0\rvert$. Thus, we can use Lemma \ref{lemma2}, with $q_0=1$ and $q_1=\frac{1}{2}$ in (\ref{xab1main}). If Bob sends Alice the message $m=1$, Alice guesses that $\beta=0$. If Bob sends $m=0$, Alice guesses that $\beta=1$. Thus, since we suppose that Bob generates $\beta$ randomly, Alice's probability $P_\text{guess}^\text{cs}(\mu)$ to guess Bob's bit $\beta$ in this attack is given by
\begin{eqnarray}
\label{attack}
P_\text{guess}^\text{cs}(\mu)&=&\frac{1}{2}\bigl[P_{\text{det}}^\text{cs}(1,0\vert 0, \rho ,\mu)+P_{\text{det}}^\text{cs}(0,1\vert 0, \rho ,\mu)\nonumber\\
&&\qquad+P_{\text{det}}^\text{cs}(0,0\vert 1, \rho ,\mu)+P_{\text{det}}^\text{cs}(1,1\vert 1,\rho, \mu)\bigl],\nonumber\\
\end{eqnarray}
where $P_{\text{det}}^\text{cs}(c_0,c_1\vert \beta, \rho ,\mu)$ is given by (\ref{ccc6lemma}), for $c_0,c_1,\beta\in\{0,1\}$, and where $q_0=1$ and $q_1=\frac{1}{2}$. In the simple case that $d_{i\beta}=0$ and $\eta_{i\beta}=\eta\in(0,1)$, for $i,\beta\in\{0,1\}$, and $q_0=1$ and $q_1=\frac{1}{2}$, (\ref{ccc6lemma}) reduces to (\ref{first}). Thus, from (\ref{first}) and (\ref{attack}), we obtain  (\ref{attackexmain}).

We implemented an experimental simulation of the multiphoton attack described above. In our simulation, Alice and Bob use setup I illustrated in Fig. \ref{fig1}. Alice's photon source is a coherent source with average photon number $\mu$. Bob measures the polarization of the pulse in the basis $\mathcal{B}_0$ or in the basis $\mathcal{B}_1$, which are the horizontal-vertical polarization basis and the diagonal-diagonal basis (corresponding to polarization angle of $45^{\circ}$ and $-45^{\circ}$ from the horizontal axis towards the vertical axis), respectively. In the Bloch sphere, these are the computational and Hadamard bases, respectively.

Alice sets the value of $\mu$ and sends Bob photon pulses with the same polarization, in vertical polarization, corresponding to the qubits state $\lvert 0\rangle$. That is, each pulse encodes in the polarization a $k-$qubit state $\rho=\rho_\text{qubit}^{\otimes k}$, where $\rho_\text{qubit}=\lvert 0\rangle\langle 0\rvert$ is a qubit state, and where $k$ has a probability distribution $p_k=\frac{e^{-\mu}\mu^k}{k!}$, for $k\in\{0,1,2,\ldots\}$. For $c_0,c_1,\beta\in\{0,1\}$, we fix Bob's measurement basis to $\mathcal{B}_\beta$ and compute the frequency of the detection events $(c_0,c_1)$ for a large number of pulses $N$, estimating experimentally in this way the value of the detection probability $P_\text{det}^{\text{cs}}(c_0,c_1\vert\beta,\rho,\mu)$. Alice's probability $P_\text{guess}^\text{cs}(\mu)$ to guess Bob's bit $\beta$ in the considered attack is given by (\ref{attack}). 

We have plotted the obtained experimental estimates of $P_\text{guess}^\text{cs}(\mu)$ for various values of $\mu$ in Fig. \ref{guessingexp}. In practice, the experimental data plotted in Fig. \ref{guessingexp} corresponds to the same experimental data plotted in Fig. \ref{newplots}, but using (\ref{attack}). As follows from (\ref{attack}), computing $P_\text{guess}^\text{cs}(\mu)$ requires some probabilities $P_\text{det}^{\text{cs}}(c_0,c_1\vert\beta,\rho,\mu)$ with both values of $\beta\in\{0,1\}$. Because in our experimental points of the plots in Fig. \ref{newplots} for different value of $\beta\in\{0,1\}$, the values of the average photon number are close but not exactly equal, we obtain each point in the experimental plot of Fig. \ref{guessingexp} by taking $\mu$ as the average $\mu=\frac{\mu_0+\mu_1}{2}$, where $\mu_i$ is the value of the average photon number for the corresponding experimental point in the plot of Fig. \ref{newplots} with $\beta=i$, for $i\in\{0,1\}$. As explained above, the experimental uncertainty for each value of $\mu_i$ is $7$  \%, for $i\in\{0,1\}$. The uncertainty for each value of $\mu$ in the experimental plot of Fig. \ref{guessingexp} is obtained by combining in quadrature the uncertainties for the corresponding values of $\mu_0$ and $\mu_1$ given in the experimental plots of Fig. \ref{newplots}. The theoretical prediction for $P_\text{guess}^\text{cs}(\mu)$ given by (\ref{attackexmain}) is also shown in Fig. \ref{guessingexp}, for comparison with the experimental data. A good agreement between the experimental data and the theoretical prediction is observed. As mentioned for the plots of Fig. \ref{newplots}, the slight disagreements observed for a few experimental points in the plot of Fig. \ref{guessingexp} with the theoretical predictions of (\ref{attackexmain}) are expected due to the uncertainties on some experimental quantities, like the detection efficiency and the measurement bases, for example, in addition to the uncertainties on the average photon number $\mu$ given by the horizontal uncertainty bars.

\subsection{Statistical Information}
The number of photon pulses used for estimating the probabilities $P_\text{det}^\text{cs}(c_0,c_1\vert \beta,\rho,\mu)$ plotted in Fig. \ref{newplots} and the probability $P_\text{guess}^\text{cs}(\mu)$ plotted in Fig. \ref{guessingexp} for each value of the average photon number $\mu$ was $N = 4\times 10^6$, for $c_0,c_1,\beta\in\{0,1\}$. The uncertainty for each experimental estimate of the probabilities $P_\text{det}^\text{cs}(c_0,c_1\vert \beta,\rho,\mu)$ plotted in Fig. \ref{newplots} was taken as the standard error, given by 
\begin{equation}
\label{standarderror}
\Delta P = \sqrt{ \frac{P (1-P)}{N}},
\end{equation}
where $P$ is the obtained experimental estimate of the probability $P_\text{det}^\text{cs}(c_0,c_1\vert \beta,\rho,\mu)$,  for $c_0,c_1,\beta\in\{0,1\}$ and for each value of $\mu$. The uncertainty for each experimental estimate of the probabilities $P_\text{guess}^\text{cs}(\mu)$ plotted in Fig. \ref{guessingexp} was obtained using (\ref{attack}) and (\ref{standarderror}), combining in quadrature the uncertainties of the experimental estimates for the  probabilities $P_\text{det}^\text{cs}(c_0,c_1\vert \beta,\rho,\mu)$ appearing in (\ref{attack}).

\dpg{
\section{Private measurement of an unknown qudit state}

In this paper we have analysed multiphoton attacks for the task of private measurement of an unknown qubit state. This task is easily extended to \emph{private measurement of an unknown qudit state}, where Alice sends Bob a quantum state of Hilbert space dimension $N_\text{out}\geq 2$, Bob measures it in one of $N_\text{bas}\geq 2$ orthogonal bases with $N_\text{out}$ outcomes. Let $\mathbb{Z}_N=\{0,1,\ldots,N-1\}$, for all integers $N\geq 2$. An ideal protocol to implement this task is the following:
\begin{enumerate}
\item Alice prepares a quantum state $\lvert \psi\rangle$ of Hilbert space dimension $N_\text{out}$ randomly from a set 
\begin{equation}
\mathcal{S}=\bigl\{\lvert \psi_{ij}\rangle\vert i\in\mathbb{Z}_{N_\text{out}}, j\in\mathbb{Z}_{N_\text{bas}}\bigr\}\nonumber
\end{equation}
and sends it to Bob. 
\item Bob generates a random number $\beta\in\mathbb{Z}_{N_\text{bas}}$ privately and measures $\lvert\psi\rangle$ in an orthogonal  basis $\mathcal{B}_{\beta}=\{\lvert\psi_{i\beta}\rangle\}_{i\in\mathbb{Z}_{N_{\text{out}}}}$.
\end{enumerate}
This protocol reduces to the ideal protocol to implement the task of private measurement of an unknown qubit state presented in section \ref{PMQS} in the case $N_\text{bas}=N_\text{out}=2$.

A protocol to implement the task is called $\epsilon_\text{guess}-$secure against Alice if Alice's probability $P_\text{guess}$ to guess Bob's chosen number $\beta$ satisfies
\begin{equation}
\label{zzz1}
P_\text{guess}\leq \frac{1}{N_\text{bas}}+\epsilon_\text{guess},
\end{equation}
for any possible strategy (not necessarily honestly following
the protocol) of Alice. As before, we say the protocol is secure against
Alice if $\epsilon_\text{guess}\rightarrow 0$ as some protocol security parameter
is increased. In general, a dishonest Alice may deviate in
any way from the protocol. The correctness criteria (\ref{new13}) and (\ref{new10}) remain the same.

An ideal setup allows us to effectively implement
the task with
perfect security against Alice ($\epsilon_\text{guess}=0$) and perfect correctness ($P_\text{det}=1$ and $P_\text{error}=0$). However, in practical implementations,  
the ideal protocol and setup have to be altered to allow for imperfect
sources, channels and measuring devices. 

Investigating practical protocols to implement this task  with photonic setups for the general case $N_\text{bas}, N_\text{out}\geq 2$ and the multiphoton attacks arising is left as an open problem. However, our results can be easily extended if the bases $\mathcal{B}_j$ are products of qubit bases, for $j\in\mathbb{Z}_{N_\text{bas}}$. 

For example, if $N_\text{bas}=N_\text{out}=2^n$ and the bases $\mathcal{B}_0,\ldots,\mathcal{B}_{N_\text{bas}-1}$ are products of $n$ qubit bases, for some integer $n\geq 1$, then $n$ consecutive implementations of the practical protocol of section \ref{APP} with setup I (see Fig. \ref{fig1}), labelled by $i=1,2,\ldots,n$, could be used as a subroutine to implement the task. In particular, if for $i=1,2,\ldots,n$, in the $i$th instance of the practical protocol with setup I, 1) Bob sends Alice the bit $m_i$ using reporting strategy II, and 2) both detectors have efficiency $\eta$, then Lemma \ref{lemma0} can be straightforwardly extended to show that Alice cannot obtain any information about $\beta$ from Bob's message $m=(m_1,m_2,\ldots,m_n)$. Thus, in this case, Bob is perfectly protected against arbitrary multiphoton attacks by Alice.

\subsection{An extension of Lemma \ref{lemma0}}

Lemma \ref{lemma0} can also be generalized to the following setup, which is an extension of setup I for the case $N_\text{bas}\geq 2,N_\text{out}\geq 2$. Alice sends Bob a photon pulse of $k$ photons, encoding an arbitrary quantum state $\rho$, where each photon encodes a quantum state of Hilbert space dimension $N_\text{out}$ in some degrees of freedom. It is clear that for $N_\text{out}>2$, these degrees of freedom cannot be the photon's polarization, but can be the time-bin, for example. Bob's setup includes a quantum measuring device, with  $N_\text{out}$ threshold single photon detectors $\text{D}_0,\ldots,\text{D}_{N_{\text{out}}-1}$. Bob sets his system to apply a quantum measurement in the orthogonal basis $\mathcal{B}_{\beta}$, which has $N_\text{out}$ outcomes, for some $\beta\in\mathbb{Z}_{\text{bas}}$. 

In an ideal setting in which Bob's detectors are perfectly efficient and do not have dark counts, there are not any losses of the transmitted photons, and Alice sends Bob a pulse with exactly one photon, exactly one of Bob's detectors clicks with unit probability. However, in a realistic scenario in which these conditions do not hold, any combination of detectors can click with some non-zero probability.

Bob reports to Alice that the pulse was (was not) successfully measured, by sending her the message $m=1$ ($m=0$) if at least one detector clicks (if no detector clicks). This is a straightforward extension of reporting strategy II to this generalized setting. Suppose that the detector D$_i$ has dark count probability $d_i$ and efficiency $\eta_i=\eta$, for $i\in\mathbb{Z}_{N_\text{out}}$. The proof of Lemma \ref{lemma0} can be straightforwardly extended to show that the probability that Bob sends the message $m=1$ to Alice is given by
\begin{equation}
\label{lastequation}
P_\text{report}(1\vert \beta,\rho,k)=1-(1-\eta)^k\prod_{i=0}^{N_{\text{out}}-1}(1-d_i),
\end{equation}
for $\beta\in\mathbb{Z}_{N_\text{bas}}$ and $k\in\{0,1,2,\ldots\}$. Thus, Alice cannot learn any information about $\beta$ from the message $m$.
}

\section{Discussion}

\label{exotic}
We have highlighted known, and introduced new, attacks on photonic
implementations of mistrustful quantum cryptographic protocols. 
These arise because Alice may send multiphoton
pulses that are different from the single photon pulses
envisaged in the ideal versions of the protocols and that
give her statistically distinguishable results when Bob
uses standard detectors.   

Lemma \ref{lemma0} suggests that a possible
  countermeasure against multiphoton attacks is to make the
  efficiencies of Bob's detectors as close as possible. Bob
  could use attenuators for this purpose; but care should be taken to
  guarantee that these act linearly, i.e that their action on
  multiphoton states is given by the product of their action on
  individual photons, and that their action does not depend on
  wavelength or other degrees of freedom that Alice could exploit to
  her advantage, as in the side-channel attacks discussed below. We
  have investigated this possibility in Appendix \ref{appB}
  by deriving security bounds when Bob applies reporting strategies II or III that provide perfect security against
  multiphoton attacks in the limit that the difference of the
  detector efficiencies equals zero (see Lemmas \ref{lemma5} -- \ref{lemmanewbound} in Appendix \ref{appB}). However, as discussed in Appendix \ref{appB}, our security bounds are not useful
  in practical settings unless this difference is extremely
  small. We leave as an open problem to investigate whether this
  countermeasure can provide useful security guarantees against
  multiphoton attacks in practice.
  
Another possible countermeasure is for Bob to use variations of the setup analysed here to probabilistically infer if a pulse is multiphoton. For example, Bob could use photon-number-resolving (PNR) detectors (e.g. \cite{ASSBW03,FJPF03,KYS08}) instead of threshold detectors, and abort if the detector determines that the received pulse has more than a threshold number of photons; approximate PNR detectors can be built from various beam splitters and threshold single photon detectors or from a single threshold detector using a time-multiplexing technique \cite{ASSBW03,FJPF03}, for instance. 
%To our knowledge, PNR detectors cannot perfectly resolve the number of photons in a pulse. Nevertheless, we think it is worth investigating their use as a countermeasure against multiphoton attacks.
This measure may be helpful against attacks in which Alice sends one or more pulses with large photon number. Note though that it tends to reduce Bob's effective detection efficiencies. 

%\dddd{About your comment: ``These measures may be helpful against bright pulse attacks, in which Alice sends a small number of pulses with large photon number.   Note though that they tend to reduce Bob's effective detection efficiencies.'', please note: 1) we are talking about a task in which Alice sends a single pulse here --as in the text--; 2) I am saying that Bob aborts is Bob infers the pulse has many photons, hence, the detection efficiency reporting double clicks is a simpler countermeasure that considerably helps against these attacks. I do not think we need to add more text about this.}

Another possible countermeasure is that Bob aborts if he obtains a double click. In particular, in extensions of the task presented here in which Alice sends Bob $N>>1$ pulses, Bob aborts if the number of double-click events is above a threshold. In this way, dishonest Alice must limit her number of multiphoton pulses. We have outlined in Appendix \ref{bounds} how to derive security bounds if Bob can guarantee that a sufficiently small fraction of pulses are multiphoton, if $N\delta \ll 1$, where $\delta$ is the dark count rate. However, this is a very tight condition for large $N$.

The previous two countermeasures cannot guarantee perfect protection in practice. As discussed, a realistic source emits multiphoton pulses with nonzero probability. Thus, a dishonest technologically advanced Alice can mimic the source by choosing the photon number of each pulse sent to mimic the overall source statistics. In this way, Bob cannot detect that Alice is cheating, but Alice can still  learn Bob's measurement basis with nonzero probability. We have given an explicit analysis of this attack in an extension of the task with $N>>1$ pulses in Appendix \ref{double}.

%\dddd{DAMIAN'S COMMENT: I prefer to use ``in practice'' rather than ``designed for imperfect sources''. As we have discussed, there are not perfect sources. For the same reason I prefer ``a realistic source emits'' instead of ``if a source emits''.}

A way to perfectly avoid multiphoton attacks is to use the trivial reporting strategy in which Bob does not report any losses to Alice. As previously discussed, this measure 
does not allow Bob to satisfy the correctness criteria (\ref{new13}) and (\ref{new10}) for a broad class of \dpg{practical experimental parameters using setup I} (see Appendix \ref{appA}). However, we believe that current technology could allow good security for a range of useful tasks based on private measurement of an unknown qubit state \dpg{using setup I}. For example, the experiment of Ref. \cite{LLRZBLBZCFZP21} reports detection efficiency higher than $80\%$ and state fidelity higher than $99\%$.

%Although analysing whether this is a viable option using state of the art quantum technology depends on the particular envisaged task and implementation, Ref. \cite{KLPGR21} argues that this is possible for quantum S-money tokens using the quantum technology reported in Ref. \cite{LLRZBLBZCFZP21}, for example.

%\led{COMMENT: I suggest to modify the previous sentence by the following. However, we believe that current technology could allow good security for a range of useful tasks based on private measurement of an unknown qubit state. For example, the experiment of Ref. \cite{LLRZBLBZCFZP21} reports detection efficiency higher than $80\%$ and state fidelity higher than $99\%$.]

%However, we believe this could be a viable option using state of the art quantum technology. For example, the experiment of Ref. \cite{LLRZBLBZCFZP21} reports detection efficiency higher than $80\%$ and state fidelity higher than $99\%$.

Another approach to avoid multiphoton attacks is to use protocols that are not based on the task of private measurement of an unknown qubit state presented here. For example, protocols using a reversed version of this task, in which Bob sends quantum states to Alice that must remain private to her, could be used instead. However, this opens other security problems. In particular, since in practice, Bob's photon source emits multiphoton pulses with nonzero probability, Alice can unambiguously distinguish the received states with nonzero probability. This applies in practice to the quantum coin flipping protocol of Ref. \cite{ZYWCCGH15}, for instance.

%\dddd{DAMIAN'S COMMENT: A realistic source emits multiphoton pulses with nonzero probability. We have discussed this in the paper. Thus, I am not replacing ``since'' by ``if''. In the last sentence I need to say ``in practice'' because the protocol considered in the paper is ideal (where qubits -- i.e. single photons -- are generated exactly).}

We have focused here on the simple experimental setup of Fig. \ref{fig1}
  with two detectors. However, our analyses can be extended to more general setups with more detectors. In Appendix \ref{NJMKW12} we have investigated a setup with four detectors, setup II, (see Fig. \ref{setupII} in Appendix \ref{NJMKW12}). There, we analyse extensions of reporting strategies I, II and III to setup II. We extend multiphoton attack I to this setup and show that it makes implementations insecure if Bob only reports single clicks as successful measurements (see Fig. \ref{plots} in Appendix \ref{NJMKW12}). We extend Lemma \ref{lemma0} to setup II and show that if Bob's detectors have equal detection efficiencies and dark count probabilities independent of the measurement bases then an extension of reporting strategy II guarantees security against arbitrary multiphoton attacks (see Lemma \ref{setup2equaldet} in Appendix \ref{NJMKW12}). We extend multiphoton attack II to this setup and show that it makes implementations insecure if Bob's detectors have different efficiencies
  %then Alice can guess Bob's measurement basis with probability arbitrarily close to unity if Alice sends pulses with arbitrarily large number of photons
  (see Lemma \ref{lastlemma} and Fig. \ref{lastfig} in Appendix \ref{NJMKW12}). We also extend Lemma \ref{lemma1} to this setup (see Lemma \ref{lemmaSLII} in Appendix \ref{NJMKW12}) and leave the extension of Theorem \ref{lemma3} to this setup as an open problem.

In addition to the degrees of freedom of the photon pulses,
  like polarization, on which Alice and Bob agreed that Alice would
  encode the quantum states, Alice may control further degrees of
  freedom of the photon pulses to her advantage, giving rise to
  various \emph{side-channel attacks} previously considered in the literature of quantum key distribution \cite{XMZLP20}, which can be straightforwardly adapted to
  mistrustful quantum cryptography \dpggg{(e.g. \cite{HBAM18})}. These degrees of freedom may
  include, for example, the light reflected from Bob's setup to
  Alice's setup in Trojan-horse attacks \cite{GFKZR06}; the time at
  which the pulse is sent, the pulse wavelength, polarization, spatial
  degrees of freedom like angle of incidence, etc. in the detection-efficiency-mismatch attacks \cite{MAS06,FTLM09}; the arrival time of
  the pulse in the time-shift attacks \cite{ZFQCL08}; the intensity of
  the pulse in the blinding \cite{M09} and bright-illumination
  \cite{LWWESM10} attacks; the time separation between consecutive
  pulses in the dead-time attacks \cite{LCCLWCLLSLZZCPZCP14}; and the
  mean photon number \cite{SRKBPMLM15}. 
 
More generally, in a side-channel attack, Alice may send any system S in place of
the expected photon pulse.   
For example, S could be a different type of elementary particle, or some
more complex system.   
It is hard to make any
general statement about the resulting measurement statistics at
Bob's detectors.   
At the more exotic end of the range of possibilities, one
can imagine \cite{LC99} Alice sending through the quantum channel miniaturised robots that are programmed to analyse the measurement settings and 
to trigger detections following statistics of Alice's choice.
Although this last possibility may perhaps seem unrealistic, the
limitation is ultimately technological, meaning that a
security proof that excludes such cases cannot strictly be
said to guarantee unconditional security \cite{LC99}. Even if Alice is restricted to photons or particles of small mass, 
it would be very hard for Bob to analyse all the possible
detector statistics and the scope they offer for side-channel attacks.

%Measurement-device independent (MDI) protocols for quantum key distribution (QKD) can remove all detector side-channel attacks \cite{LCQ12}. Although MDI protocols for mistrustful cryptographic tasks (e.g. \cite{}) can also be free of multiphoton and detector side-channel attacks, they can be vulnerable to side-channel attacks to the source and other security problems arising due to the imperfect sources, for example photon sources emitting multiphoton pulses with non-zero probability as discussed above. Fully device-independent QKD protocols offer perfect protection against arbitrary side-channels based on the violation of a Bell inequality \cite{BHK05}. 

Measurement-device independent protocols protect against multiphoton and detector side-channel attacks \cite{LCQ12,ZYWCCGH15}. But they can be vulnerable to side-channel attacks on the source and other security problems arising from imperfect sources -- for example, photon sources emitting multiphoton pulses with non-zero probability, as discussed above. Countermeasures against these attacks have been proposed in QKD (e.g \cite{NPCT21}). These do not in general apply to mistrustful quantum cryptography because Alice is assumed honest in QKD. However, it might be useful to investigate them in the context of mistrustful cryptography.

Fully device-independent protocols provide unconditional security, in principle, from the violation of a Bell inequality  \dpggg{\cite{BHK05,AK15.2,KST20}}. But this opens further challenges, which include closing the locality, detection \cite{Pearle70} and collapse-locality \cite{K05} loopholes. Although the detection and locality loopholes have been simultaneously closed experimentally \cite{HBDRKBRVSAAPMMTEWTH15}, the collapse-locality loophole remains largely untested \cite{K20}.

%The only countermeasure of  which we are aware 
  Another countermeasure
  that provides unconditional security, in
  principle, is for Bob to filter Alice's signals via teleportation,
as suggested in another context by Lo-Chau \cite{LC99}. Indeed, if Bob teleports the quantum state
  encoded in each pulse sent by Alice to a photon entering his setup,
using an ideal teleportation device, 
  he is guaranteed that his setup does not receive anything else than
  a single photon. Thus, if Bob applies reporting strategy III, he is guaranteed from Lemma \ref{lemma1} that Alice
can get negligible information about his bit
  $\beta$ when the dark count probabilities of his detectors are very
  small.   In practice, when preparing pairs of entangled
  photons with parametric down conversion, there is a small but
  nonzero probability of producing pairs with more than one photon. 
Thus, in practice Bob cannot guarantee that pulses with
  more than one photon do not enter his setup. Although the quantum
  states of the multiphoton pulses entering Bob's setup with this
  countermeasure are out of Alice's control, 
Bob must guarantee that Alice cannot exploit this imperfection to
learn some information about $\beta$. We believe this countermeasure deserves further investigation.

An approach to avoid multiphoton and side-channel attacks perfectly is to use protocols not requiring any quantum communication. Some mistrustful cryptographic tasks (e.g. bit commitment \cite{LKBHTWZ15,CCL15} and coin flipping \cite{K99.2,PG21}) can be implemented by purely classical protocols with relativistic constraints, while still satisfying unconditional security -- at least in principle.
   Whether classical or quantum protocols are advantageous in any given context depends on resource costs and other trade-offs.    Our discussion highlights costs and security issues for quantum protocols that need to be taken into account.

%This is possible for some tasks in mistrustful cryptography (e.g. bit commitment \cite{LKBHTWZ15,CCL15} and coin flipping \cite{K99.2,PG21}), using relativistic protocols, while still satisfying unconditional security -- at least in principle. However, relativistic protocols can be vulnerable to other security loopholes. For example, previous experimental demonstration of relativistic cryptography \cite{LKBHTKGWZ13,LCCLWCLLSLZZCPZCP14,LKBHTWZ15,VMHBBZ16,ABCDHSZ21} are vulnerable to attacks in which a dishonest party spoofs the GPS signals and in this way alters the measured spacetime coordinates of the other party.}

We conclude that implementations of theoretically unconditionally secure mistrustful quantum cryptographic protocols need very careful
analysis of practical possibilities for attacks of the type
described.  We hope our analyses of multiphoton and side-channel attacks will
serve as cautionary examples and stimulate further investigations.

\begin{acknowledgments}
A.K. and D.P.G. thank Yang Liu, F\'{e}lix Bussi\`{e}res and Erika Andersson for helpful conversations.
M.B., A.C. and E.D. acknowledge financial support from the French National Research Agency (ANR) project quBIC. A.K. and D.P.G. acknowledge financial support from the UK Quantum Communications Hub grants no. EP/M013472/1 and EP/T001011/1. A.K. is partially supported  by Perimeter Institute for Theoretical Physics. Research at Perimeter Institute is supported by the Government of Canada through Industry Canada and by the Province of Ontario through the Ministry of Research and Innovation.

\textbf{Author contributions}
D.P.G. conceived the project. D.P.G. did the majority of the theoretical work, with input from A.K. The experiment was conceived by M.B., A.C., E.D. and D.P.G. M.B. and A.C. constructed the experimental setup and took the experimental data. D.P.G. analysed the experimental data with input from M.B. and A.C. A.K. and D.P.G. wrote the manuscript with input from M.B., A.C. and E.D.
  \end{acknowledgments}

\appendix

\section{The trivial reporting strategy and the correctness properties}
\label{appA}
\begin{lemma}
\label{lemma4}
Let Alice and Bob implement the practical
protocol. Let $P_{\text{error}\vert c_0c_1\beta}$ be the probability
that Bob obtains the wrong measurement outcome when he measures in the
basis $\mathcal{B}_\beta$ of preparation by Alice and obtains a
detection event $(c_0,c_1)$, for $c_0,c_1,\beta\in\{0,1\}$. If the
photon pulse produces a detection event $(c,c)$, Bob assigns a random
measurement outcome, which we assume gives a wrong outcome with some
probability $P_{\text{error}\vert cc\beta}$ such that
$\frac{1}{2}-\delta_{\text{error}}^{\text{equal}}\leq
P_{\text{error}\vert cc\beta}\leq
\frac{1}{2}+\delta_{\text{error}}^\text{equal}$,
for a small $\delta_{\text{error}}^{\text{equal}}>0$ and for
$c,\beta\in\{0,1\}$. We assume that
$P_{\text{error}\vert c\bar{c}\beta}\leq
\delta_{\text{error}}^{\text{diff}}$,
for a small $\delta_{\text{error}}^{\text{diff}}>0$ and for
$c,\beta\in\{0,1\}$. Let
$P_{\text{det}}(c_0,c_1\vert \beta)$ be the probability that
the pulse sent by Alice produces a detection event $(c_0,c_1)$ in
Bob's detectors when Bob measures in the basis $\mathcal{B}_\beta$,
for $c_0,c_1,\beta\in\{0,1\}$.

If the trivial reporting strategy (\ref{new6}) is implemented with $S=1$ using a \emph{first class} of experimental setups with small probability of loss and high detection efficiencies, satisfying $P_{\text{det}}(c,c\vert \beta)\leq \delta_{cc}^{\text{I}}$, for a small $\delta_{cc}^{\text{I}}>0$ and for $c,\beta\in\{0,1\}$, then the condition (\ref{new13}) of the correctness property is satisfied for any $\delta_{\text{det}}\leq 1$, and the condition (\ref{new10}) of the correctness property is guaranteed if 
\begin{equation}
\label{new9}
(\delta_{00}^{\text{I}}+\delta_{11}^{\text{I}})\Bigl(\frac{1}{2}+\delta_{\text{error}}^{\text{equal}}\Bigr)+\delta_{\text{error}}^{\text{diff}}\leq \delta_\text{error}.
\end{equation}

If a probabilistic reporting strategy (\ref{a4}) is implemented using
a \emph{second class} of experimental setups where Bob's
detection efficiencies are small and the probability of loss is high,
satisfying that
$P_{\text{det}}(0,0\vert \beta)\geq
  1-\delta_{00}^{\text{II}}$,
for a small $\delta_{00}^{\text{II}}>0$ and for
$\beta\in\{0,1\}$, and the conditions (\ref{new13}) and (\ref{new10}) hold with
\begin{eqnarray}
\label{new14}
\delta_{\text{det}}&>& \frac{\delta_{00}^{\text{II}}}{5(1-\delta_{00}^{\text{II}})},\\
\label{new15}
\delta_{\text{error}}&<& \frac{1}{12}-\frac{\delta_{\text{error}}^{\text{equal}}}{6},
\end{eqnarray}
then there exists $(c_0,c_1)\in\{(0,1),(1,0),(1,1)\}$ and $\beta\in\{0,1\}$ satisfying $S_{c_0c_1\beta}>S_{00\beta}$, i.e. the reporting strategy is not the trivial one.
\end{lemma}

\begin{proof}
  We consider the case of the first class of experimental
    setups following the trivial strategy (\ref{new6}) with $S=1$. It
  follows straightforwardly from (\ref{new6}) and (\ref{new7}), since
  $S=1$, that Bob's probability to report Alice's pulse as producing a
  measurement outcome is $P_{\text{det}}=1$.  Thus, the
  condition (\ref{new13}) is satisfied for any
  $\delta_{\text{det}}\leq 1$, as claimed. The probability of error
  $P_{\text{error}}$ is computed for the case in which Bob's
  basis $\mathcal{B}_\beta$ equals Alice's basis and Bob reports the
  pulse as valid ($m=1$), which in this case occurs with unit
  probability.  We have
  \begin{eqnarray}
\label{new8}
P_{\text{error}}&=&\sum_{c_0=0}^1\sum_{c_1=0}^1P_{\text{det}}(c_0,c_1\vert \beta)P_{\text{error}\vert c_0c_1\beta}\nonumber\\
&\leq&(\delta_{00}^{\text{I}}+\delta_{11}^{\text{I}})\Bigl(\frac{1}{2}+\delta_{\text{error}}^{\text{equal}}\Bigr)+\delta_{\text{error}}^{\text{diff}}.
\end{eqnarray}
Thus, the condition (\ref{new10}) of the correctness property is guaranteed if (\ref{new9}) holds, as claimed.

Now we consider the case of the second class of experimental setups that follows a probabilistic reporting strategy. We show that satisfaction of (\ref{new10}) and (\ref{new15}) require
\begin{equation}
\label{new16}
S_{00\beta}< \frac{\delta_{00}^{\text{II}}}{5(1-\delta_{00}^{\text{II}})},
\end{equation}
for $\beta\in\{0,1\}$. Then we show that satisfaction of (\ref{new13}), (\ref{new14}) and (\ref{new16}) imply that there exists $(c_0,c_1)\in\{(0,1),(1,0),(1,1)\}$ and $\beta\in\{0,1\}$ satisfying $S_{c_0c_1\beta}>S_{00\beta}$, i.e. that the reporting strategy is not the trivial one.

We show that (\ref{new16}) follows from (\ref{new10}) and (\ref{new15}). Suppose that (\ref{new16}) does not hold. That is, suppose that there exists $\beta'\in\{0,1\}$ for which 
\begin{equation}
\label{new17}
S_{00\beta'}\geq \frac{\delta_{00}^{\text{II}}}{5(1-\delta_{00}^{\text{II}})}.
\end{equation}
We show from (\ref{new17}) that
\begin{equation}
\label{new18}
P_{\text{error}}\geq \frac{1}{12}-\frac{\delta_{\text{error}}^{\text{equal}}}{6}.
\end{equation}
However, from (\ref{new10}) and (\ref{new15}) we see that (\ref{new18}) cannot hold. Thus, (\ref{new17}) cannot hold either. Thus, (\ref{new16}) holds.

We show (\ref{new18}) from (\ref{new17}). The probability of error $P_{\text{error}}$ is computed for the case in which Bob's basis $\mathcal{B}_\beta$ equals Alice's basis and Bob reports the pulse as valid ($m=1$). We consider the case $\beta=\beta'$. We have
\begin{eqnarray}
\label{new11}
P_{\text{error}}&=&\frac{\sum_{c_0,c_1}S_{c_0c_1\beta'}P_{\text{det}}(c_0,c_1\vert \beta')P_{\text{error}\vert c_0c_1\beta'}}{\sum_{c_0,c_1}S_{c_0c_1\beta'}P_{\text{det}}(c_0,c_1\vert \beta')}\nonumber\\
&\geq&\frac{S_{00\beta'}P_{\text{det}}(0,0\vert \beta')P_{\text{error}\vert 00\beta'}}{\sum_{c_0,c_1}S_{c_0c_1\beta'}P_{\text{det}}(c_0,c_1\vert \beta')}\nonumber\\
&\geq&\frac{P_{\text{error}\vert 00\beta'}}{1+\frac{\sum_{(c_0,c_1)\neq(0,0)}S_{c_0c_1\beta'}P_{\text{det}}(c_0,c_1\vert \beta')}{S_{00\beta'}P_{\text{det}}(00\vert \beta')}}\nonumber\\
&\geq&\frac{P_{\text{error}\vert 00\beta'}}{1+\frac{\sum_{(c_0,c_1)\neq(0,0)}P_{\text{det}}(c_0,c_1\vert \beta')}{S_{00\beta'}P_{\text{det}}(00\vert \beta')}}\nonumber\\
&\geq&\frac{\frac{1}{2}-\delta_{\text{error}}^{\text{equal}}}{1+\frac{\delta_{00}^{\text{II}}}{S_{00\beta'}(1-\delta_{00}^{\text{II}})}}\nonumber\\
&\geq&\frac{1}{6}\Bigl(\frac{1}{2}-\delta_{\text{error}}^{\text{equal}}\Bigr),
\end{eqnarray}
which is the claimed bound (\ref{new18}), where $c_0$ and $c_1$ run over $\{0,1\}$, where in the fourth line we used $S_{c_0c_1\beta'}\leq 1$ for $c_0,c_1\in\{0,1\}$, in the fifth line we used that $P_{\text{error}\vert cc\beta}\geq \frac{1}{2}-\delta_{\text{error}}^\text{equal}$, $\sum_{(c_0,c_1)\neq(0,0)}P_{\text{det}}(c_0,c_1\vert \beta')=1-P_{\text{det}}(0,0\vert \beta')$ and $P_{\text{det}}(0,0\vert \beta')\geq1-\delta_{00}^{\text{II}}$, and in the last line we used (\ref{new17}).

We show that satisfaction of (\ref{new13}), (\ref{new14}) and (\ref{new16}) imply that there exists $(c_0,c_1)\in\{(0,1),(1,0),(1,1)\}$ and $\beta\in\{0,1\}$ satisfying $S_{c_0c_1\beta}>S_{00\beta}$. We suppose that this does not hold, i.e. that
\begin{equation}
\label{new19}
S_{c_0c_1\beta}\leq S_{00\beta},
\end{equation}
for $c_0,c_1,\beta\in\{0,1\}$. We show that this implies a contradiction. From (\ref{new16}) and (\ref{new19}), we show that the probability that Bob reports Alice's photon pulse as activating a measurement outcome  $P_{\text{det}}$ satisfies
\begin{equation}
\label{new20}
P_{\text{det}}< \frac{\delta_{00}^{\text{II}}}{5(1-\delta_{00}^{\text{II}})}.
\end{equation}
However, from (\ref{new13}) and (\ref{new14}), we see that (\ref{new20}) cannot hold. Thus, (\ref{new19}) does not hold, as claimed.

We show (\ref{new20}) from (\ref{new16}) and (\ref{new19}). We have
\begin{eqnarray}
\label{new21}
P_{\text{det}}&=&\sum_{c_0,c_1,\beta}P(\beta)S_{c_0c_1\beta}P_{\text{det}}(c_0,c_1\vert \beta)\nonumber\\
&\leq&\sum_\beta P(\beta)S_{00\beta} \sum_{c_0,c_1}P_{\text{det}}(c_0,c_1\vert \beta)\nonumber\\
&\leq& \max_{\beta\in\{0,1\}}\{S_{00\beta}\}\nonumber\\
&<& \frac{\delta_{00}^{\text{II}}}{5(1-\delta_{00}^{\text{II}})},
         \end{eqnarray}
       where $c_0,c_1$ and $\beta$ run over $\{0,1\}$, where in the
       second
 line we used (\ref{new19}), in the third line we used that
 $\sum_{c_0,c_1}P_{\text{det}}(c_0,c_1\vert \beta)=1$ and
 $\sum_\beta P(\beta)=1$, 
and in the last line we used (\ref{new16}).
\end{proof}

The second class of experimental setups includes
implementations with weak coherent sources, for which
$P_{\text{det}}(0,0\vert \beta)\geq
  1-\delta_{00}^{\text{II}}$,
for a small $\delta_{00}^{\text{II}}>0$ and for
$\beta\in\{0,1\}$.   For example, the experimental
  demonstrations of mistrustful quantum cryptography given by
  Refs. \cite{LKBHTKGWZ13,LCCLWCLLSLZZCPZCP14,PJLCLTKD14} 
  used weak coherent sources with $\mu \approx 0.05,
  0.183, 0.005$ respectively.   For $\mu=0.05$, the
  probability of a zero photon pulse is $p_0=e^{-\mu} \approx 0.95$. Thus, due to losses in the quantum
channel and Bob's detection efficiencies being smaller than unity, we
have $\delta_{00}^{\text{II}}<0.05$ in this case. The
condition (\ref{new14}) is motivated by the fact that if
$ P_{\text{det}}(0,0\vert \beta)\lesssim
  1-\delta_{00}^{\text{II}}$
then $P_{\text{det}}\gtrsim \delta_{00}^{\text{II}}$, hence
$P_\text{det}>\delta_{\text{det}}$ with
$\delta_{\text{det}}$ satisfying (\ref{new14}) is a suitable required
lower bound for $P_\text{det}$. Similarly, the
  experimental demonstrations of mistrustful quantum cryptography
  given by
  Refs. \cite{NJMKW12,LKBHTKGWZ13,LCCLWCLLSLZZCPZCP14,ENGLWW14} report
  the required upper bounds $P_\text{error}<0.05$, $P_\text{error}<0.015$, $P_\text{error}<0.046$ and
  $P_\text{error}<0.01$, respectively, in order to guarantee security
  against Bob. Thus, we see that in
  Refs. \cite{NJMKW12,LKBHTKGWZ13,LCCLWCLLSLZZCPZCP14,ENGLWW14},
  $P_\text{error}<\frac{1}{12}$ and
 we have $P_\text{error}<\delta_{\text{error}}$ with
  $\delta_{\text{error}}$ satisfying (\ref{new15}).

\section{Security bounds}
\label{appB}
The following lemmas give upper bounds on the amount of information
that Alice can obtain about $\beta$ from Bob's message $m$. Lemma
\ref{lemma5} provides a bound for an arbitrary reporting strategy by
Bob. Together with Lemma \ref{lemma5}, Lemmas \ref{lemma6} and
\ref{lemmanewbound} provide bounds when Bob follows reporting
strategies II and III, respectively.

\begin{lemma}
\label{lemma5}
Suppose that Alice sends Bob a pulse of $k$ photons, encoding an
arbitrary $k-$qubit state $\rho$, which may be arbitrarily entangled
and which may be entangled with an ancilla held by Alice. Suppose
that Bob chooses the bit $\beta$ with some probability
\begin{equation}
\label{a46}
P_{\text{basis}}(\beta)\leq \frac{1}{2}+\epsilon_{\text{basis}},
\end{equation}
for $\beta\in\{0,1\}$ and for some $ \epsilon_{\text{basis}}\geq
0$.
Consider an arbitrary strategy by Bob to report the message
$m\in\{0,1\}$ to Alice. Then Alice's probability $P_{\text{guess}}$ to
guess $\beta$ from Bob's message $m$ satisfies
\begin{eqnarray}
\label{a43}
P_{\text{guess}}&\leq& \Bigl(\frac{1}{2}+\epsilon_{\text{basis}}\Bigr)\Bigl(1+\bigl\lvert P_{\text{report}}(1\lvert 1, \rho, k)\nonumber\\
&&\qquad\qquad\qquad-P_{\text{report}}(1\lvert 0, \rho, k)\bigr\rvert\Bigr).
\end{eqnarray}
\end{lemma}

\begin{proof}
Alice's most general strategy to guess $\beta$ from the message $m$ is as follows. Alice guesses $\beta=i$ with some probability $P_{\text{Alice}}(i\vert m)$, when she receives the message $m$, for $i,m\in\{0,1\}$. Alice's optimal strategy is a deterministic strategy. That is, one of the two following equations holds:
\begin{eqnarray}
\label{det1}
P_{\text{Alice}}(i\vert m)&=&\delta_{i,m},\\
\label{det2}
P_{\text{Alice}}(i\vert m)&=&\delta_{\bar{i},m},
\end{eqnarray}
for $i,m\in\{0,1\}$. Alice's average probability to guess $\beta$ is given by
\begin{equation}
\label{a45}
P_{\text{guess}}=\sum_{\beta=0}^1\sum_{m=0}^1P_{\text{basis}}(\beta)P_{\text{Alice}}(\beta\vert m)P_{\text{report}}(m\lvert \beta, \rho, k).
\end{equation}

If (\ref{det1}) holds, it follows from (\ref{a46}), (\ref{det1}) and (\ref{a45}) that
\begin{equation}
\label{det3}
P_{\text{guess}}\leq \Bigl(\frac{1}{2}+\epsilon_{\text{basis}}\Bigr)\sum_{\beta=0}^1 P_{\text{report}}(\beta\lvert \beta, \rho, k).
\end{equation}
If (\ref{det2}) holds, it follows from (\ref{a46}), (\ref{det2}) and (\ref{a45}) that
\begin{equation}
\label{det4}
P_{\text{guess}}\leq \Bigl(\frac{1}{2}+\epsilon_{\text{basis}}\Bigr)\sum_{\beta=0}^1 P_{\text{report}}(\bar{\beta}\lvert \beta, \rho, k).
\end{equation}
Since 
\begin{equation}
\label{det5}
P_{\text{report}}(\bar{m}\lvert \beta, \rho, k) =1-P_{\text{report}}(m\lvert \beta, \rho, k),
\end{equation}
for $\beta,m\in\{0,1\}$, we obtain that
\begin{eqnarray}
\label{det6}
\sum_{\beta=0}^1 P_{\text{report}}(\beta\lvert \beta, \rho, k)&=&1+P_{\text{report}}(1\lvert 1, \rho, k)\nonumber\\
&&\qquad\!\!-P_{\text{report}}(1\lvert 0, \rho, k)\nonumber\\
&\leq&1+\bigl\lvert P_{\text{report}}(1\lvert 1, \rho, k)\nonumber\\
&&\qquad\!\! -P_{\text{report}}(1\lvert 0, \rho, k)\bigr\rvert.
\end{eqnarray}
Similarly, from (\ref{det5}), we obtain
\begin{eqnarray}
\label{det7}
\sum_{\beta=0}^1 P_{\text{report}}(\bar{\beta}\lvert \beta, \rho, k)&=&1+P_{\text{report}}(1\lvert 0, \rho, k)\nonumber\\
&&\qquad\!\!\!\!-P_{\text{report}}(1\lvert 1, \rho, k)\nonumber\\
&\leq&1+\bigl\lvert P_{\text{report}}(1\lvert 1, \rho, k)\nonumber\\
&&\qquad\!\!\!\!\!\!\!-P_{\text{report}}(1\lvert 0, \rho, k)\bigr\rvert.
\end{eqnarray}
Thus, from (\ref{det3}), (\ref{det4}), (\ref{det6}) and (\ref{det7}), the claimed result (\ref{a43}) follows.
\end{proof}

\begin{lemma}
\label{lemma6}
Let $\eta_{\text{low}}$ and $\eta_{\text{up}}$ be such that
\begin{equation}
\label{qqqxyza54}
0<\eta_{\text{low}}\leq \eta_{i\beta}\leq \eta_{\text{up}}<1,
\end{equation}
for $i,\beta\in\{0,1\}$. Let $0\leq d_{i\beta}\leq\delta$, for $i,\beta\in\{0,1\}$ and for some $0\leq\delta<\frac{1}{2}$. We define
\begin{equation}
\label{qqqxyza56}
B_\text{det}^{\text{II}}\equiv  (1-2\delta)\biggl(\frac{\ln{\bigl(1-\eta_{\text{up}}\bigr)}}{\ln{\bigl(1-\eta_{\text{low}}\bigr)}}
-1\biggr)\bigl(1-\eta_{\text{up}}\bigr)^{B_{\text{exp}}^{\text{II}}},
\end{equation}
where
\begin{equation}
\label{qqqxyza57}
B_{\text{exp}}^{\text{II}}\equiv{\frac{\ln{\Bigl((1-2\delta)\ln{\bigl(1-\eta_{\text{up}}\bigr)}/ \ln{\bigl(1-\eta_{\text{low}}\bigr)}\Bigr)}}{\ln{\bigl(1-\eta_{\text{low}}\bigr)}-\ln{\bigl(1-\eta_{\text{up}}\bigr)}}}.
\end{equation}
Suppose that Alice sends Bob a pulse of $k$ photons, encoding an arbitrary $k-$qubit state $\rho$, which may be arbitrarily entangled and which may be entangled with an ancilla held by Alice. Suppose also that Bob uses reporting strategy II. Then
\begin{equation}
\label{qqqxyza58}
\bigl\lvert P_{\text{report}}(1\lvert 1, \rho, k)- P_{\text{report}}(1\lvert 0, \rho, k)\bigr\rvert\leq B_\text{II},
\end{equation}
for $k\in\{0,1,2,\ldots\}$, where
\begin{equation}
\label{qqqxyza59}
B_\text{II}\equiv\max\bigl\{2\delta,B_\text{det}^\text{II}\bigr\}.
\end{equation}
\end{lemma}

\begin{lemma}
\label{lemmanewbound}
Let $\eta_{\text{min}}$ and $\eta_{\text{max}}$ be such that
\begin{eqnarray}
\label{xyza54}
\eta_{\text{min}}&\equiv&\min_{i,\beta\in\{0,1\}}\{\eta_{i\beta}\}>0,\nonumber\\
\eta_{\text{max}}&\equiv&\max_{i,\beta\in\{0,1\}}\{\eta_{i\beta}\}<1.
\end{eqnarray}
Let $0\leq d_{i\beta}\leq\delta$, for $i,\beta\in\{0,1\}$ and for some $0\leq\delta<\frac{1}{2}$. We define
\begin{equation}
\label{xyza56}
B_\text{det}^{\text{III}}\equiv  1-\frac{\eta_\text{min}}{\eta_\text{max}}+(1-2\delta)\biggl(\frac{\ln{\bigl(1-\eta_{\text{max}}\bigr)}}{\ln{\bigl(1-\eta_{\text{min}}\bigr)}}
-1\biggr)\bigl(1-\eta_{\text{max}}\bigr)^{B_{\text{exp}}^{\text{III}}},
\end{equation}
where
\begin{equation}
\label{xyza57}
B_{\text{exp}}^{\text{III}}\equiv{\frac{\ln{\Bigl((1-2\delta)\eta_\text{max}\ln{\bigl(1-\eta_{\text{max}}\bigr)}/ \eta_\text{min}\ln{\bigl(1-\eta_{\text{min}}\bigr)}\Bigr)}}{\ln{\bigl(1-\eta_{\text{min}}\bigr)}-\ln{\bigl(1-\eta_{\text{max}}\bigr)}}}.
\end{equation}
Suppose that Alice sends Bob a pulse of $k$ photons, encoding an arbitrary $k-$qubit state $\rho$, which may be arbitrarily entangled and which may be entangled with an ancilla held by Alice. Suppose also that Bob uses reporting strategy III with
\begin{equation}
\label{xyza100}
\frac{\eta_\text{min}}{\eta_\text{max}}\leq S_{11\beta}\leq 1,
\end{equation}
for $\beta\in\{0,1\}$. Then
\begin{equation}
\label{xyza58}
\bigl\lvert P_{\text{report}}(1\lvert 1, \rho, k)- P_{\text{report}}(1\lvert 0, \rho, k)\bigr\rvert\leq B_\text{III},
\end{equation}
for $k\in\{0,1,2,\ldots\}$, where
\begin{equation}
\label{xyza59}
B_\text{III}\equiv\max\bigl\{2\delta,B_\text{det}^{\text{III}}\bigr\}.
\end{equation}
\end{lemma}

We note that the bounds $B_\text{II}$ and $B_\text{III}$ in Lemmas \ref{lemma6} and \ref{lemmanewbound}, respectively, tend to
$2\delta$ if all the detection efficiencies tend to the same
value $\eta$, as in this case we have
$\eta_{\text{up}}\rightarrow\eta$, $\eta_{\text{low}}\rightarrow \eta$, $\eta_{\text{max}}\rightarrow\eta$ and
$\eta_{\text{min}}\rightarrow \eta$. Additionally, if the dark count probabilities are
zero then $\delta=0$ and the bounds tend to zero in this
case. This is expected. However, as we illustrate in
Fig. \ref{figbound}, even for relatively close values for the
detection efficiencies, the obtained bounds are not very small.

The proofs of Lemmas \ref{lemma6} and \ref{lemmanewbound} use the following lemma.

\begin{figure}
\includegraphics[width=0.48\textwidth]{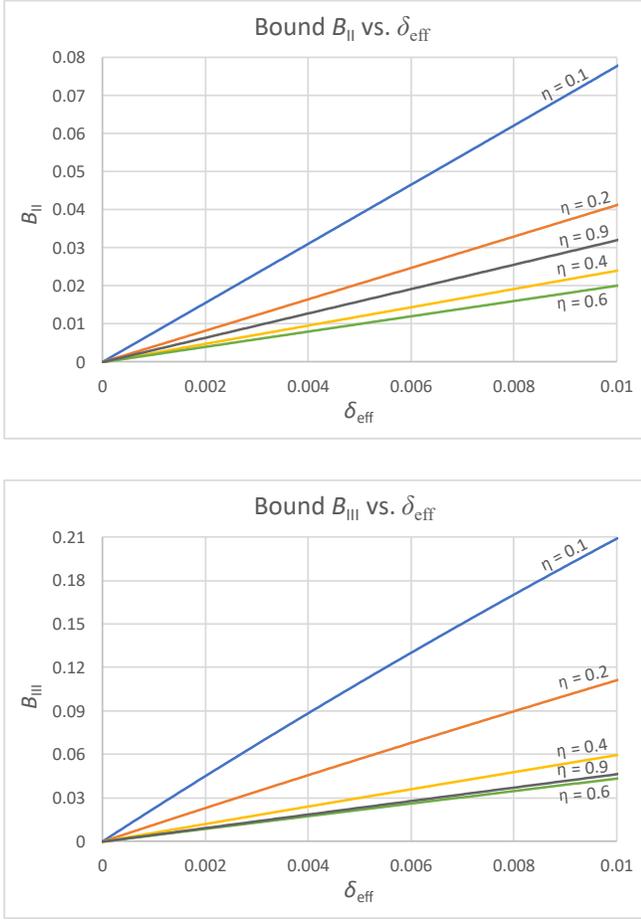}
\caption{\label{figbound} \textbf{Security bounds $B_\text{II}$ and $B_\text{III}$.} We show plots for the bounds $B_\text{II}$ (top) and $B_\text{III}$ (bottom), defined by (\ref{qqqxyza59}) in Lemma \ref{lemma6} and by (\ref{xyza59}) in Lemma \ref{lemmanewbound}, respectively. We consider the case $\delta=10^{-5}$ in Lemmas \ref{lemma6} and \ref{lemmanewbound}. We also consider $\eta_{\text{low}}=\eta-\delta_{\text{eff}}$ and $\eta_{\text{up}}=\eta+\delta_{\text{eff}}$ in Lemma \ref{lemma6}; and $\eta_{\text{min}}=\eta-\delta_{\text{eff}}$ and $\eta_{\text{max}}=\eta+\delta_{\text{eff}}$ in Lemma \ref{lemmanewbound}. We plot for five values of $\eta$ and for the range $\delta_{\text{eff}}\in[0,0.01]$.}
\end{figure}

\begin{lemma}
\label{lemmanew}
Suppose that Alice sends Bob a pulse of $k$ photons, encoding an arbitrary $k-$qubit state $\rho$, which may be arbitrarily entangled and which may be entangled with an ancilla held by Alice. Let $\eta_{\text{low}}$ and $\eta_{\text{up}}$ be such that
\begin{equation}
\label{newxyza54}
0<\eta_{\text{low}}\leq \eta_{i\beta}\leq\eta_{\text{up}}<1,
\end{equation}
for $i,\beta\in\{0,1\}$. Let $0\leq d_{i\beta}\leq \delta$, for $i,\beta\in\{0,1\}$, and for some $0\leq \delta<\frac{1}{2}$. Then
\begin{equation}
\label{rrra64}
(1-2\delta)(1-\eta_\text{up})^k\leq P_{\text{det}}(0,0\lvert \beta, \rho, k)\leq (1-\eta_\text{low})^k,
\end{equation}
for $k\in\{0,1,2,\ldots\}$ and $\beta\in\{0,1\}$.
\end{lemma}

\begin{proof}
Consider an arbitrary, and possibly entangled, $k-$qubit state $\rho$ encoded in a pulse of $k$ photons that Alice sends Bob, which may be in an arbitrary entangled state with an ancilla held by Alice. Let $P(k_0\vert\beta,\rho,k)$ be the probability that $k_0$ photons go to the detector $D_0$ and $k-k_0$ photons go to the detector $D_1$, when Bob measures in the basis $\mathcal{B}_\beta$, for $\beta\in\{0,1\}$, $k_0\in\{0,1,\ldots,k\}$ and $k\in\{0,1,2,\ldots\}$. It follows straightforwardly that
\begin{eqnarray}
\label{rrra62}
P_{\text{det}}(0,0\lvert \beta, \rho, k)&=&(1-d_{0\beta})(1-d_{1\beta})\sum_{k_0=0}^k P(k_0\vert\beta,\rho,k)\times\nonumber\\
&&\quad\times\bigl[(1-\eta_{0\beta})^{k_0} (1-\eta_{1\beta})^{k-k_0}\bigr]\nonumber\\
&\leq&(1-\eta_\text{low})^k\sum_{k_0=0}^k P(k_0\vert\beta,\rho,k)\nonumber\\
&=&(1-\eta_\text{low})^k,
\end{eqnarray}
for any $k-$qubit state $\rho$, for $k\in\{0,1,2,\ldots\}$, and for $\beta\in\{0,1\}$; where in the second line we used that $\eta_{i\beta}\geq \eta_\text{low}$ and $d_{i\beta}\geq 0$, for $i,\beta\in\{0,1\}$; and in the third line we used that $\sum_{k_0=0}^k P(k_0\vert\beta,\rho,k)=1$. Similarly, using that $\eta_{i\beta}\leq \eta_\text{up}$ and $d_{i\beta}\leq \delta$, it follows straightforwardly that
\begin{equation}
\label{rrra63}
P_{\text{det}}(0,0\lvert \beta, \rho, k)\geq (1-2\delta)(1-\eta_\text{up})^k,
\end{equation}
for any $k-$qubit state $\rho$, for $k\in\{0,1,2,\ldots\}$, and for $\beta\in\{0,1\}$. Thus, from (\ref{rrra62}) and (\ref{rrra63}), we obtain (\ref{rrra64}).

\end{proof}

\begin{proof}[Proof of Lemma \ref{lemma6}]
Consider an arbitrary, and possibly entangled, $k-$qubit state $\rho$ encoded in a pulse of $k$ photons that Alice sends Bob, which may be in an arbitrary entangled state with an ancilla held by Alice. In the reporting strategy II, the probability that Bob reports the message $m=1$ to Alice is the probability that at least one of his detectors click. That is,
\begin{equation}
\label{qqqa60}
P_{\text{report}}(1\lvert \beta, \rho, k)=1-P_{\text{det}}(0,0\lvert \beta, \rho, k),
\end{equation}
for $k\in\{0,1,2,\ldots\}$ and $\beta\in\{0,1\}$. Thus, showing (\ref{qqqxyza58}) reduces to showing
\begin{equation}
\label{qqqa61}
\bigl\lvert P_{\text{det}}(0,0\lvert 1, \rho, k)- P_{\text{det}}(0,0\lvert 0, \rho, k)\bigr\rvert\leq\max\{2\delta, B_{\text{det}}^{\text{II}}\},\\
\end{equation}
for $k\in\{0,1,2,\ldots\}$, where $B_{\text{det}}^{\text{II}}$ is given by (\ref{qqqxyza56}).

We show (\ref{qqqa61}). From Lemma \ref{lemmanew}, we obtain
\begin{equation}
\label{qqqrrr4}
-f(k)\leq P_{\text{det}}(0,0\lvert 1, \rho, k)- P_{\text{det}}(0,0\lvert 0, \rho, k)\leq f(k),
\end{equation}
where we define
\begin{equation}
\label{qqqrrr5}
f(k)\equiv (1-\eta_\text{low})^k-(1-2\delta)(1-\eta_\text{up})^k,
\end{equation}
for $k\in\mathbb{R}$. It follows from (\ref{qqqrrr4}) that
\begin{equation}
\label{qqqrrr6}
\bigl\lvert P_{\text{det}}(0,0\lvert 1, \rho, k)- P_{\text{det}}(0,0\lvert 0, \rho, k)\bigr\rvert\leq f(k),
\end{equation}
for $k\in\{0,1,2,\ldots\}$. If $\eta_\text{low}=\eta_\text{up}$, it follows from (\ref{qqqxyza54}) and (\ref{qqqrrr5}), and from $\delta\geq 0$, that 
\begin{equation}
\label{qqqrrr7}
f(k)\leq 2\delta,
\end{equation}
for $k\in\{0,1,2,\ldots\}$. We show below that if $\eta_\text{low}<\eta_\text{up}$ then
\begin{equation}
\label{qqqrrr9}
 f(k)\leq f(B_\text{exp}^{\text{II}}),
\end{equation}
for $k\in\{0,1,2,\ldots\}$, where $B_\text{exp}^{\text{II}}$ is given by (\ref{qqqxyza57}). It follows straightforwardly from (\ref{qqqxyza56}), (\ref{qqqxyza57}) and (\ref{qqqrrr5}) that
\begin{equation}
\label{qqqrrr10}
f(B_\text{exp}^{\text{II}})=B_\text{det}^{\text{II}},
\end{equation}
where $B_\text{det}^{\text{II}}$ is given by (\ref{qqqxyza56}). Thus, (\ref{qqqa61}) follows from (\ref{qqqrrr6}) -- (\ref{qqqrrr10}).

We assume that $\eta_\text{low}<\eta_\text{up}$ and show (\ref{qqqrrr9}). We define $k^*\in\mathbb{R}$ such that 
\begin{equation}
\label{qqqrrr11}
f'(k^*)=0.
\end{equation}
We show below that 
\begin{equation}
\label{qqqrrr12}
k^*=B_\text{exp}^{\text{II}}.
\end{equation}
We also show that
\begin{equation}
\label{qqqrrr13}
f''(k^*)<0.
\end{equation}
Thus, $f(k^*)$ is a maximum and
\begin{equation}
\label{qqqrrr14}
f(k)\leq f(k^*),
\end{equation}
for $k\in\{0,1,2,\ldots\}$. Therefore, (\ref{qqqrrr9}) follows from (\ref{qqqrrr12}) and (\ref{qqqrrr14}).

We show (\ref{qqqrrr12}). From (\ref{qqqrrr5}), we have
\begin{eqnarray}
\label{qqqrrr15}
f'(k)&=&\ln(1-\eta_\text{low})(1-\eta_\text{low})^{k}\nonumber\\
&&\quad\!\!-(1-2\delta)\ln(1-\eta_\text{up})(1-\eta_\text{up})^{k},
\end{eqnarray}
for $k\in\mathbb{R}$. Thus, (\ref{qqqrrr12}) follows straightforwardly from (\ref{qqqxyza57}), (\ref{qqqrrr11}) and (\ref{qqqrrr15}).

We complete the proof by showing (\ref{qqqrrr13}). From (\ref{qqqrrr15}), we have
\begin{eqnarray}
\label{qqqrrr16}
f''(k^*)&=&\bigl(\ln(1-\eta_\text{low})\bigr)^2(1-\eta_\text{low})^{k^*}\nonumber\\
&&\quad\!\!-(1-2\delta)\bigl(\ln(1-\eta_\text{up})\bigr)^2(1-\eta_\text{up})^{k^*},\nonumber\\
&=&\bigl(\ln(1-\eta_\text{low})-\ln(1-\eta_\text{up})\bigr)\times\nonumber\\
&&\quad\times (1-2\delta)\ln(1-\eta_\text{up})(1-\eta_\text{up})^{k^*}\nonumber\\
&<&0,
\end{eqnarray}
where in the second line we used (\ref{qqqrrr11}) and (\ref{qqqrrr15}); and in the third line we used (\ref{qqqxyza54}), our assumption that $\eta_\text{low}<\eta_\text{up}$, and that $0\leq\delta<\frac{1}{2}$.
\end{proof}

\begin{proof}[Proof of Lemma \ref{lemmanewbound}]
From the definition of reporting strategy III, given by (\ref{a5}), and from (\ref{a4}),  we have
\begin{eqnarray}
\label{rrr1}
&&P_\text{report}(1\vert 1,\rho,k)-P_\text{report}(1\vert 0,\rho,k)\nonumber\\
&&\quad=\frac{\eta_\text{min}}{\eta_{11}}P_\text{det}(0,1\vert 1,\rho,k)+\frac{\eta_\text{min}}{\eta_{01}}P_\text{det}(1,0\vert 1,\rho,k)\nonumber\\
&&\qquad+S_{111}P_\text{det}(1,1\vert 1,\rho,k)-\frac{\eta_\text{min}}{\eta_{10}}P_\text{det}(0,1\vert 0,\rho,k)\nonumber\\
&&\qquad-\frac{\eta_\text{min}}{\eta_{00}}P_\text{det}(1,0\vert 0,\rho,k)-S_{110}P_\text{det}(1,1\vert 0,\rho,k).\nonumber\\
\end{eqnarray}
From (\ref{xyza100}) and (\ref{rrr1}), we obtain
\begin{eqnarray}
\label{rrr2}
&&P_\text{report}(1\vert 1,\rho,k)-P_\text{report}(1\vert 0,\rho,k)\nonumber\\
&&\quad\leq 1-P_\text{det}(0,0\vert 1,\rho,k)-\frac{\eta_\text{min}}{\eta_\text{max}}\bigl(1-P_\text{det}(0,0\vert 0,\rho,k)\bigr),\nonumber\\
\end{eqnarray}
and
\begin{eqnarray}
\label{rrr3}
&&P_\text{report}(1\vert 1,\rho,k)-P_\text{report}(1\vert 0,\rho,k)\nonumber\\
&&\quad\geq\frac{\eta_\text{min}}{\eta_\text{max}}\bigl(1-P_\text{det}(0,0\vert 1,\rho,k)\bigr) -\bigl(1-P_\text{det}(0,0\vert 0,\rho,k)\bigr).\nonumber\\
\end{eqnarray}

Thus, from (\ref{rrr2}) and (\ref{rrr3}), and from Lemma \ref{lemmanew}, we obtain
\begin{equation}
\label{rrr4}
-g(k)\leq P_\text{report}(1\vert 1,\rho,k)-P_\text{report}(1\vert 0,\rho,k)\leq g(k),
\end{equation}
where we define
\begin{equation}
\label{rrr5}
g(k)\equiv 1-\frac{\eta_\text{min}}{\eta_\text{max}}+\frac{\eta_\text{min}}{\eta_\text{max}}(1-\eta_\text{min})^k-(1-2\delta)(1-\eta_\text{max})^k,
\end{equation}
for $k\in\mathbb{R}$. It follows from (\ref{rrr4}) that
\begin{equation}
\label{rrr6}
\bigl\lvert P_\text{report}(1\vert 1,\rho,k)-P_\text{report}(1\vert 0,\rho,k)\bigr\rvert\leq g(k),
\end{equation}
for $k\in\{0,1,2,\ldots\}$. 

If $\eta_\text{min}=\eta_\text{max}$, it follows from (\ref{xyza54}) and (\ref{rrr5}), and from $\delta\geq 0$, that 
\begin{equation}
\label{rrr7}
g(k)\leq 2\delta,
\end{equation}
for $k\in\{0,1,2,\ldots\}$. We show below that if $\eta_\text{min}<\eta_\text{max}$ then
\begin{equation}
\label{rrr9}
 g(k)\leq g(B_\text{exp}^{\text{III}}),
\end{equation}
for $k\in\{0,1,2,\ldots\}$, where $B_\text{exp}^{\text{III}}$ is given by (\ref{xyza57}). It follows straightforwardly from (\ref{xyza56}), (\ref{xyza57}) and (\ref{rrr5}) that
\begin{equation}
\label{rrr10}
g(B_\text{exp}^{\text{III}})=B_\text{det}^{\text{III}},
\end{equation}
where $B_\text{det}^{\text{III}}$ is given by (\ref{xyza56}). Thus, the claimed bound (\ref{xyza58}) follows from (\ref{rrr6}) -- (\ref{rrr10}).

We assume that $\eta_\text{min}<\eta_\text{max}$ and show (\ref{rrr9}). We define $\tilde{k}\in\mathbb{R}$ such that 
\begin{equation}
\label{rrr11}
g'(\tilde{k})=0.
\end{equation}
We show below that 
\begin{equation}
\label{rrr12}
\tilde{k}=B_\text{exp}^{\text{III}}.
\end{equation}
We also show that
\begin{equation}
\label{rrr13}
g''(\tilde{k})<0.
\end{equation}
Thus, $g(\tilde{k})$ is a maximum and
\begin{equation}
\label{rrr14}
g(k)\leq g(\tilde{k}),
\end{equation}
for $k\in\{0,1,2,\ldots\}$. Therefore, (\ref{rrr9}) follows from (\ref{rrr12}) and (\ref{rrr14}).

We show (\ref{rrr12}). From (\ref{rrr5}), we have
\begin{eqnarray}
\label{rrr15}
g'(k)&=&\frac{\eta_\text{min}}{\eta_\text{max}}\ln(1-\eta_\text{min})(1-\eta_\text{min})^{k}\nonumber\\
&&\quad\!\!-(1-2\delta)\ln(1-\eta_\text{max})(1-\eta_\text{max})^{k},
\end{eqnarray}
for $k\in\mathbb{R}$. Thus, (\ref{rrr12}) follows straightforwardly from (\ref{xyza57}), (\ref{rrr11}) and (\ref{rrr15}).

We complete the proof by showing (\ref{rrr13}). From (\ref{rrr15}), we have
\begin{eqnarray}
\label{rrr16}
g''(\tilde{k})&=&\frac{\eta_\text{min}}{\eta_\text{max}}\bigl(\ln(1-\eta_\text{min})\bigr)^2(1-\eta_\text{min})^{\tilde{k}}\nonumber\\
&&\quad\!\!-(1-2\delta)\bigl(\ln(1-\eta_\text{max})\bigr)^2(1-\eta_\text{max})^{\tilde{k}},\nonumber\\
&=&\bigl(\ln(1-\eta_\text{min})-\ln(1-\eta_\text{max})\bigr)\times\nonumber\\
&&\quad\times (1-2\delta)\ln(1-\eta_\text{max})(1-\eta_\text{max})^{\tilde{k}}\nonumber\\
&<&0,
\end{eqnarray}
where in the second line we used (\ref{rrr11}) and (\ref{rrr15}); and in the third line we used (\ref{xyza54}), our assumption that $\eta_\text{min}<\eta_\text{max}$, and that $0\leq\delta<\frac{1}{2}$.
\end{proof}

\section{Mistrustful quantum cryptography}
\label{appC}
In mistrustful cryptography, two or more parties who do not trust each
other collaborate to implement a cryptographic task. Important
cryptographic tasks in mistrustful cryptography are bit commitment,
oblivious transfer, secure multi-party computation and coin flipping,
for example.  We say that a cryptographic protocol is
\emph{unconditionally secure} if it is secure based only on the laws
of physics, without imposing any technological limitations on the
dishonest parties. There exist quantum and relativistic protocols in
mistrustful cryptography that exploit the laws of quantum physics and
relativity to guarantee security, respectively. For some tasks in
mistrustful cryptography, there exists some impossibility results
stating that some tasks in mistrustful cryptography cannot achieve
unconditional security with quantum nonrelativistic protocols
\cite{M97,LC97,L97,LC98}, or even with quantum relativistic protocols
for other tasks \cite{C07}. On the other hand, there are relativistic
protocols that achieve unconditional security for some tasks
\cite{K99,K99.2,K05.2,CK06,C07,LKBHTWZ15,K11.2,K11.3,K12,AK15.1,AK15.2,PG15.1,PGK18,PG19}. However,
by imposing technological limitations on the dishonest parties,
security of some quantum nonrelativistic protocols can be
guaranteed. For example, some tasks that cannot achieve unconditional
security can be implemented securely in the noisy storage model
\cite{DFSS08,WST08} in which the dishonest parties can only store
quantum states in noisy quantum memories with finite coherence times.

\subsection{Quantum bit commitment}
\label{QBC}

In bit commitment, Bob (the committing party) commits a secret bit $b$
to Alice at a given time $t_\text{commit}$. Bob chooses to unveil $b$
to Alice at some time $t_\text{unveil}>t_\text{commit}$. A bit
commitment protocol must satisfy two security conditions. First, a bit
commitment protocol is \emph{hiding} if, when Bob follows the protocol
and Alice deviates arbitrarily from the protocol, the probability that
Alice guesses Bob's bit $b$ before Bob unveils satisfies
$P_\text{guess}\leq \frac{1}{2}+\epsilon_\text{hiding}$, for some
$\epsilon_\text{hiding}\geq 0$ that goes to zero as some security
parameter goes to infinity. Second, a bit commitment protocol is
\emph{binding} if, when Alice follows the protocol and Bob deviates
arbitrarily from the protocol, the probability $p_i$ that Bob
successfully unveils the bit $b=i$ satisfies
$p_0+p_1\leq 1+\epsilon_\text{binding}$, for $i\in\{0,1\}$, and for
some $\epsilon_\text{binding}\geq 0$ that goes to zero as some
security parameter goes to infinity. The hiding and binding properties
are also called \emph{security against Alice} and \emph{security
  against Bob}, respectively, when Bob is the committing party.

Nonrelativistic quantum bit commitment cannot achieve unconditional
security \cite{M97,LC97,LC98}. However, secure bit commitment can be
achieved in the noisy storage model \cite{DFSS08,WST08}. Nevertheless, there
are relativistic classical \cite{K99,K05.2,LKBHTWZ15} and
quantum \cite{K11.2,K12,AK15.1,AK15.2} bit commitment protocols that
are unconditionally secure.

\subsection{Quantum oblivious transfer}
In a $1$-out-of-$2$ oblivious transfer (OT) protocol \cite{K88}, Alice
inputs two strings of $n$ bits, $x_0$ and $x_1$, initially secret from
Bob. Bob inputs a bit $b$, initially secret from Alice. At the end of
the protocol, Bob outputs the string $x_b$. Two security conditions
must be fulfilled, called security against Alice and security against
Bob. \emph{Security against Alice} states that, if Bob follows the
protocol and Alice deviates arbitrarily from the protocol, the
probability that Alice guesses Bob's input $b$ satisfies
$P_\text{Alice}\leq \frac{1}{2}+\epsilon_\text{Alice}$, for some
$\epsilon_\text{Alice}\geq 0$ that goes to zero as some security
parameter goes to infinity. \emph{Security against Bob} states that,
if Alice follows the protocol and Bob deviates arbitrarily from the
protocol, Bob cannot learn both strings $x_0$ and $x_1$; this can be
quantified by stating that the probability that Bob obtains both
messages satisfies $P_\text{Bob}\leq \epsilon_\text{Bob}$, for some
$\epsilon_\text{Bob}\geq 0$ that goes to zero as some security
parameter goes to infinity.

$1$-out-of-$m$ OT cannot be implemented with unconditional
security in quantum cryptography \cite{L97}. This impossibility
theorem holds even in the setting of relativistic quantum cryptography
\cite{CK06}, although some relativistic variations of the task can be
achieved with unconditional security
\cite{PG15.1,PGK18,PG19}. However, $1$-out-of-$2$ OT can be
implemented securely in the noisy storage model \cite{DFSS08,WST08}.

\subsection{Quantum coin flipping}

In strong coin flipping, Bob and Alice, who are at distant locations,
obtain a bit $a$ that is random and which cannot be biased by neither
of them. A strong coin flipping protocol must satisfy security against
Alice and security against Bob. \emph{Security against Alice}
(\emph{Bob}) states that if Alice (Bob) follows the protocol and Bob
(Alice) deviates arbitrarily from the protocol, and Alice (Bob)
obtains as outcome the bit $a$ then it holds that
$P(a=i)\leq \frac{1}{2}+\epsilon$, for $i\in\{0,1\}$, and for some
$\epsilon\geq 0$ that goes to zero as some security parameter goes to
infinity.

Although there are relativistic protocols for strong coin flipping
that are unconditionally secure \cite{K99.2}, quantum nonrelativistic
protocols for strong coin flipping cannot achieve unconditional
security \cite{LC98}. However there are quantum
nonrelativistic protocols for strong coin flipping that
unconditionally guarantee some level of security, in that the bias 
a dishonest party can give the coin is bounded below one \cite{ATVY00,SR02,Kitaev02,ABDR04,NS03,A04,CK09}.

\section{multiphoton attacks on experimental demonstrations of
  mistrustful cryptography protocols}
\label{appD}
Some of the quantum protocols discussed in the previous section
offer unconditional security guarantees; others offer security based
on technological assumptions.    Over the past decade, they have
been implemented in pioneering experimental demonstrations (e.g. \cite{NJMKW12,LKBHTKGWZ13,LCCLWCLLSLZZCPZCP14,ENGLWW14,PJLCLTKD14}), often 
adapting quantum key distribution technology.  
 Refs. \cite{NJMKW12,ENGLWW14} demonstrated
  quantum bit commitment and 1-out-of-2 quantum oblivious transfer in
  the noisy storage model. Ref. \cite{PJLCLTKD14} demonstrated quantum
  coin flipping performing better than classical protocols over a
  distance of various kilometres, gaining three orders of magnitude in
  communication distance over previous
  experiments. Refs. \cite{LKBHTKGWZ13,LCCLWCLLSLZZCPZCP14}
  implemented quantum relativistic bit commitment protocols for the
  first time, showing that the protocol of Ref. \cite{K12}
can be implemented in practice, over short and long range.  

Two of us (E.D. and A.K.) have contributed to this work. 
Scientific honesty compels us to acknowledge that we did not appreciate all the obstacles
to attaining provable and truly unconditional security in practical implementations of mistrustful quantum cryptographic
protocols.   In hindsight, we believe the 
implementations \cite{NJMKW12,LKBHTKGWZ13,LCCLWCLLSLZZCPZCP14,ENGLWW14,PJLCLTKD14}
may be best seen as proofs of principle.
They show that some key technological challenges in implementing the protocols 
have been met and give significant and valuable security guarantees
based on assumptions about the parties' behaviour and technology
that are natural in some scenarios.  
However, more work appears to be required to deliver provably unconditional
security. In particular, provable unconditional security against Alice (i.e. the
sender of the quantum states, in our convention) requires
implementing techniques countering the various side-channel and multiphoton attacks we have
described and showing that her cheating probability can thus be 
made arbitrarily close to the ideal bound (which in most cases is
zero).  

Users and developers of mistrustful quantum
cryptosystems may also need to consider whether some technological
assumptions may be both necessary and sufficient in practical 
implementations, even for protocols that are theoretically
unconditionally secure, given the difficulty in provably 
countering every possible side-channel attack.

We will discuss
Refs. \cite{NJMKW12,LKBHTKGWZ13,LCCLWCLLSLZZCPZCP14,ENGLWW14,PJLCLTKD14}
separately below.  
We should first note that 
reanalysing previous implementations in the light of our attacks
is not completely straightforward.   
Although the relevant papers are generally very clear and comprehensive,
experimental details that at the time may not have seemed significant 
were not always given.   For example, some of the experiments used
off-the-shelf quantum key distribution equipment which may
have been programmed either to discard double clicks or 
choose a random outcome.  After discussing with colleagues,
we understand it may not now be possible to say with certainty which
option was used.   
This is unimportant for the future of the field.    
We believe we should note, though, when (what we now realise are) significant instructions
are missing from the published reports, since future
users trying to reproduce or improve on previous implementations
may rely on these.   Ref. \cite{NJMKW12} explicitly
states that only single click detection events are reported by Bob as
valid measurement outcomes, while Refs. \cite{LKBHTKGWZ13,ENGLWW14,PJLCLTKD14} do not say whether
multiple clicks are reported by Bob or not.

A related issue is the question of precisely how the symmetrization of
losses technique was or might have been implemented.  This technique
was introduced by Ref. \cite{NJMKW12} and claimed to guarantee
security against Alice by
Refs. \cite{NJMKW12,LKBHTKGWZ13,ENGLWW14,PJLCLTKD14}.  Indeed, as
Lemma \ref{lemma1} states for the symmetrization of losses with setup I (see
Fig. \ref{fig1} of the main text)
-- as applied by Ref. \cite{PJLCLTKD14}, for
example -- security against Alice is guaranteed when she does not send
multiphoton pulses, if Bob reduces the probability of assigning a
measurement outcome due to a single click event $(c,\bar{c})$ by a
suitable factor $S_{c\bar{c}\beta}$, for $c,\beta\in\{0,1\}$. A
similar result is shown by Lemma \ref{lemmaSLII} below in a setup with
four detectors (setup II, see Fig. \ref{setupII}), as implemented by
Refs. \cite{NJMKW12,ENGLWW14}.  However, the original symmetrization
of losses technique \cite{NJMKW12} treats multiple clicks as invalid,
and thus does not provide effective protection against multiphoton
attack I.  A natural defence against this attack is to report multiple
clicks as valid measurement outcomes.  However, one then needs to
define an extension of the symmetrization of losses reporting strategy
applicable to multiple click outcomes. A general extension of
symmetrization of losses in setup I is what we have defined as
reporting strategy III in (\ref{a5}), where a double click is reported
by Bob with some probability $S_{11\beta}\in[0,1]$ when Bob measures
in the basis $\mathcal{B}_\beta$, for $\beta\in\{0,1\}$. Theorem \ref{lemma3} then shows that multiphoton attack II applies if the
detector efficiencies of Bob's detectors are different and known by
Alice.  A countermeasure against this attack is to ensure Bob's
detectors have equal efficiency. This cannot be implemented
perfectly.  Lemma \ref{lemmanewbound} guarantees that,
if the efficiencies of Bob's detectors are
sufficiently close, then Alice can obtain little information about $\beta$
if Bob applies reporting strategy III with the parameters of Lemma
\ref{lemmanewbound}.   However, as Fig. \ref{figbound} illustrates, 
the efficiency differences need to be very small.

We should note that  Ref. \cite{LKBHTKGWZ13} implemented a
slightly different version of symmetrization of losses with setup I 
which aimed not to symmetrize the detection efficiencies but to
symmetrize Bob's detection probabilities for both values of Bob's
measurement basis $\beta\in\{0,1\}$.
This version of symmetrization of losses naturally incorporates
double click events, which can be assigned a random measurement
outcome.
However, multiphoton attack II still applies.

We also note again that Ref. \cite{LCCLWCLLSLZZCPZCP14} discussed multiphoton attack I
and implemented the countermeasure of reporting single
and double clicks.   Again, multiphoton attack II still
applies.

Below we discuss variations of the multiphoton attacks I and II
presented in the main text, as they apply to the protocols of
Refs. \cite{NJMKW12,LKBHTKGWZ13,LCCLWCLLSLZZCPZCP14,ENGLWW14,PJLCLTKD14}.
Because these protocols extend the task of private measurement of an
unknown qubit state of the main text to a setting with $N>1$ photon
pulses in different ways, we need to discuss the attacks separately
for each protocol. A summary is given in Table \ref{tablemaintext} in the main text.

\subsection{multiphoton attacks on the relativistic quantum bit commitment protocol of Ref. \cite{LKBHTKGWZ13}}
\label{LKBHTKGWZ13}

\subsubsection{The relativistic quantum bit commitment protocol of Ref. \cite{LKBHTKGWZ13}}

\label{Lunghiprotocol}

Ref. \cite{LKBHTKGWZ13} (co-authored by one of us)
demonstrated the quantum relativistic bit commitment protocol of
Ref. \cite{K12}, using an extra stage of pre-processing to 
allow commitments to be made by the pre-agreed actions of 
parties separated by several thousand kilometres.  
We discuss here the application of multiphoton attacks on the
protocol (as presented) and show that it does not guarantee 
hiding with unconditional security if the committer's detectors
have unequal efficiency and the committee becomes aware of 
the efficiencies. 
Ref. \cite{LKBHTKGWZ13} takes Bob as the committing party. We follow
this convention below. 

The protocol of Ref. \cite{LKBHTKGWZ13} is a relativistic quantum
protocol. The attacks that we present below are implemented in the
quantum stage of the protocol, which is nonrelativistic. For this
reason, here we only need to discuss the quantum stage.

The nonrelativistic quantum stage of the protocol is as follows. Alice
and Bob use setup I discussed in the main text and illustrated in
Fig. \ref{fig1}. Alice's photon source is a weak coherent source with
small average photon number $\mu$. Alice sends Bob $N$ photon pulses,
each encoding in the polarization a qubit state chosen randomly from
the set
$\mathcal{S}=\{\lvert 0\rangle,\lvert 1\rangle,\lvert
\tilde{+}\rangle,\lvert \tilde{-}\rangle\}$,
where $\mathcal{B}_0=\{\lvert 0\rangle,\lvert 1\rangle\}$ and
$\mathcal{B}_1=\{\lvert \tilde{+}\rangle,\lvert \tilde{-}\rangle\}$
are two qubit orthogonal bases. Ref. \cite{LKBHTKGWZ13} considers the
particular case that $\mathcal{B}_0$ and $\mathcal{B}_1$ are the
computational and Hadamard bases, respectively. Bob chooses a random
bit $\beta$. Immediately after their reception, Bob measures each of
the $N$ photon pulses in the qubit orthogonal basis
$\mathcal{B}_\beta$. In order to deal with losses in the quantum
channel, for each pulse sent by Alice, Bob sends Alice a message $m=1$
if the pulse produced a valid measurement outcome and $m=0$
otherwise. Bob chooses the committed bit $b$ and sends the message
$c= \beta\oplus b$ to Alice.

Ref. \cite{LKBHTKGWZ13} implemented the following version of
  the symmetrization of losses strategy. Bob tests his system by
implementing the protocol agreed with Alice, with the agreed
experimental parameters. Then, Bob computes the ratio
$R=\frac{n_0}{n_1}$, where $n_\beta$ is the number of pulses producing
valid measurement outcomes when Bob measures all pulses in the basis
$\mathcal{B}_\beta$, for $\beta\in\{0,1\}$. Then, when implementing
the protocol with Alice, Bob performs the following actions. If
$R\leq 1$, when a pulse produces a click,
Bob sets $m=1$ and assigns a valid measurement outcome with
probability $R$ if $\beta=1$, or with unit probability if
$\beta=0$. On the other hand, if $R> 1$, when a pulse produces a click, Bob sets $m=1$ and assigns a valid
measurement outcome with probability $\frac{1}{R}$ if $\beta=0$, or
with unit probability if $\beta=1$. This effectively makes Bob's detection probabilities for
the cases $\beta=0$ and $\beta=1$ equal when both parties follow the
protocol. We note that Ref. \cite{LKBHTKGWZ13} does not
  explicitly say whether this procedure 
applies only to single clicks, or to single and double clicks,
although we understand that double clicks were probably counted.

However, we need to consider cheating strategies available
to  Alice.  We make some simplifying assumptions (which if anything
reduce the power of cheating attacks) to illustrate these.
We assume that $\mathcal{B}_0$
and $\mathcal{B}_1$ are the computational and Hadamard bases,
respectively, as in Ref. \cite{LKBHTKGWZ13}. We assume that Bob uses
the single photon threshold detectors $D_0$ and $D_1$ to register
outcomes associated with the states $\lvert 0\rangle$
($\lvert +\rangle$) and $\lvert 1\rangle$ ($\lvert -\rangle$) if
$\beta=0$ ($\beta=1$). We assume that $\eta_{i\beta}=\eta_i\in(0,1)$,
for $i,\beta\in\{0,1\}$. Since Bob cannot guarantee his detector to
have exactly the same efficiencies, we assume that
$\eta_0\neq \eta_1$. Without loss of generality, we assume that
$1>\eta_0>\eta_1>0$.    We also
assume that Alice's preparation devices and Bob's measurement devices
are perfectly aligned. Thus, random BB84 states are prepared exactly
by Alice. Bob randomly chooses $\beta\in\{0,1\}$ and measures all
pulses exactly in the basis $\mathcal{B}_\beta$, where $\mathcal{B}_0$
and $\mathcal{B}_1$ are exactly the computational and Hadamard bases,
respectively. We note from these assumptions that, due to symmetry,
the detection probabilities when Bob measures in the basis
$\mathcal{B}_0$ are the same as the detection probabilities when Bob
measures in the basis $\mathcal{B}_1$. Thus, for large $N$, the ratio
$R$ obtained by Bob in his symmetrization of losses strategies is very
close to unity. We assume here that $R=1$.

\subsubsection{multiphoton attack I}
\label{singleattack}

Suppose that Bob only assign valid measurement outcomes to pulses
that produce a click in only one of his detectors. Furthermore, we
assume that Bob applies the symmetrization of losses strategies
described above. As discussed in the main text, Alice can implement
multiphoton attack I and gain information about Bob's
measurement basis $\mathcal{B}_\beta$. To illustrate this, consider a
setup in which Alice's polarization preparation devices and Bob's
polarizers are precisely aligned.  Alice can send a pulse with a large
number of photons $k$ in the same polarization state chosen from
$\mathcal{S}$; for example,
$\rho=(\lvert 0\rangle\langle 0\rvert)^{\otimes k}$. If Bob measures
the pulse in the basis $\mathcal{B}_0$ then the detection event
$(c_0,c_1)=(1,0)$ occurs with high probability, giving $m=1$ with unit
probability if $R\leq 1$, or with probability $\frac{1}{R}$ if
$R>1$. If Bob measures in the basis $\mathcal{B}_1$ then the detection
event $(c_0,c_1)=(1,1)$ occurs with high probability, giving
$m=0$. Thus, given $m$, Alice can learn $\beta$ with high
probability. Note that this attack applies
  whether or not $\eta_0=\eta_1$.

In practice, Alice and Bob's devices cannot be perfectly
precise. However, Alice can still learn significant information about
Bob's bit $\beta$ from the message $m$ with an appropriate choice of
$k$. In particular, Alice can send Bob photon pulses prepared with a
coherent source with average photon number $\mu>>1$ and guess
Bob's bit $\beta$ with probability close to unity, as our experimental
simulation shows (see Fig. \ref{guessingexp} of the main text). Furthermore, Alice can
increase her probability to guess $\beta$ by sending Bob various
pulses with a large number of photons.

Partial countermeasures that Bob can apply against these
attacks are the following. First, Bob can use reporting strategy II  and assign a random outcome when a double click is obtained; or he can use reporting strategy III with the parameters of Lemma \ref{lemmanewbound}. Lemma \ref{lemma0} shows reporting strategy II is perfectly effective if Bob's detectors have exactly equal efficiencies. Lemmas \ref{lemma6} and \ref{lemmanewbound} guarantee that,
if the efficiencies of Bob's detectors are
sufficiently close, then Alice can obtain little information about $\beta$
if Bob applies reporting strategy II or III, with the parameters of Lemma \ref{lemma6}  or \ref{lemmanewbound}, respectively. However, as Fig. \ref{figbound} illustrates, 
the efficiency differences need to be very small.
Second, Bob can
abort the protocol if a fraction of pulses greater than
$r_\text{double}^\text{max}\in(0,1)$, previously agreed with Alice,
produces double clicks.

\subsubsection{multiphoton attack II}
\label{app1}

As noted above, multiphoton attack I would apply if Bob discarded
double clicks as invalid measurements.  We thus assume that Bob sends Alice the message $m=1$ for each pulse sent
by Alice that produces a click in at least one of his detectors. We
also assume that Bob applies the countermeasure against multiphoton
attacks in which he aborts if he observes a ratio of double click
events higher than a maximum value
$r_\text{double}^\text{max}\in(0,1)$ agreed with Alice. We present a
multiphoton attack by Alice that for certain parameters allows Alice
to guess Bob's bit $\beta$ with probability approaching unity as the
number of pulses $N$ of the protocol increases.

Alice generates two nonempty and nonintersecting subsets of $[N]$,
$\Omega_\text{protocol}$ and $\Omega_\text{attack}$, satisfying
$\Omega_\text{protocol}\cup\Omega_\text{attack}=[N]$. Let
$a=\frac{\lvert\Omega_\text{attack}\rvert}{N}$. It follows that
$\lvert\Omega_\text{attack}\rvert=aN$ and
$\lvert\Omega_\text{protocol}\rvert=(1-a)N$, with $a\in(0,1)$. Alice
sends Bob $N$ photon pulses. The polarization degrees of freedom of
each of the pulses with labels from the set $\Omega_\text{protocol}$
are prepared by Alice in a quantum state as established in the
protocol agreed with Bob. Each of the pulses with labels from the set
$\Omega_\text{attack}$ is prepared by Alice in $k^*>1$ photons,
encoding in the polarization a quantum state $\rho$. Let
$P_\text{protocol}(1\vert \beta)$ and $P_\text{attack}(1\vert \beta)$
be the probabilities that a pulse with label from the sets
$\Omega_\text{protocol}$ and $\Omega_\text{attack}$ activates a
detection in at least one of the two detectors, respectively, i.e that
Bob sends Alice the message $m=1$ for that pulse, when Bob measures
the pulses in the basis $\mathcal{B}_\beta$, for $\beta\in\{0,1\}$.

Alice chooses $k^*$, $\rho$ and $a$ in such a way that: 1) the probability that there are more than $Nr_\text{double}^\text{max}$ double clicks in Bob's detectors is negligible; and 2) it holds that 
\begin{equation}
\label{extended}
g_1(a,\rho,k^*)>g_0(a,\rho,k^*)>0,
\end{equation}
where
\begin{equation}
\label{extended0}
g_\beta(a,\rho,k^*)=aP_\text{attack}(1\vert \beta)+(1-a)P_\text{protocol}(1\vert \beta),
\end{equation}
for $\beta\in\{0,1\}$. Let $Z_\beta$ denote the random variable corresponding to the number of pulses producing that at least one of the two of Bob's detectors click when Bob measures in the basis $\mathcal{B}_\beta$, and let $E(Z_\beta)$ denote its expectation value, for $\beta\in\{0,1\}$. We have that 
\begin{equation}
\label{extended2}
E(Z_\beta)=Ng_\beta(a,\rho,k^*),
\end{equation}
for $\beta\in\{0,1\}$. Thus, from (\ref{extended}) and (\ref{extended2}), we have
\begin{equation}
\label{extended3}
E(Z_1)>E(Z_0).
\end{equation}
Alice defines a parameter $\delta\in(0,1)$ satisfying 
\begin{equation}
\label{extended5}
E(Z_0)(1+\delta)=E(Z_1)(1-\delta)=G_N.
\end{equation}
If the number of events $Z$ reported by Bob to give at least one
click, i.e. for which Bob sends the message $m=1$ to Alice, is smaller
than $G_N$, Alice guesses that Bob measured in the basis
$\mathcal{B}_0$, otherwise Alice guesses that Bob measured in the
basis $\mathcal{B}_1$. It follows that Alice's average probability to
guess Bob's bit $\beta$ in this attack is given by
\begin{eqnarray}
\label{z30}
P_{\text{guess}}&=&\frac{1}{2}\text{Prob}[Z_0<E(Z_0)(1+\delta)]\nonumber\\
&&\qquad+\frac{1}{2}\text{Prob}[Z_1\geq E(Z_1)(1-\delta)]\nonumber\\
&=&1-\frac{1}{2}\text{Prob}[Z_0\geq E(Z_0)(1+\delta)]\nonumber\\
&&\qquad-\frac{1}{2}\text{Prob}[Z_1< E(Z_1)(1-\delta)]\nonumber\\
&\geq&1-\frac{1}{2}\Bigl[e^{-N\frac{g_0(a,\rho,k^*)\delta^2}{3}}+e^{-N\frac{g_1(a,\rho,k^*)\delta^2}{2}}\Bigr],\nonumber\\
\end{eqnarray}
where in the last line we used (\ref{extended2}) and two Chernoff bounds \cite{Mitzenmacherbook}. Thus, since $g_\beta(a,\rho,k^*)>0$ for $\beta\in\{0,1\}$, we see from (\ref{z30}) that $P_{\text{guess}}$ approaches unity with increasing $N$.

A possible countermeasure by Bob against this type of attack
is to measure suitable statistical properties of the pulses sent by
Alice.  For example, Bob may modify his setup by adding more beam
splitters and single photon detectors. In this way, Bob can measure
the number of pulses producing clicks across different combinations of
his detectors. If Bob observes statistics that deviate considerably
from the statistics expected in the protocol agreed with Alice, Bob
may abort the protocol.

\subsubsection{Example of multiphoton attack II: a double-photon attack}
\label{double}

Suppose Bob reports double clicks, i.e. sends the message
$m = 1$ to Alice and randomly chooses the outcome. We present an attack by Alice in which the photon statistics of the
pulses that she sends to Bob correspond to those agreed for her
weak coherent source. 
Thus, there is no way in which Bob can know that Alice is
cheating.

We illustrate the attack on an implementation with parameters
$N=2\times 10^{7}$, $\eta_{0}=0.12$, $\eta_{1}=0.08$,
$d_0=d_1=10^{-5}$ and a weak coherent source with average
  photon number $\mu=0.05$. The values $\eta_0=0.12$ and
  $\eta_1=0.08$ are consistent with uncertainty values for the
  detection efficiencies of $0.02$, which is a common value.
  %  \cite{conversations2}.
  We show that Alice can guess Bob's bit $\beta$
  with failure probability smaller than $0.035$. Our example shows
  that Alice could in principle undetectably exploit the difference of
  Bob's detector efficiencies due to their experimental uncertainty to
  guess Bob's input $\beta$ with high probability, for some values of
  the experimental parameters.
 
Note that the actual experimental parameters reported in
  Ref. \cite{LKBHTKGWZ13} are different: 
$N=2.2\times 10^{6}$, $\eta\approx0.06$ and $\mu=0.05$. The values of
$d_0,d_1$ and the uncertainties of Bob's detection efficiencies were
not reported by Ref. \cite{LKBHTKGWZ13}.   Even if we assume the above values for
$d_0$ and $d_1$ and uncertainties of the detection
efficiencies not greater than $0.02$, we have not shown the
attack discussed allows Alice to guess $\beta$ with probability close
to unity in the experiment actually implemented.   What our illustration shows is that reproducing the
implementation with modestly different and plausible parameters leads to an
unnoticed insecurity.

As discussed in the main text, we assume that Alice knows the detection efficiencies
$\eta_0$ and $\eta_1$ of Bob's respective detectors $D_0$ and
$D_1$. This is, for example, because Alice has manufactured the
detectors used by Bob, or because she has obtained information about
their detection efficiencies in some other way. We also assume that
Alice can know the number of photons for each pulse prepared by her weak
coherent source. In particular, we assume that Alice knows which of her
prepared pulses have two photons. Alice prepares the states
perfectly. These assumptions might seem very strong. But these are
standard assumptions when trying to show unconditional security
against Alice, in which it is assumed that Alice has access to
perfect technology and is only limited by the laws of physics.

Furthermore, we assume that Bob does not know the values of $\eta_0$
and $\eta_1$. For example, Bob only knows that $\eta_0$ and $\eta_1$
are within some -- possibly small -- range due to their assigned
uncertainties. We assume that both Bob and Alice know the dark count
probabilities $d_0$ and $d_1$ of Bob's detectors.

In her attack, as required by the protocol agreed with Bob, Alice
prepares $N$ pulses with a weak coherent source of average photon
number $\mu$, with all the photons in each pulse encoding the same
qubit state, which is randomly chosen from the BB84 set for each
pulse. Alice sends the $N$ pulses to Bob with their labels
$j\in[N]$. Let $\Omega$ be the set of labels for pulses with two
photons prepared in the state $\lvert 0\rangle$, and let
$\Omega^{\text{rep},\beta}$ be the set of labels from $\Omega$ for
which Bob sends Alice the message $m=1$, i.e for which he assigns a
valid measurement outcome, when he measures the pulses in the basis
$\mathcal{B}_\beta$, for $\beta\in\{0,1\}$. The probability that a
pulse prepared by Alice has two photons is given by
$p_2=e^{-\mu}\mu^2/2$.  Let $N_0=\lvert \Omega\rvert$ and
$N_0^{\text{rep},\beta}=\lvert \Omega^{\text{rep},\beta}\rvert$, for
$\beta\in\{0,1\}$. In her attack, Alice only focuses on the labels of
pulses from the set $\Omega^{\text{rep},\beta}$, with the goal of
guessing the value of $\beta\in\{0,1\}$.

Let $P_{\beta}(i)$ be the probability that a pulse of two photons that Alice
prepares in the state $\lvert 0\rangle$ activates a detection in at
least one of the two Bob's detectors if $i=1$, and that it does not
activate any detection if $i=0$, when Bob measures in the basis
$\mathcal{B}_\beta$, for $\beta\in\{0,1\}$. We compute $P_{\beta}(i)$ for
$i,\beta\in\{0,1\}$ below and show that $P_0(1)>P_1(1)>0$ for
$\eta_0>\eta_1>0$. Thus, we have
$E(N_0^{\text{rep},0})>E(N_0^{\text{rep},1})$, where
\begin{equation}
\label{z1.5}
E(N_0^{\text{rep},\beta})=N e^{-\mu}\mu^2P_\beta(1)/8
\end{equation}
 is the expectation value of the random variable $N_0^{\text{rep},\beta}$, for $\beta\in\{0,1\}$. Therefore, since $P_0(1)>P_1(1)>0$, we can find numbers $\delta>0$ and $G$ such that
\begin{equation}
\label{z1.2}
E(N_0^{\text{rep},1})(1+\delta)=E(N_0^{\text{rep},0})(1-\delta)=G.
\end{equation}
It is straightforward to obtain that
\begin{eqnarray}
\label{z1.3}
\delta&=&\frac{P_{0}(1)-P_{1}(1)}{P_{0}(1)+P_{1}(1)},\\
\label{z1.4}
G&=&\frac{N e^{-\mu}\mu^2 P_{0}(1) P_{1}(1)}{4(P_{0}(1) +P_{1}(1))}.
\end{eqnarray}

Alice's guessing strategy is as follows. Bob measures in the basis
$\mathcal{B}_\beta$ and Alice knows the value of $N_0^{\text{rep},\beta}$, for
some $\beta\in\{0,1\}$. But Alice does not know the value of $\beta$, which she
is trying to guess. To clarify this issue, let us denote
$N_0^{\text{rep}}=N_0^{\text{rep},\beta}$. Alice computes the value of
$N_0^{\text{rep}}$ from the pulses for which Bob reported $m=1$. If
$N_0^{\text{rep}}\leq G$, Alice guesses $\beta=1$, otherwise she guesses
$\beta=0$. As we show below, from (\ref{z1.2}) and from
Chernoff bounds, Alice's probability of failure $P_{\text{fail}}$
in guessing $\beta$ is very small for $N$ large enough.

We compute an upper bound on Alice's probability of failure
$P_{\text{fail}}$. From (\ref{z1.3}), and since we show below that
$P_0(1)>P_1(1)>0$ from $\eta_0>\eta_1>0$, we have $0<\delta<1$. It
follows that
\begin{eqnarray}
\label{abcd3}
P_{\text{fail}}&=&\frac{1}{2}\bigl( \text{Prob}(N_0^{\text{rep},0}\leq G) + \text{Prob}(N_0^{\text{rep},1}> G)\bigr)\nonumber\\
&<&\frac{1}{2}\biggl( e^{-\frac{E(N_0^{\text{rep},0})\delta^2}{2}} + e^{-\frac{E(N_0^{\text{rep},1})\delta^2}{3}}\biggr)\nonumber\\
&=&\frac{1}{2}\biggl( e^{-\frac{N e^{-\mu}\mu^2P_0(1)\delta^2}{16}} + e^{-\frac{N e^{-\mu}\mu^2P_1(1)\delta^2}{24}}\biggr),
\end{eqnarray}
where in the first line we used the definition of Alice's guessing
strategy; in the second line we used (\ref{z1.2}), $0<\delta<1$ and
Chernoff bounds \cite{Mitzenmacherbook}; and in the last line we used
(\ref{z1.5}). Thus, we see from (\ref{abcd3}) that Alice's probability
of failure can be negligible for $N$ large enough.

It is important to note that Bob's reported detection frequencies are the
same for both cases: 1) when he measures all the received pulses in
the computational basis ($\beta=0$), and 2) when he measures all the
received pulses in the Hadamard basis ($\beta=1$). This is easy to see by
noting the following. First, Alice prepares BB84 states
randomly. Second, Bob's detection probabilities for pulses prepared
by Alice in the state $\lvert 0\rangle$ ($\lvert 1\rangle$) and measured
by Alice in a basis $\mathcal{B}_\beta$ are the same for pulses prepared
by Alice in the state $\lvert +\rangle$ ($\lvert -\rangle$) and measured
by Bob in a basis $\mathcal{B}_{\bar{\beta}}$ and vice versa. Thus, in the
symmetrization of losses strategy implemented by Bob he gets the ratio $R$ very close to unity, for large $N$. We conclude from this attack that the symmetrization of losses strategy does not guarantee to Bob that
Alice cannot obtain information about her choice of basis $b$. It follows that the experimental demonstration of bit commitment of Ref. \cite{LKBHTKGWZ13} is not hiding; hence, it is not unconditionally secure.

We compute $P_{0}(1)$ and $P_{1}(1)$ and we show that
$P_{0}(1) > P_{1}(1)>0$, for $\eta_0>\eta_1>0$. Both photons of a
two-photon pulse prepared by Alice in the state $\lvert 0\rangle$ go to
the detector $D_0$ when Bob measures in the computational basis
(i.e. when $\beta=0$). Thus, we have
\begin{equation}
\label{abcd1}
P_{0}(1)=1-(1-d_0)(1-d_1)(1-\eta_0)^2,
\end{equation}
where $d_i$ the probability of a dark count by the detector $D_i$, for $i\in\{0,1\}$. On the other hand, when Bob measures in the Hadamard basis (i.e. when $\beta=1$), each photon goes to the detector $D_i$ with probability $\frac{1}{2}$. Thus, we have
\begin{eqnarray}
\label{abcd2}
P_{1}(1)&=&1-\frac{(1-d_0)(1-d_1)}{4}\bigl[(1-\eta_0)^2+(1-\eta_1)^2\nonumber\\
&&\qquad +2(1-\eta_0)(1-\eta_1)\bigr].
\end{eqnarray}
Note that (\ref{abcd1}) and (\ref{abcd2}) follow straightforwardly from Lemma \ref{lemma2} of the main text. It follows from (\ref{abcd1}) and (\ref{abcd2}) that $P_{0}(1) > P_{1}(1)>0$, for $\eta_0>\eta_1>0$.

We compute an example of Alice's probability of failure
$P_{\text{fail}}$ in this attack. We consider the parameters
$d_0=d_1=10^{-5}$, $\mu=0.05$, $\eta_0=0.12$, $\eta_1=0.08$,
$N=2\times 10^7$. These are reasonable experimental parameters. For
example, $\mu=0.05$, $N=2.2\times 10^6$ and values of $\eta$ around
$0.06$ are reported by Ref. \cite{LKBHTKGWZ13}. We note that our
chosen value for $N$ is roughly only ten times the value chosen in
Ref. \cite{LKBHTKGWZ13}. Although Ref. \cite{LKBHTKGWZ13} does not
report this, uncertainties in detection efficiency of $0.02$ or higher
are usual;
%\cite{conversations2}
hence, values
$\eta_0=0.12$ and $\eta_1=0.08$ are reasonable. Moreover, a reasonable
experimental value for the dark count rate is 1000/s, which for a
standard gate width value of 10 ns gives a dark count probability of
$10^{-5}$. From (\ref{z1.3}), (\ref{abcd1}) and (\ref{abcd2}), we
obtain $P_{0}(1)=0.2256$, $P_{1}(1)=0.1900$, $\delta=0.0857$. Thus,
from (\ref{abcd3}), we obtain $P_{\text{fail}}<0.035$.

The previous attack can straightforwardly be extended in various
ways. For example, Alice may try to guess Bob's measurement choice
$\beta\in\{0,1\}$ from pulses for which Bob reported $m=1$ that Alice prepared
in some particular state $\lvert\psi\rangle$ from the BB84 set with
some particular number of photons $k$, or with number of photons in
some range, for example with $k\geq 2$. More generally, the idea of
the previous attack is that Alice follows the protocol agreed with Bob, but still
is able to exploit the fact that Bob's detectors have different
detection efficiencies $\eta_0\neq\eta_1$, and that she has some
knowledge about $\eta_0$ and $\eta_1$.

In the previous attack and extensions discussed above, Alice follows the
protocol agreed with Bob, i.e. she prepares BB84 states randomly with a weak
coherent source set to a small average photon number $\mu$ previously
agreed with Bob. This guarantees to Alice that Bob cannot detect her
cheating, because Alice is following the agreed protocol. In more
general attacks, Alice may deviate from the agreed protocol in various
ways. For example, she may prepare BB84 states with a probability
distribution different to random, or she may prepare more general
states, she could prepare some photon pulses with different average
photon number, she could use different photon sources, etc. However,
all these variations from the agreed protocol may be detected in
principle by Bob. Bob could, for example, apply some specific
quantum measurements in a subset of the received pulses and verify
that the detection probabilities and the outcome probabilities
correspond to the expected values in the agreed protocol. 

We note that the previous attack exploits the fact that the
probability that Bob reports the detection of a pulse is different for
different measurement basis. A partial countermeasure against this
attack, or extensions, consists in implementing a nontrivial
probabilistic reporting strategy, discussed in the main text, with
appropriate values of the reporting probabilities $S_{c_0c_1\beta}$,
for $c_0,c_1,\beta\in\{0,1\}$. However, as shown in Theorem \ref{lemma3} of the
main text, for any nontrivial probabilistic reporting strategy applied
by Bob, Alice can obtain some information about Bob's chosen bit
$\beta$. Thus, this countermeasure cannot be perfectly effective.

\subsubsection{Extending the bounds of Lemmas \ref{lemma5} --\ref{lemmanewbound}}

\label{bounds}

We recall that in Ref. \cite{LKBHTKGWZ13} Bob applies a
  particular form of the symmetrization of losses as reporting
  strategy, in which the goal is to symmetrize his detection
  probabilities for his both choices $\beta\in\{0,1\}$ of measurement
  bases. However, below we deduce an upper bound on the probability
  that Alice can guess Bob's bit $\beta$ when Bob applies either the
  reporting strategy II or the reporting strategy III, with the
  parameters of Lemmas \ref{lemma6} and \ref{lemmanewbound},
  respectively.

The bounds of Lemmas \ref{lemma6} and \ref{lemmanewbound} apply to the bit commitment protocol of
Ref. \cite{LKBHTKGWZ13} when Alice implements a multiphoton attack
consisting in a single photon pulse and Bob applies reporting strategy
II or III, respectively.
Thus, Lemma \ref{lemma5} implies an upper
bound on Alice's probability to guess Bob's bit $\beta$ with a
multiphoton attack of a single pulse. Alice can extend her attack by
sending various multiphoton pulses. In principle, the $N$ pulses that
Alice sends Bob in the protocol can be chosen by Alice as multiphoton
pulses with appropriately chosen photon numbers $k$. Lemmas
\ref{lemma5} and \ref{lemma6}, or \ref{lemma5} and
\ref{lemmanewbound}, can be used to deduce an upper bound on Alice's probability $P_\text{guess}$ to guess $\beta$ for this more
general case, when Bob applies reporting strategy II or III,
respectively. 

Let us assume for now that Bob applies an arbitrary reporting strategy. For $i\in[N]$, let $k_i$ be the number of photons encoded
in the polarization degrees of freedom of the $i$th pulse. We define $B_i$ as a number satisfying
\begin{equation}
\label{rrrb1}
\bigl\lvert P_\text{report}\bigl(1\vert 1,\rho_i,k_i\bigr)-P_\text{report}\bigl(1\vert 0,\rho_i,k_i\bigr)\bigr\rvert\leq B_i,
\end{equation}
for the reporting strategy applied by Bob, and for any quantum state
$\rho_i$ of $k_i$ qubits encoded in the $k_i$ photons of the $i$th
pulse, which may be entangled among the $k_i$ qubits, and which may
also be entangled with any other quantum systems. For example, if Bob
applies the reporting strategy II with the parameters of Lemma
\ref{lemma6} then $B_i$ can be given by the bound $B_\text{II}$ of
Lemma \ref{lemma6}, which is valid for any
$k_i\in\{0,1,2,\ldots\}$. If Bob applies the reporting strategy III
with the parameters of Lemma \ref{lemmanewbound} then $B_i$ can be
given by the bound $B_\text{III}$ of Lemma \ref{lemmanewbound}, which
is valid for any $k_i\in\{0,1,2,\ldots\}$. Furthermore, If Bob applies
the reporting strategy III with the parameters of Lemma \ref{lemma1}
and $k_i\in\{0,1\}$ then $B_i$ can be given by the bound
$B_\text{III}^{k_i}$ of Lemma \ref{lemma1}. We show below that
\begin{equation}
\label{bound}
P_\text{guess} \leq \frac{1}{2}+\frac{1}{2}\sum_{i=1}^N B_i,
\end{equation}
for any multiphoton attack by Alice.

If Bob applies the reporting strategy II with the parameters of Lemma \ref{lemma6} then we have from (\ref{bound}) and from Lemma \ref{lemma6} that
\begin{equation}
\label{boundII}
P_\text{guess} \leq \frac{1}{2}+\frac{NB_{\text{II}}}{2},
\end{equation}
for any multiphoton attack by Alice. However, as illustrated in Fig. \ref{figbound}, the bound $B_{\text{II}}$ of Lemma \ref{lemma6} is not very small even for relatively close values for the detection
efficiencies. Since $N$ is usually required to be large in order to
guarantee security against Bob -- for example, $N=2.2\times 10^6$ in
the protocol of Ref. \cite{LKBHTKGWZ13} -- the upper bound on
$P_\text{guess}$ given by (\ref{boundII}) is not smaller than unity for
any practical case.

It seems that our bounds would be more useful if Bob applies the reporting strategy III instead. Suppose that Bob applies reporting strategy III with the parameters of Lemma \ref{lemmanewbound}. From the $N$ pulses that Alice sends Bob, let $N_\text{empty}$, $N_\text{single}$ and $N_\text{mult}$ be the number of pulses with zero photons, with one photon, and with more than one photon, respectively.
Suppose  that Bob can guarantee with probability at least $1-\epsilon$ that
\begin{eqnarray}
\label{condition1}
N_\text{empty}&\geq& (1-\delta_\text{empty})N,\\
\label{condition2}
N_\text{mult}&\leq& \delta_\text{mult}N,
\end{eqnarray}
where
\begin{equation}
\label{condition3}
0\leq\delta_\text{mult}\leq\delta_\text{empty}\leq1,
\end{equation}
for some $0\leq\epsilon<<1$. Bob could guarantee this by aborting if the number of pulses not producing any click is below a threshold, or if the number of double clicks is above a threshold, for instance. These thresholds must be chosen carefully so that the probability that Bob aborts is negligible if Alice and Bob follow the agreed protocol. Since, from Lemmas \ref{lemma1} and \ref{lemmanewbound}, we have $B_\text{III}^0\leq B_\text{III}^1\leq B_\text{III}$ for $\delta$ small enough, it follows from (\ref{bound}), from (\ref{condition1}) -- (\ref{condition3}), and from Lemmas \ref{lemma1} and \ref{lemmanewbound} that
\begin{eqnarray}
\label{boundIV}
P_\text{guess}& \leq& \epsilon +(1-\epsilon)\biggl[\frac{1}{2}+\frac{N}{2}\Bigl((1-\delta_\text{empty})B_{\text{III}}^0\nonumber\\
&&\qquad\quad+(\delta_\text{empty}-\delta_\text{mult})B_{\text{III}}^1+\delta_\text{mult}B_{\text{III}}\Bigr)\biggr].\nonumber\\
\end{eqnarray}

Bob can always guarantee (\ref{condition1}) and (\ref{condition3}) with $\delta_\text{empty}=1$ and any $\delta_\text{mult}\in[0,1]$. In this case, if (\ref{condition2}) is also satisfied, then (\ref{boundIV}) reduces to
\begin{eqnarray}
\label{boundIII}
P_\text{guess}& \leq& \epsilon +(1-\epsilon)\biggl[\frac{1}{2}+\frac{N}{2}\Bigl((1-\delta_\text{mult})B_{\text{III}}^1\nonumber\\
&&\qquad\qquad\qquad+\delta_\text{mult}B_{\text{III}}\Bigr)\biggr].
\end{eqnarray}
We see from Lemma \ref{lemma1} that $B_{\text{III}}^1\leq 11\delta +3\delta^2$. Thus if the dark count probabilities are small enough; more precisely, if $N\delta<<1$, then the bound (\ref{boundIII}) can guarantee $P_\text{guess}-\frac{1}{2}<<\frac{1}{2}$ if $\epsilon<<1$ and $N\delta_\text{mult}B_{III}<<1$. 

The bound (\ref{boundIII}) seems promising in protocols agreed by Alice and Bob where the fraction of pulses with more than one photon is small enough. For example, if in the protocol agreed with Bob, Alice has a source of pairs of entangled photons in
which she measures one of the photons and sends the other one to
Bob, with transmissions considered valid if Alice obtains a
measurement outcome, the probability that a pulse has more than one photon is very small. Thus, Bob expects that the the fraction of pulses with more than one photon is small, if Alice follows the agreed protocol. 

The bound (\ref{boundIV}) with $\delta_\text{empty}<<1$ seems promising in protocols agreed by Alice and Bob where the fraction of pulses with more than one photon is small enough and the fraction of empty pulses is large enough. For example, if in the protocol agreed by Alice and Bob Alice uses a weak coherent source with small average photon number $\mu$, the probability that a pulse is empty is $p_0=e^{-\mu}\approx 1$, and the probability that a pulse has more than one photon is $p_\text{mult}=1-(1+\mu)e^{-\mu}<<1$. Thus, Bob expects that the the fraction of empty pulses is large and the fraction of pulses with more than one photon is small, if Alice and Bob follow the agreed protocol. 

If useful partial security can indeed be attained with
current technology by (for example) these methods, it would be worth
exploring whether greater security can be
attained by extending the bit commitment protocols.    
Suppose for example that one can implement a bit commitment protocol
with a guarantee that $P_{\rm guess} < \frac{1}{2} + \epsilon_{\rm
  hiding} $ and 
$p_0 + p_1 < 1 + \epsilon_{\rm binding}$, where $\epsilon_{\rm hiding}
\ll \frac{1}{2} $ and $\epsilon_{\rm binding}$ is very small.     
This guarantee may be conditional on some assumptions, for example
that Alice only uses photonic sources, with no side-channel attacks.
One could then, for example, consider a new protocol with $M$
iterations of the
above as sub-protocols in which   
Bob commits to bits $z_1 , \ldots , z_M$. The new protocol then commits Bob to the XOR bit $z = z_1 \oplus \ldots \oplus z_M$.
If Alice carries out only the attacks we have described separately
on each iteration, she obtains incomplete information
about each $z_i$.   Hence, by the piling-up lemma \cite{Pilinguplemma},
her information about the committed bit $z$ tends to zero
as $M$ gets large.   However, whether Alice has more powerful
attacks (under the given assumptions) needs to be
analysed, as does the dependence of $\epsilon_{\rm binding}$ on $M$.   

We leave as open problems to investigate whether the bounds
(\ref{boundIV}) and (\ref{boundIII}) can provide useful security in
practice and (if so) whether variant protocols offer further practical security
advantages.    Similar comments apply to the other implementations
of bit commitment schemes discussed below.

We show (\ref{bound}). In an arbitrary multiphoton attack, Alice
encodes a quantum state $\rho$ in the $N$ photon pulses that Alice
sends Bob. Each pulse has an arbitrary number of photons chosen by
Alice. The state $\rho$ is an arbitrary entangled state among all the
qubits encoded by the polarization degrees of freedom of the photons
in the $N$ pulses, which can also be entangled with an ancilla held by
Alice.

Bob sends Alice a message $m=(m_1,\ldots,m_N)\in\{0,1\}^N$ indicating
that the $j$th pulse produced a valid measurement outcome if $m_j=1$,
or otherwise if $m_j=0$, for $j\in[N]$. Let $P(m\vert \beta)$ be the
probability that Bob sends the message $m\in\{0,1\}^N$ when he
measures in the basis $\mathcal{B}_\beta$, and let $P_\beta$ denote
the probability distribution for $m\in\{0,1\}^N$ given $\beta$, for
$\beta\in\{0,1\}$. The probability $P_\text{guess}$ that Alice guesses
Bob's bit $\beta$ satisfies
\begin{eqnarray}
\label{boundeq1}
P_\text{guess}&\leq&\frac{1}{2}+\frac{1}{2}\lVert P_0 - P_1\rVert\nonumber\\
&=&\frac{1}{2}+\frac{1}{4}\sum_{m}\bigl\lvert P(m\vert 0)-P(m\vert 1)\bigr\rvert,\nonumber\\
\end{eqnarray}
where $\lVert P_0 - P_1\rVert$ is the variational distance between the
probability distributions $P_0$ and $P_1$, where in the second line we
used the definition of the variational distance, and where $m$ runs
over $\{0,1\}^N$.

 We show below that 
\begin{equation}
\label{boundeq2}
\sum_{m}\bigl\lvert P(m\vert 0)-P(m\vert 1)\bigr\rvert\leq 2\sum_{i=1}^NB_i.
\end{equation}
Thus, from (\ref{boundeq1}) and (\ref{boundeq2}) we obtain the bound (\ref{bound}).

Let $m_0=1$. We have
\begin{equation}
\label{boundeq3}
P(m\vert \beta)=\prod_{i=1}^N P_i(m_i\vert m_0m_1\cdots m_{i-1}\beta),
\end{equation}
where $P_i(j_i\vert j_0j_1\cdots j_{i-1}b)$ is the probability that $m_i=j_i$ given that
$m_l=j_l$ and that $\beta=b$,  for $(j_0,\ldots,j_N)\in\{0,1\}^{N+1}$, $b\in\{0,1\}$, $l\in\{0,1,\ldots,i-1\}$ and $i\in[N]$. From (\ref{boundeq2}) and (\ref{boundeq3}) we see that (\ref{boundeq2}) follows from
\begin{eqnarray}
\label{boundeq6}
2\sum_{i=1}^N B_i&\geq&\sum_{m}\biggl\lvert \prod_{i=1}^N P_i(m_i\vert m_0m_1\cdots m_{i-1}0)\nonumber\\
&&\quad-\prod_{i=1}^N P_i(m_i\vert m_0m_1\cdots m_{i-1}1)\biggr\rvert.
\end{eqnarray}

We show (\ref{boundeq6}). Let
$\rho_i(j_1\ldots j_{i-1}b)$ be the quantum state encoded in the $k_i$
qubits of photon-polarization in the $i$th pulse when $m_l=j_l$ and
$\beta=b$, for $l\in\{1,2,\ldots,i-1\}$, $b\in\{0,1\}$ and $i\in[N]$.
We see that
\begin{eqnarray}
\label{boundeq4}
&&P_i(m_i\vert m_0m_1\cdots m_{i-1}\beta)\nonumber\\
&&\qquad= P_\text{report}\bigl(m_i\vert \beta,\rho_i(m_1\ldots m_{i-1}\beta),k_i\bigr).
\end{eqnarray}
Thus, from (\ref{rrrb1}), we have
\begin{equation}
\label{boundeq5}
\lvert P_i(1\vert m_0m_1\cdots m_{i-1}1)-P_i(1\vert m_0m_1\cdots m_{i-1}0)\rvert\leq B_i,
\end{equation}
for $i\in[N]$.

We show (\ref{boundeq6}) by induction. We first show (\ref{boundeq6}) for the case $N=1$. We have,
\begin{eqnarray}
\label{boundeq7}
&&\sum_{m_1=0}^1\lvert P_1(m_1\vert m_01)-P_1(m_1\vert m_00)\rvert\nonumber\\
&&\qquad\qquad\qquad =2\lvert P_1(1\vert m_01)-P_1(1\vert m_00)\rvert\nonumber\\
&&\qquad\qquad\qquad \leq 2B_1,
\end{eqnarray}
as claimed, where in the first line we used that $P_1(1\vert m_0\beta)=1-P_1(0\vert m_0\beta)$ for $\beta\in\{0,1\}$, and in the second line we used (\ref{boundeq5}) for the case $i=N=1$.

Now we assume that (\ref{boundeq6}) holds for the case
$N=M\in\mathbb{N}$. We show that (\ref{boundeq6}) holds for the case
$N=M+1$. To simplify notation we take $m\in\{0,1\}^{M+1}$,
$y=(m_0,m_1,\ldots,m_N)$ and $x=m_{N+1}$. Using a telescoping sum, we
have
\begin{eqnarray}
\label{boundeq8}
W&\equiv& \sum_m\Biggl\lvert \prod_{i=1}^N P_i(m_i\vert m_0m_1\cdots m_{i-1}1) \nonumber\\
&&\qquad\qquad - \prod_{i=1}^N P_i(m_i\vert m_0m_1\cdots m_{i-1}0)\Biggr\rvert\nonumber\\
&=&\sum_x\sum_y\Biggl\lvert \bigl(P_{M+1}(x\vert y1)-P_{M+1}(x\vert y0)\bigr)\times\nonumber\\
&&\qquad\qquad \times\prod_{i=1}^M P_i(y_i\vert y_0y_1\cdots y_{i-1}0)\nonumber\\
&&\qquad +\biggl( \prod_{i=1}^M P_i(y_i\vert y_0y_1\cdots y_{i-1}1)\nonumber\\
&&\qquad\qquad\quad-\prod_{i=1}^M P_i(y_i\vert y_0y_1\cdots y_{i-1}0)\biggr)\times\nonumber\\
&&\qquad\qquad \times P_{M+1}(x\vert y1)\Biggr\rvert\nonumber\\
&=&\sum_y\Biggl\lvert \bigl(P_{M+1}(0\vert y1)-P_{M+1}(0\vert y0)\bigr)\times\nonumber\\
&&\qquad\qquad \times\prod_{i=1}^M P_i(y_i\vert y_0y_1\cdots y_{i-1}0)\nonumber\\
&&\qquad +\biggl( \prod_{i=1}^M P_i(y_i\vert y_0y_1\cdots y_{i-1}1)\nonumber\\
&&\qquad\qquad\quad-\prod_{i=1}^M P_i(y_i\vert y_0y_1\cdots y_{i-1}0)\biggr)\times\nonumber\\
&&\qquad\qquad \times P_{M+1}(0\vert y1)\Biggr\rvert\nonumber\\
&&\quad +\sum_y\Biggl\lvert \bigl(P_{M+1}(1\vert y1)-P_{M+1}(1\vert y0)\bigr)\times\nonumber\\
&&\qquad\qquad \times\prod_{i=1}^M P_i(y_i\vert y_0y_1\cdots y_{i-1}0)\nonumber\\
&&\qquad +\biggl( \prod_{i=1}^M P_i(y_i\vert y_0y_1\cdots y_{i-1}1)\nonumber\\
&&\qquad\qquad\quad-\prod_{i=1}^M P_i(y_i\vert y_0y_1\cdots y_{i-1}0)\biggr)\times\nonumber\\
&&\qquad\qquad \times P_{M+1}(1\vert y1)\Biggr\rvert\nonumber\\
&\leq&\sum_y C_y,
\end{eqnarray}
where $C_y=\max\{C_y^{-1,-1},C_y^{-1,1},C_y^{1,-1},C_y^{1,1}\}$, and where
\begin{eqnarray}
\label{boundeq9}
C_y^{a,a'}&=&a\Biggl[ \bigl(P_{M+1}(0\vert y1)-P_{M+1}(0\vert y0)\bigr)\times\nonumber\\
&&\qquad\qquad \times\prod_{i=1}^M P_i(y_i\vert y_0y_1\cdots y_{i-1}0)\nonumber\\
&&\qquad +\biggl( \prod_{i=1}^M P_i(y_i\vert y_0y_1\cdots y_{i-1}1)\nonumber\\
&&\qquad\qquad\quad-\prod_{i=1}^M P_i(y_i\vert y_0y_1\cdots y_{i-1}0)\biggr)\times\nonumber\\
&&\qquad\qquad \times P_{M+1}(0\vert y1)\Biggr]\nonumber\\
&&\quad +a'\Biggl[ \bigl(P_{M+1}(1\vert y1)-P_{M+1}(1\vert y0)\bigr)\times\nonumber\\
&&\qquad\qquad \times\prod_{i=1}^M P_i(y_i\vert y_0y_1\cdots y_{i-1}0)\nonumber\\
&&\qquad +\biggl( \prod_{i=1}^M P_i(y_i\vert y_0y_1\cdots y_{i-1}1)\nonumber\\
&&\qquad\qquad\quad-\prod_{i=1}^M P_i(y_i\vert y_0y_1\cdots y_{i-1}0)\biggr)\times\nonumber\\
&&\qquad\qquad \times P_{M+1}(1\vert y1)\Biggr],\nonumber\\
\end{eqnarray}
for $a,a'\in\{0,1\}$. Using that $P_{M+1}(0\vert y\beta)+P_{M+1}(1\vert y\beta)=1$, for all $y$ and $\beta\in\{0,1\}$, we obtain from (\ref{boundeq9}) that
\begin{eqnarray}
\label{boundeq10}
C_y^{a,a}&=&a\Biggl( \prod_{i=1}^M P_i(y_i\vert y_0y_1\cdots y_{i-1}1)\nonumber\\
&&\qquad\qquad-\prod_{i=1}^M P_i(y_i\vert y_0y_1\cdots y_{i-1}0)\Biggr)\nonumber\\
&\leq&\Biggl\lvert \prod_{i=1}^M P_i(y_i\vert y_0y_1\cdots y_{i-1}1)\nonumber\\
&&\qquad\qquad-\prod_{i=1}^M P_i(y_i\vert y_0y_1\cdots y_{i-1}0)\Biggr\rvert,
\end{eqnarray}
for all $y$ and for $a\in\{-1,1\}$. Similarly, we obtain
\begin{eqnarray}
\label{boundeq11}
C_y^{a,-a}&=&2a\bigl( P_{M+1}(0\vert y1)-P_{M+1}(0\vert y0)\bigr)\times\nonumber\\
&&\qquad\times \prod_{i=1}^M P_i(y_i\vert y_0y_1\cdots y_{i-1}0)\nonumber\\
&&\quad+a\bigl(2P_{M+1}(0\vert y1)-1\bigr)\times\nonumber\\
&&\qquad\times\Biggl( \prod_{i=1}^M P_i(y_i\vert y_0y_1\cdots y_{i-1}1)\nonumber\\
&&\qquad\qquad-\prod_{i=1}^M P_i(y_i\vert y_0y_1\cdots y_{i-1}0)\Biggr)\nonumber\\
&\leq&2\bigl\lvert P_{M+1}(0\vert y1)-P_{M+1}(0\vert y0)\bigr\rvert\times\nonumber\\
&&\qquad\times \prod_{i=1}^M P_i(y_i\vert y_0y_1\cdots y_{i-1}0)\nonumber\\
&&\quad+\bigl\lvert 2P_{M+1}(0\vert y1)-1\bigr\rvert \times\nonumber\\
&&\qquad\times\Biggl\lvert \prod_{i=1}^M P_i(y_i\vert y_0y_1\cdots y_{i-1}1)\nonumber\\
&&\qquad\qquad-\prod_{i=1}^M P_i(y_i\vert y_0y_1\cdots y_{i-1}0)\Biggr\rvert\nonumber\\
&\leq&2\bigl\lvert P_{M+1}(0\vert y1)-P_{M+1}(0\vert y0)\bigr\rvert\times\nonumber\\
&&\qquad\times \prod_{i=1}^M P_i(y_i\vert y_0y_1\cdots y_{i-1}0)\nonumber\\
&&\quad+\Biggl\lvert \prod_{i=1}^M P_i(y_i\vert y_0y_1\cdots y_{i-1}1)\nonumber\\
&&\qquad\qquad\!\!\!\!-\prod_{i=1}^M P_i(y_i\vert y_0y_1\cdots y_{i-1}0)\Biggr\rvert,
\end{eqnarray}
for all $y$ and for $a\in\{-1,1\}$. Thus, from (\ref{boundeq8}), (\ref{boundeq10}) and (\ref{boundeq11}), we have
 \begin{eqnarray}
\label{boundeq12}
W&\leq&\sum_y\Bigl[2\bigl\lvert P_{M+1}(0\vert y1)-P_{M+1}(0\vert y0)\bigr\rvert\times\nonumber\\
&&\qquad\times \prod_{i=1}^M P_i(y_i\vert y_0y_1\cdots y_{i-1}0)\Bigr]\nonumber\\
&&\quad+\sum_y\Biggl[\biggl\lvert \prod_{i=1}^M P_i(y_i\vert y_0y_1\cdots y_{i-1}1)\nonumber\\
&&\qquad\qquad-\prod_{i=1}^M P_i(y_i\vert y_0y_1\cdots y_{i-1}0)\biggr\rvert\Biggr]\nonumber\\
&\leq&2B_{M+1} \sum_y\prod_{i=1}^M P_i(y_i\vert y_0y_1\cdots y_{i-1}0)\Bigr] +2\sum_{i=1}^MB_i\nonumber\\
&=&2B_{M+1}\sum_y P(y\vert 0)+2\sum_{i=1}^MB_i\nonumber\\
&=&2\sum_{i=1}^{M+1}B_i,
\end{eqnarray}
as claimed, where in the second line we used (\ref{boundeq5}) with $i=M+1$ and that $\bigl\lvert P_{M+1}(0\vert y1)-P_{M+1}(0\vert y0)\bigr\rvert= \bigl\lvert P_{M+1}(1\vert y1)-P_{M+1}(1\vert y0)\bigr\rvert$, and the assumption that (\ref{boundeq6}) holds for the case $N=M$; where in the third line we used (\ref{boundeq3}); and where in the last line we used $\sum_y P(y\vert 0)=1$.

\subsection{multiphoton attacks on the relativistic quantum bit commitment protocol of Ref. \cite{LCCLWCLLSLZZCPZCP14}}

Ref. \cite{LCCLWCLLSLZZCPZCP14} demonstrated Kent's
  \cite{K12} quantum relativistic protocol for bit commitment. We show
  below that there are multiphoton attacks that apply to the protocol
  of Ref. \cite{LCCLWCLLSLZZCPZCP14}, which imply that this
  protocol is not perfectly hiding.

In the notation of Ref. \cite{LCCLWCLLSLZZCPZCP14} Alice is the
committer. However, we follow our usual convention, taking
Bob to be the committer.

The protocol of Ref. \cite{LCCLWCLLSLZZCPZCP14} is a relativistic
quantum protocol. The attacks discussed below are implemented in the
quantum stage of the protocol, which is nonrelativistic. For this
reason, here we only need to discuss the quantum stage.

The nonrelativistic quantum stage of the protocol of Ref.
\cite{LCCLWCLLSLZZCPZCP14} is equivalent to the quantum stage of the
protocol of Ref. \cite{LKBHTKGWZ13}, discussed in section
\ref{Lunghiprotocol}. Alice and Bob use setup I discussed in the main
text and illustrated in Fig. \ref{fig1} of the main text. Alice's photon source
is a weak coherent source with small average photon number
$\mu$. Alice sends Bob $N$ photon pulses, each encoding in the
polarization a qubit state chosen randomly from the BB84 set. Bob
chooses a random bit $\beta$. Immediately after their reception, Bob
measures each of the $N$ photon pulses in the qubit orthogonal basis
$\mathcal{B}_\beta$, where $\mathcal{B}_0$ and $\mathcal{B}_1$ are the
computational and Hadamard bases, respectively. In order to deal with
losses in the quantum channel, for each pulse sent by Alice, Bob sends
Alice a message $m=1$ if the pulse produced a valid measurement
outcome and $m=0$ otherwise.

We note that in the protocol of
  Ref. \cite{LCCLWCLLSLZZCPZCP14}, there are two setups equivalent to
  setup I (see Fig. \ref{fig1} of the main text) that work in parallel. For
  this reason Bob has four detectors in that protocol. Despite this
  difference with the setup of Ref. \cite{LKBHTKGWZ13}, the analyses
  of section \ref{LKBHTKGWZ13} apply to the protocol of
  Ref. \cite{LCCLWCLLSLZZCPZCP14} too, because Alice may simply apply
  parallel attacks on both setups, or she can choose to apply attacks
  in only one of the setups.

Differently to the protocol of Ref. \cite{LKBHTKGWZ13}, in the
protocol of Ref. \cite{LCCLWCLLSLZZCPZCP14}, Bob's committed bit is
$\beta$. In the multiphoton attacks discussed in section
\ref{LKBHTKGWZ13} for the protocol of Ref. \cite{LKBHTKGWZ13}, Alice
tries to guess Bob's bit $\beta$. Thus, as discussed below, some of
the attacks of section \ref{LKBHTKGWZ13} apply to the protocol of
Ref. \cite{LCCLWCLLSLZZCPZCP14}.

Unlike Ref. \cite{LKBHTKGWZ13},
Ref. \cite{LCCLWCLLSLZZCPZCP14} does not implement the symmetrization
of losses strategy. Instead, Ref. \cite{LCCLWCLLSLZZCPZCP14} applies
the countermeasure in which Bob reports a valid measurement outcome,
i.e. sets $m=1$, for each pulse for which at least one detector
clicks. Thus, multiphoton attack I does not apply.

However, the versions of the multiphoton attack II discussed in
sections \ref{app1} and \ref{double} apply to the protocol of
Ref. \cite{LCCLWCLLSLZZCPZCP14}. This is because it cannot be
guaranteed that Bob's detection efficiencies are exactly equal. In the
actual experiment, the measured values for Bob's detection
efficiencies were $50.4\pm1.1\%$, $50.4\pm0.4\%$, $52.4\pm1.0\%$ and
$50.2\pm1.1\%$ \cite{YangLiu}. However,
  Ref. \cite{LCCLWCLLSLZZCPZCP14} also makes the important point that
  Bob could include calibrated attenuators in front of his detectors
  to make their detection efficiencies equal, in order to defend
  against attacks exploiting unequal detector
  efficiencies.    As far as we are aware, this is the
  first discussion of this valuable countermeasure (albeit not in the
  context of the specific attacks we consider here).    We
  note, though, that even if this countermeasure were
  implemented, experimental uncertainties  would remain.  
  There is no way to guarantee that Bob's detector efficiencies are exactly equal for a
  given frequency, nor that they are precisely
    frequency-independent or that Alice uses precisely the stipulated
  frequency.

The upper bound on Alice's probability to guess Bob's bit $\beta$
derived in section \ref{bounds} assumes that Bob sets $m=1$ for pulses
producing a click in at least one of Bob's detectors and that Bob's
detector efficiencies are not exactly equal. Thus, this bound also
applies to the protocol of Ref. \cite{LCCLWCLLSLZZCPZCP14}. However,
this particular bound does not appear strong enough to provide any
security guarantee for the parameters of Ref. \cite{LCCLWCLLSLZZCPZCP14}.
It would also be very interesting to explore whether
practical security could be attained by a combination of the 
use of attenuators and other methods, using the results of section
\ref{bounds}. For example, Alice
could be required to use a single photon source, and Bob
could test a subset of the received pulses to ensure that
the number of multiphoton pulses is suitably small.

\subsection{A multiphoton attack on the quantum coin flipping protocol of Ref. \cite{PJLCLTKD14}}
\label{coinflipattack}

\subsubsection{The quantum coin flipping protocol of Ref. \cite{PJLCLTKD14}}

Ref. \cite{PJLCLTKD14} demonstrated a nonrelativistic quantum strong
coin flipping protocol and argued that the implementation guaranteed  with information theoretic security an upper bound $p_c$ on the probability
that an all-powerful malicious party can bias the coin. This bound 
depends on various
experimental parameters, including the honest abort probability and
the distance between Bob and Alice, and is strictly smaller than
unity for some parameter values (see Figure 2 of
Ref. \cite{PJLCLTKD14}, for example).

Ref. \cite{PJLCLTKD14}
does not say whether a double detection event is registered as a
measurement outcome, and if so, how it is registered. Assuming that Bob only assigns single clicks as valid measurement outcomes, we discuss below how a version of the multiphoton attack I discussed in the main text applies to the protocol of Ref. \cite{PJLCLTKD14}, where Alice succeeds in setting the outcome of the coin flip with probability very close to unity.

As discussed below, if Bob reports single and double clicks as valid measurement outcomes, he can only reduce Alice's success probability in this attack to a value slightly below $\frac{7}{8}$. It is not clear whether Ref. \cite{PJLCLTKD14} obtained an upper bound on the probability $P_\text{Alice}$ that dishonest Alice can bias the coin flip that is below $\frac{7}{8}$ for any values of the possible experimental parameters, which would imply that this attack implies violation of the bound. In particular, Figure 2 of Ref. \cite{PJLCLTKD14} suggests that the smallest value for this bound demonstrated experimentally was approximately 0.91, achieved for a communication distance of 15 km. Nevertheless, Equation (1) in the Supplementary Information of Ref. \cite{PJLCLTKD14} states that $P_\text{Alice}\leq \frac{3}{4}+\frac{1}{2}\sqrt{y(1-y)}$, where the parameter $y$ satisfies $y\in(\frac{1}{2},1)$. There exist $y\in(\frac{1}{2},1)$ for which $\frac{3}{4}+\frac{1}{2}\sqrt{y(1-y)}<\frac{7}{8}$. Thus, if any parameter $y$ from the range $(\frac{1}{2},1)$ can be obtained in principle experimentally then our attack implies a violation of the security bound derived in Ref. \cite{PJLCLTKD14}.

The quantum strong coin flipping protocol of Ref. \cite{PJLCLTKD14} is
the following. First, Alice sends to Bob $N$ photon pulses prepared in
states $\lvert \Phi_{\alpha_i,\gamma_i}\rangle$, encoding the bit $\gamma_i$ in
a basis labelled by the bit $\alpha_i$, and their labels $i$, for
$i\in[N]=\{1,\ldots,N\}$ and for a predetermined integer $N$. The
states are given by
$\lvert \Phi_{\alpha_i,0}\rangle=\sqrt{y}\lvert 0\rangle
+(-1)^{\alpha_i}\sqrt{1-y}\lvert 1\rangle$
and
$\lvert \Phi_{\alpha_i,1}\rangle=\sqrt{1-y}\lvert 0\rangle
-(-1)^{\alpha_i}\sqrt{y}\lvert 1\rangle$,
for $\alpha_i\in\{0,1\}$, for some predetermined parameter
$y\in\bigl(\frac{1}{2},1\bigr)$. We note that
$\mathcal{B}_{\alpha_i}=\{\lvert \Phi_{\alpha_i,r}\rangle\}_{r=0}^1$ forms an
orthonormal basis, for $\alpha_i\in\{0,1\}$. Second, Bob measures the
$i$th pulse in the basis $\mathcal{B}_{\beta_i}$, where the bit $\beta_i$ is
chosen randomly by Bob, and obtains some bit outcome $o_i$ or does not
register any measurement outcome, for $i\in[N]$. We note that due to
losses and detection efficiencies smaller than unity, not
all pulses produce a detection in Bob's detectors, i.e. Bob does not obtain a measurement outcome for
all $i\in[N]$. Let $j\in[N]$ be the first pulse for which Bob
registers a measurement outcome, which is a bit $o_j$. Third, Bob
generates a random bit $b$ and communicates $j$ and $b$ to Alice. Fourth, Alice communicates $\alpha_j$ and $\gamma_j$ to Bob. Fifth,
Bob checks whether $\alpha_j=\beta_j$, in which case Bob aborts if
and only if $o_j\neq \gamma_j$. If $\alpha_j\neq \beta_j$, Bob does not
abort. If Bob does not abort then Bob and Alice agree that the outcome
of the coin flip is $a=\gamma_j\oplus b$, i.e. the XOR of the bit $b$ given
by Bob and the bit $\gamma_j$ given by Alice.  Ref. \cite{PJLCLTKD14}  implemented the technique of symmetrization of losses presented in Ref. \cite{NJMKW12} and discussed in the main text.

In the attack we describe below, we assume that in the protocol of Ref. \cite{PJLCLTKD14} double clicks are not assigned as valid measurement outcomes by Bob and that the symmetrization of losses technique is implemented as described in the main text. Thus, we assume that Bob assigns a valid measurement outcome to the detection event $(c,\bar{c})$ with some probability $S_{c\bar{c}\beta}>0$ when he measured a pulse in the basis $\mathcal{B}_\beta$, in such a way that $S_{c\bar{c}\beta}$ satisfies the conditions of Lemma \ref{lemma1} in the main text, for $c,\beta\in\{0,1\}$. Let $S_\text{min}=\min\{S_{c\bar{c}\beta}\}_{c,\beta\in\{0,1\}}$. In particular, it is stated by Ref. \cite{PJLCLTKD14} that there was not any important deviation between the frequencies of detection events for the two measurement bases applied by Bob. However, Ref. \cite{PJLCLTKD14} reports a ratio of approximately 0.68 in the number of detections observed by Bob in his two detectors, giving a value of $S_\text{min}=0.68$. Thus, in the protocol implemented with Alice, Bob assigned valid measurement outcomes to detection events of his detector with higher number of detections with a probability of 0.68, and to detection events of his detector with lower number of detections with unit probability.

The protocol uses a slight variation of setup I given in the main text (see Fig. \ref{fig1} of the main text). The experimental setup of Ref. \cite{PJLCLTKD14} consists in a plug and play system with a two-way approach: light pulses of 1550 nm are sent from Bob to Alice, and Alice uses a phase modulator to prepare the qubit states. The pulses are then reflected using a Faraday mirror and attenuated to the desired average photon number before being sent back to Bob. Finally, Bob uses a phase modulator to choose his measurement basis and register the detection events using two threshold single photon detectors. In practice, we can suppose that the experimental setup is equivalent to setup I given in Fig. \ref{fig1} of the main text.

\subsubsection{multiphoton attack I}
In order to present an attack to the quantum coin flipping protocol
above, we first make some observations. As in the experimental
demonstration of Ref. \cite{PJLCLTKD14}, we assume that Bob uses
threshold single photon detectors for registering the outcome of his
measurement. A measurement by Bob uses two single photon detectors
D$_0$ and D$_1$. An outcome $o_i$ corresponds to a detector D$_i$
registering a detection, for D$_i\in\{0,1\}$. 

Assuming that Bob only assigns valid measurement outcomes to detection events in which only one of his detectors click and that he applies the symmetrization of losses technique as described above, we present a version of the multiphoton attack I presented in the main text to the protocol of Ref. \cite{PJLCLTKD14}. We assume that Bob
is honest and that Alice is cheating. Alice chooses a set
$\Omega\subseteq [N]$ with $M$ elements, for some positive integer
$M<N$. For $i\notin\Omega$, Alice does not send any pulses to Bob. For
$i\in\Omega$, Alice sends a pulse with $k$ photons in the state
$\lvert \Phi_{\alpha_i,\gamma_i}\rangle$, where the bits $\alpha_i$ and
$\gamma_i$ are chosen randomly by Alice. Alice chooses her
desired coin flip outcome $a$.  Bob sends the index $j\in[N]$
and the bit $b$ to Alice. Let $\tilde{\alpha}_j$ and $\tilde{\gamma}_j$ be
the bits that Alice sends Bob, indicating her supposed chosen basis
and encoded bit, respectively.  Alice sets $\tilde{\gamma}_j=b\oplus a$. If
$\tilde{\gamma}_j=\gamma_j$ then Alice sets $\tilde{\alpha}_j=\alpha_j$,
otherwise she sets $\tilde{\alpha}_j=\alpha_j\oplus 1$. Alice sends
$\tilde{\alpha}_j$ and $\tilde{\gamma}_j$ to Bob.

We show that this attack allows Alice to succeed with very high
probability. First, we note that if Bob does not abort and if he assigns at least one measurement outcome, the outcome of
the coin flip is $\tilde{\gamma}_j\oplus b$, which equals the bit $a$
chosen by Alice. It remains to be shown that the probability that Bob
aborts or that he does not assign any measurement outcome is negligible. In order to illustrate the attack more easily,
consider an ideal situation in which there are not dark counts at
Bob's detectors, Alice's state preparations are perfect and Bob's
measurement outcomes do not have errors. Let us also consider the
limit $k\rightarrow \infty$, i.e. Alice prepares pulses with an
infinite number of photons. Thus, we see that Bob can only register
detection events at his detectors for pulses with labels from the set
$\Omega$. For $i\in\Omega$, if $\beta_i\neq\alpha_i$, i.e. if Bob
measures in a basis different to the one of preparation by Alice, then
an infinite number of photons go to the detector D$_0$ and an infinite
number of photons go to the detector D$_1$; hence, both detectors
register a detection, and Bob does not assign a measurement outcome to
the pulse with label $i$. On the other hand, for $i\in\Omega$, if
$\beta_i=\alpha_i$, i.e. if Bob measures in the basis of preparation
by Alice, then all photons go to the detector D$_{\gamma_{i}}$; hence, the
detector D$_{\gamma_{i}}$ registers a detection, and Bob assigns the
measurement outcome $o_i=\gamma_i$ with a probability greater or equal than $S_\text{min}$. Thus, the first pulse for which Bob
registers a measurement outcome, which has label $j$, satisfies
$j\in\Omega$, $\beta_j=\alpha_j$ and
$o_j=\gamma_j$.
The probability for this to hold is the probability that Bob chooses
$\beta_i=\alpha_i$ and Bob assigns the detector click as a valid measurement outcome for at least one $i\in\Omega$, which is easily seen to be greater or equal than $1-\bigl(1-\frac{S_\text{min}}{2}\bigr)^M$. Therefore, we see that in the case that $\tilde{\gamma}_j=\gamma_j$,
Alice reports $\tilde{\alpha}_j=\alpha_j=\beta_j$
and $\tilde{\gamma}_j=\gamma_j=o_j$;
hence, Bob does not abort.  In the case that $\tilde{\gamma}_j=\gamma_j\oplus
1$, Alice reports $\tilde{\alpha}_j=\alpha_j\oplus 1=\beta_j\oplus
1$ and $\tilde{\gamma}_j=\gamma_j\oplus 1=o_j\oplus
1$; hence, Bob does not abort, as in this case
$\beta_j=\tilde{\alpha}_j\oplus
1$. We see that Bob does not abort in any case, and Alice succeeds in
her cheating strategy if Bob chooses $\beta_i=\alpha_i$ and Bob assigns the detector click as a valid measurement outcome
for at least one $i\in\Omega$,
which occurs with probability greater or equal than $1-\bigl(1-\frac{S_\text{min}}{2}\bigr)^M$. Thus, since $S_\text{min}>0$, for $M$
large enough, Alice can make his cheating probability arbitrarily
close to unity.

In practice, Alice cannot send pulses with infinite number of photons,
the state preparation and measurements will have some errors, Bob's
detectors have nonzero dark count probabilities. However, it is easy
to see that an adaptation of the attack discussed above can be applied
in a realistic scenario. For example, given the parameters of the
protocol agreed by Alice and Bob, Alice can simulate Bob's detection
probabilities and choose the set $\Omega$ to lie within the first few
pulses, and pulses with labels not from the set $\Omega$ can be chosen
by Alice to have specific average number of photons, in such a way
that Bob's detection probabilities are very close to what is expected
in the agreed protocol. In practice, $N$ is chosen very large, for
example $N$ is of the order of $10^{10}$ in the experimental
demonstration of Ref. \cite{PJLCLTKD14}. Thus, Bob can choose $M$ to
be large but small compared to $N$, for example, $M=40$, which gives
$\bigl(1-\frac{S_\text{min}}{2}\bigr)^M=6\times 10^{-8}$, for the value $S_\text{min}=0.68$ reported by Ref. \cite{PJLCLTKD14}. Alice cannot
set the number of photons $k$ of the dishonest pulses to be infinite,
but she can set the average photon number to a value $\mu>>1$ with a
coherent source, as in the attack that we have simulated
experimentally and discussed in the main text, for example.

\subsubsection{multiphoton attack II}

As noted above, multiphoton attack I would apply if Bob discarded double clicks as invalid measurements. We thus assume that Bob assigns
%Bob can apply a partial countermeasure against \ddd{the previous} attack by assigning
a random measurement outcome to detection events in which
both of his detectors register a detection. It is straightforward to
see that in the ideal attack we discussed above with $M=1$ and assuming there is no symmetrization of losses (i.e. that $S_\text{min}=1$), with this partial
countermeasure, Bob aborts with probability $\frac{1}{8}$,
corresponding to the case $\beta_j=\alpha_j\oplus 1=\tilde{\alpha}_j$,
$\gamma_j=b\oplus a\oplus 1=\tilde{\gamma}_j\oplus 1$ and
$o_j=\gamma_j=\tilde{\gamma}_j\oplus 1$. Thus, in this case, assuming that symmetrization of losses is implemented, the probability that Bob does not abort and assigns a valid measurement outcome is equal or greater than $\frac{7}{8}S_\text{min}$. Thus, in situations where $S_\text{min}$ is close to unity, Alice can succeed with this attack with probability close to $\frac{7}{8}$. It does not help Alice in this case to set $M>>1$ , as in this case it is not difficult to see that the probability that Bob aborts is equal or greater than $1-\Bigl(1-\frac{S_\text{min}}{8}\Bigr)^M$, which approaches unity as $M$ increases.

\subsubsection{Partial countermeasures and more general attacks}

In the attacks discussed above Alice does not send any pulses for $i\notin \Omega$. But this could be easily discovered by Bob if $M=\lvert \Omega\rvert $ is small, as for the case $M=1$. Alice can refine her attacks by sending pulses with the statistics of the protocol agreed with Bob for $i\notin\Omega$ and to define $\Omega=\{1,2,\ldots,M\}$, i.e to let $\Omega$ correspond to the first $M$ pulses. The attacks presented above apply straightforwardly in this case and the lower bounds on Alice's probabilities to succeed in these attacks derived above hold in this case too.

In addition to reporting single and double clicks, a possible extra countermeasure by Bob to apply against the previous attacks is to choose the index $j$ randomly from the set $[N]$, instead of choosing $j$ as the index of the first successfully measured pulse, as in the protocol of Ref. \cite{PJLCLTKD14}. In this way, Alice cannot effectively choose a small set $\Omega$ for her dishonest multiphoton pulses that guarantees with high probability that $j\in\Omega$. On the other hand if Alice chooses a large set $\Omega$ then Bob aborts with high probability, as discussed above. 

Broadly, what the previous attack illustrates is that, because Alice
has the ability to prepare pulses with many photons and because Bob's
single photon detectors are threshold detectors, the probabilities of
Bob's measurement outcomes are not in general proportional to
$\langle \Phi_{\alpha_i,\gamma_i}\rvert \rho \lvert
\Phi_{\alpha_i,\gamma_i}\rangle$,
for photon pulses prepared by Alice to encode a qubit state $\rho$ in
each photon,
contrary to the assumption made in the security proof of
Ref. \cite{PJLCLTKD14}.

\subsection{A multiphoton attack on the quantum bit commitment protocol of Ref. \cite{NJMKW12}}

\label{NJMKW12}

\subsubsection{The quantum bit commitment protocol of Ref. \cite{NJMKW12}}

Ref. \cite{NJMKW12} demonstrated quantum bit commitment in the noisy storage model. Below, we present an attack to this protocol and show that it is insecure: we show that it is not hiding.

In the protocol of Ref. \cite{NJMKW12} Alice is the committer. Although this notation is different to the one introduced in section \ref{QBC}, we follow this notation here because it is consistent with an extension of setup I, which we introduce below. The protocol of Ref. \cite{NJMKW12} includes a subroutine that uses a setup different to the setup I described in the main text. We call this \emph{setup II} (see Fig. \ref{setupII}).

\begin{figure}
\includegraphics[scale=0.49]{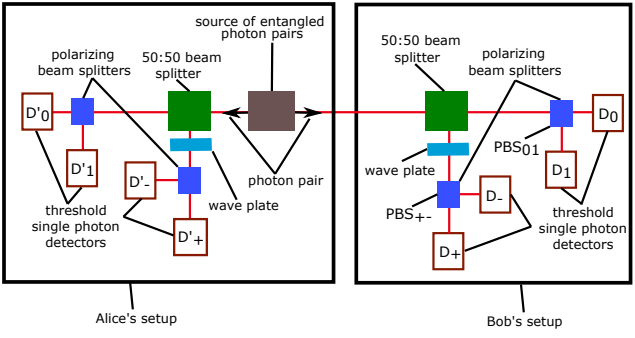}
\caption{\label{setupII} \textbf{Setup II.} Alice's setup consists in an
  source of entangled photon pairs, a 50:50 beam splitter, a half-wave plate, two polarizing beam splitters and four threshold single photon detectors D'$_0$, D'$_1$, D'$_+$ and D'$_-$. Bob's setup
  consists in a 50:50 beam splitter, a wave plate, two polarizing beam splitters (PBS$_{01}$ and PBS$_{+-}$) and four threshold single photon detectors D$_0$, D$_1$, D$_+$ and D$_-$.}
\end{figure}

In setup II, Alice has a source of pair of entangled photons. For each pair generated by Alice, she sends a photon to Bob and she measures the other photon randomly in one of two qubit orthogonal bases, $\mathcal{B}_0$ and $\mathcal{B}_1$. Bob also measures his photon randomly in the bases $\mathcal{B}_0$ and $\mathcal{B}_1$. Without loss of generality we suppose that $\mathcal{B}_0=\{\lvert 0\rangle,\lvert 1\rangle\}$ is the computational basis and $\mathcal{B}_1=\{\lvert \tilde{+}\rangle,\lvert \tilde{-}\rangle\}$ is a qubit orthogonal basis where the Bloch vector of the qubit state $\lvert\tilde{+}\rangle$ has an angle $\theta\in[0,\frac{\pi}{2})$ from the $x$ axis towards the $z$ axis in the Bloch sphere.  At each site, the random measurement is implemented with a $50:50$ beam splitter, followed by two polarizing beam splitters and four single photon detectors.  The quantum channel between the $50:50$ beam splitter and one of the polarizing beam splitters contains a wave plate that rotates the polarization an angle $\frac{\pi}{2}-\theta$ from the $z$ axis towards the $x$ axis in the Bloch sphere. Thus, one of the polarizing beam splitters measures the received photon pulse in the basis $\mathcal{B}_0$ and the other one measures the received photon pulse in the basis $\mathcal{B}_1$. In particular, Bob's polarizing beam splitters PBS$_{01}$ and PBS$_{+-}$ implement qubit measurements in the bases $\mathcal{B}_0$ and $\mathcal{B}_1$, respectively. The case $\theta=0$ corresponds to $\mathcal{B}_0$ and $\mathcal{B}_1$ being the computational and Hadamard bases, respectively, which is the particular case considered in Ref. \cite{NJMKW12}.

Alice and Bob have each four single photon detectors, each one corresponding to one of the four generated states.  Let D$_0$, D$_1$, D$_+$ and D$_-$ be Bob's detectors corresponding to the states $\lvert 0\rangle$,  $\lvert 1\rangle$, $\lvert \tilde{+}\rangle$ and $\lvert \tilde{-}\rangle$, respectively. The detectors are threshold detectors, i.e. they cannot distinguish the number of photons producing a detection. 

Let $0\leq d_i<<1$ and $\eta_i\in(0,1)$ be the dark count probability and the detection efficiency of Bob's detector D$_i$, respectively, which we assume are known by Alice, for $i\in\{0,1,+,-\}$. Let $\eta_\text{min}=\min\{\eta_0,\eta_1,\eta_+,\eta_-\}$.

It is explicitly stated in Ref. \cite{NJMKW12} that only detection events where a single detector clicks are considered as valid measurement outcomes.
Thus, in a valid round a pair of photon pulses generated by Alice produces a single click in one of Alice's detectors and a single click in one of Bob's detectors.
In order to deal with losses in the quantum channel, for each photon pulse sent by Alice, Bob sends a message $m=1$ to Alice indicating that a single click in one of his detectors is produced or a message $m=0$ indicating the opposite. The symmetrization of losses technique is applied by Bob. More precisely, the following reporting strategy is applied by Bob in Ref. \cite{NJMKW12}.

\begin{definition}[Symmetrization of losses in setup II (SLII)]
\label{symmetrizationsetup2}
Bob tests his setup by preparing and measuring states as in the
protocol agreed with Alice, a large number of times $N$ in parallel. Then, for $i\in\{0,1,+,-\}$, Bob computes the frequency $F_i$ of detection events in which only the detector D$_i$ clicks, which provides a good estimate of the corresponding probability $P_i$ if $NF_i>>1$. Bob then
computes numbers $S_i\in(0,1]$ satisfying
\begin{equation}
\label{SLIIeq}
S_i F_i=F_\text{min},
\end{equation}
for $i\in\{0,1,+,-\}$, where $F_\text{min}=\min\{F_0,F_1,F_+,F_-\}$. Bob sends Alice the message $m=1$ and assigns a valid measurement outcome corresponding to the detector D$_i$ with a probability $S_i$ if only the detector D$_i$ clicks, for $i\in\{0,1,+,-\}$.
\end{definition}

Similarly to Lemma \ref{lemma1} of the main text, the following lemma
shows that the \hyperref[symmetrizationsetup2]{SLII}
reporting strategy guarantees to Bob that Alice cannot obtain any
information about Bob's assigned measurement basis $\mathcal{B}_\beta$
if Alice's pulse does not have more than one photon and $d_{i}=0$, for
$\beta\in\{0,1\}$, for arbitrary $\eta_{i}\in(0,1)$ and for
$i\in\{0,1,+,-\}$. Furthermore, it guarantees that Alice cannot obtain
much information about $\beta$ if Alice's pulse does not have more
than one photon and $0<d_{i}\leq \delta$, for $0<\delta<< 1$ and
$i\in\{0,1,+,-\}$.

\begin{lemma}
\label{lemmaSLII}
Consider setup II illustrated in Fig. \ref{setupII}. Let
$d_{i}\leq \delta$, for some $0\leq \delta<1$ and for
$i\in\{0,1,+,-\}$. Suppose that Alice sends Bob a single photon pulse ($k=1$)
in arbitrary qubit state $\rho$ or that she does not send him any
photons ($k=0$).  Suppose that Bob applies the \hyperref[symmetrizationsetup2]{SLII} reporting strategy with 
\begin{equation}
\label{SLIIeq2} 
S_i=\frac{\eta_\text{min}}{\eta_i},
\end{equation}
for $i\in\{0,1,+,-\}$. Then 
\begin{equation}
\label{a3new}
\bigl\lvert P_{\text{report}}(1,1\lvert \rho, k) -P_{\text{report}}(1,0\lvert  \rho, k)\bigr\rvert\leq 6\delta,
\end{equation}
for $k\in\{0,1\}$, where $P_{\text{report}}(1,\beta\lvert \rho, k)$ is the probability that Bob reports the message $m=1$ to Alice and assigns a measurement in the basis $\mathcal{B}_\beta$, for $\beta\in\{0,1\}$.
\end{lemma}

Note that we have slightly changed the notation $P_{\text{report}}(1\lvert\beta, \rho, k)$ used for setup I to $P_{\text{report}}(1,\beta\lvert \rho, k)$ for setup II. This is because in setup I Bob chooses the measurement basis by appropriately setting the wave plate. However, in setup II, the assigned measurement basis is not chosen by Bob; it is an outcome, which depends on the detectors that are activated.

The proof of Lemma \ref{lemmaSLII} uses Lemma \ref{lemmasetup2}, given below. Thus, it is presented below.

As stated in Ref. \cite{NJMKW12}, it is a necessary condition for security against Alice that she cannot learn any information about the measurement bases obtained by Bob. It is claimed by Ref. \cite{NJMKW12} that the \hyperref[symmetrizationsetup2]{SLII} reporting strategy guarantees this condition. Below we show that this claim is wrong: that Alice can obtain a lot of information about Bob's measurement bases.

\subsubsection{multiphoton attack I}

We describe a version of multiphoton attack I that applies in setup II. The attack works equally well if Bob's setup is the same as in setup II, independently of whether Alice's photon source is a source of pairs of entangled photons or not. What matters is that Alice sends Bob photon pulses and that Bob measures randomly in one of two bases, $\mathcal{B}_0$ and $\mathcal{B}_1$, using Bob's setup of four threshold single photon detectors illustrated in Fig. \ref{setupII}.

\begin{definition}[multiphoton attack I in setup II (MPAII)]
\label{MPAII}
Suppose that Alice and Bob use setup II, illustrated in Fig. \ref{setupII}.
Alice prepares a \emph{dishonest pulse} of $k$ photons encoding a $k-$qubit state $\rho$ in the polarization degrees of freedom, for some nonnegative integer $k$ chosen by Alice. The quantum state $\rho$ is chosen by Alice and can be an arbitrary entangled state, which can also be entangled with an ancilla held by Alice. When specified, we will consider the particular case $\rho=\rho_\text{qubit}^{\otimes k}$ in which the dishonest pulse consist in $k$ photons, each of them encoding the qubit state $\rho_\text{qubit}$ with Bloch vector $\vec{r}=(r_x,r_y,r_z)$. Let $P_\text{report}(1,\beta\vert\rho,k)$ be the probability that Bob reports the message $m=1$ to Alice and that Bob assigns a valid measurement outcome in the basis $\mathcal{B}_\beta$, for $\beta\in\{0,1\}$. We assume that Alice knows the value of $P_\text{report}(1,\beta\vert\rho,k)$, for $\beta\in\{0,1\}$. If $P_\text{report}(1,0\vert\rho,k)\geq P_\text{report}(1,1\vert\rho,k)$, and if Bob reports to Alice the message $m=1$ , Alice guesses that Bob's obtained basis is $\mathcal{B}_0$. On the other hand, if $P_\text{report}(1,1\vert\rho,k)> P_\text{report}(1,0\vert\rho,k)$, and if Bob reports to Alice the message $m=1$, Alice guesses that Bob's obtained basis is $\mathcal{B}_1$. Thus, the probability that Alice guesses Bob's obtained basis $\mathcal{B}_\beta$ for the dishonest pulse is given by
\begin{equation}
\label{guessing}
P_{\text{guess}}=\frac{\max_{\beta\in\{0,1\}}\{P_\text{report}(1,\beta\vert\rho,k)\}}{P_\text{report}(1,0\vert \rho,k)+P_\text{report}(1,1\vert\rho,k)}.
\end{equation}
\end{definition}

We note that we must condition on Bob assigning a measurement outcome, either in the basis $\mathcal{B}_0$ or in the basis $\mathcal{B}_1$. We show in Lemma \ref{guessingprobability} below that $P_\text{report}(1,0\vert\rho,k)\neq P_\text{report}(1,1\vert \rho,k)$ for a range of parameters. Thus, it follows from (\ref{guessing}) that $P_{\text{guess}}>\frac{1}{2}$. We also show that for a range of parameters it holds that $P_{\text{guess}}\rightarrow 1$ if $k\rightarrow \infty$. Thus, Alice can guess Bob's assigned measurement basis for the dishonest pulse with great probability, which violates security against Alice.

The \hyperref[MPAII]{MPAII} attack can be easily extended as follows. Let $N$ be the number of photon pulses that Alice sends Bob in the predetermined protocol. Alice chooses a subset $\Omega$ of labels for these $N$ pulses and prepares a dishonest pulse as in the \hyperref[MPAII]{MPAII} attack if its label belongs to $\Omega$. For each pulse with label from the set $\Omega$, Alice guesses the measurement basis assigned by Bob as in the \hyperref[MPAII]{MPAII} attack.

Let $P_\text{det}(c_0c_1c_+c_-\vert \rho,k)$ be the probability that the detector D$_0$ clicks if $c_0=1$ and does not click if $c_0=0$, that the detector D$_1$ clicks if $c_1=1$ and does not click if $c_1=0$, that the detector D$_+$ clicks if $c_+=1$ and does not click if $c_+=0$, and that the detector D$_-$ clicks if $c_-=1$ and does not click if $c_-=0$, for $c_0,c_1,c_+,c_-\in\{0,1\}$, when Alice sends Bob a photon pulse of $k$ photons encoding the $k-$qubit state $\rho$ in the polarization degrees of freedom. 

In order to present Lemma \ref{guessingprobability}, we first need to introduce the following lemma.

\begin{lemma}
\label{lemmasetup2}
Consider Alice's \hyperref[MPAII]{MPAII} attack, which uses setup II illustrated in Fig. \ref{setupII}, in the particular case $\rho=\rho_\text{qubit}^{\otimes k}$ for some qubit state $\rho_\text{qubit}$ with Bloch vector $\vec{r}=(r_x,r_y,r_z)$. Suppose that Bob uses the \hyperref[symmetrizationsetup2]{SLII} reporting strategy, with $S_i=\frac{\eta_\text{min}}{\eta_i}$, for $i\in\{0,1,+,-\}$. The probability $P_\text{report}(1,\beta\vert\rho,k)$ that Bob sends the message $m=1$ to Alice and assigns a measurement outcome in the basis $\mathcal{B}_\beta$, for $\beta\in\{0,1\}$, is given by
\begin{eqnarray}
\label{bases}
P_\text{report}(1,0\vert\rho,k)&=&S_0P_0+S_1P_1,\nonumber\\
P_\text{report}(1,1\vert \rho,k)&=&S_+P_++S_-P_-,
\end{eqnarray}
where 
\begin{eqnarray}
\label{clicks}
P_0&=&(1-d_1)(1-d_+)(1-d_-)\Biggl\{
\biggl[1-\frac{1}{2}\bigl[\eta_-+(1-q_0)\eta_1\nonumber\\
&&\qquad+q_+(\eta_+-\eta_-)\bigr]\biggr]^k
-(1-d_0)\biggl[1-\frac{1}{2}\bigl[\eta_1+\eta_-\nonumber\\
&&\qquad\qquad+q_0(\eta_0-\eta_1)+q_+(\eta_+-\eta_-)\bigr]\biggr]^k
\Biggr\},\nonumber\\
P_1&=&(1-d_0)(1-d_+)(1-d_-)\Biggl\{
\biggl[1-\frac{1}{2}\bigl[\eta_-+q_0\eta_0\nonumber\\
&&\qquad+q_+(\eta_+-\eta_-)\bigr]\biggr]^k
-(1-d_1)\biggl[1-\frac{1}{2}\bigl[\eta_1+\eta_-\nonumber\\
&&\qquad\qquad+q_0(\eta_0-\eta_1)+q_+(\eta_+-\eta_-)\bigr]\biggr]^k
\Biggr\},\nonumber\\
P_+&=&(1-d_0)(1-d_1)(1-d_-)\Biggl\{
\biggl[1-\frac{1}{2}\bigl[\eta_1+q_0(\eta_0-\eta_1)\nonumber\\
&&\qquad+(1-q_+)\eta_-\bigr]\biggr]^k
-(1-d_+)\biggl[1-\frac{1}{2}\bigl[\eta_1+\eta_-\nonumber\\
&&\qquad\qquad+q_0(\eta_0-\eta_1)+q_+(\eta_+-\eta_-)\bigr]\biggr]^k
\Biggr\},\nonumber\\
P_-&=&(1-d_0)(1-d_1)(1-d_+)\Biggl\{
\biggl[1-\frac{1}{2}\bigl[\eta_1+q_0(\eta_0-\eta_1)\nonumber\\
&&\qquad+q_+\eta_+\bigr]\biggr]^k
-(1-d_-)\biggl[1-\frac{1}{2}\bigl[\eta_1+\eta_-\nonumber\\
&&\qquad\qquad+q_0(\eta_0-\eta_1)+q_+(\eta_+-\eta_-)\bigr]\biggr]^k
\Biggr\},
\end{eqnarray}
and where
\begin{eqnarray}
\label{qs}
q_0&=&\frac{1}{2}\bigl(1+r_z\bigr),\nonumber\\
q_+&=&\frac{1}{2}\bigl(1+r_x\cos\theta +r_z\sin\theta\bigr).
\end{eqnarray}
\end{lemma}

We note that we do not lose generality by considering that $\mathcal{B}_0$ is the computational basis and $\mathcal{B}_1$ is a qubit orthogonal basis in the $x-z$ plane in the Bloch sphere. This is because any pair of qubit orthogonal bases $\mathcal{B}_0$ and $\mathcal{B}_1$ define a plane in the Bloch sphere. Without loss of generality we can take this plane to be the $x-z$ plane. Then, without loss of generality we can also suppose that $\mathcal{B}_0$ is the basis along the $z$ axis, i.e. the computational basis.

\begin{proof}[Proof of Lemma \ref{lemmasetup2}]
Consider Bob's setup in Fig. \ref{setupII}. In the considered \hyperref[MPAII]{MPAII} attack, Alice sends Bob a pulse of $k$ photons, where each photon encodes the qubit state $\rho_\text{qubit}$, for some nonnegative integer $k$. That is, $\rho=\rho_\text{qubit}^{\otimes k}$. Let $k_{01}$ be the number of photons that are transmitted through the 50:50 beam splitter towards the polarizing beam splitter PBS$_{01}$ and let $k_{+-}=k-k_{01}$ be the number of photons that are reflected from the 50:50 beam splitter towards the polarizing beam splitter PBS$_{+-}$, for $k_{01}\in\{0,1,\ldots,k\}$. Let $k_0$ and $k_1=k_{01}-k_0$ be the number of photons that go towards the detectors D$_0$ and D$_1$, respectively, for $k_0\in\{0,1,\ldots,k_{01}\}$. Let $k_+$ and $k_-=k-k_{01}-k_+$ be the number of photons that go towards the detectors D$_+$ and D$_-$, respectively, for $k_+\in\{0,1,\ldots,k-k_{01}\}$.

The probabilities that a photon is transmitted through the 50:50 beam splitter towards the polarizing beam splitter PBS$_{01}$ and reflected from the 50:50 beam splitter towards the polarizing beam splitter PBS$_{+-}$ are both $\frac{1}{2}$. We note that we do not lose generality by supposing that the 50:50 beam splitter has transmission and reflection probabilities exactly equal to $\frac{1}{2}$. If these probabilities were different, these values could be absorbed in the efficiencies of the detectors, leaving the equivalent transmission and reflection probabilities of the 50:50 beam splitter effectively equal to $\frac{1}{2}$. The probabilities that a photon directed towards the polarizing beam splitter PBS$_{01}$ goes to the detectors D$_0$ and D$_1$ are $q_0=\langle 0\rvert \rho_\text{qubit}\lvert 0\rangle= \frac{1}{2}\bigl(1+r_z\bigr)$ and $q_1=1-q_0$, respectively, where $r_z$ is the $z$ component of the Bloch vector $\vec{r}$ of the qubit state $\rho$. The probabilities that a photon directed towards the polarizing beam splitter PBS$_{+-}$ goes to the detectors D$_+$ and D$_-$ are $q_+=\langle \tilde{+}\rvert \rho_\text{qubit}\lvert \tilde{+}\rangle= \frac{1}{2}\bigl(1+r_x\cos\theta+r_z\sin\theta\bigr)$ and $q_-=1-q_+$, respectively, where $r_x$ is the $x$ component of the Bloch vector $\vec{r}$ of the qubit state $\rho_{\text{qubit}}$ and where $\theta$ is the angle from the $x$ axis towards the $z$ axis of the Bloch vector $\vec{r}$ in the Bloch sphere.

Let $P_i(0\vert \rho k k_{01}k_0k_+)$ and $P_i(1\vert \rho k k_{01}k_0k_+)=1-P_i(0\vert \rho k k_{01}k_0k_+)$ be the probabilities that the detector D$_i$ does not click and clicks, respectively, given the values of $\rho$, $k$, $k_{01},k_0$ and $k_+$, for $i\in\{0,1,+,-\}$. Let $P(k_{01}k_0k_+\vert \rho,k)$ be the probability of the values $k_{01}$, $k_0$ and $k_+$, given $\rho$ and $k$. We have that
\begin{eqnarray}
\label{detections}
&&P_\text{det}(c_0c_1c_+c_-\vert \rho,k)\nonumber\\
&&\quad=\sum_{k_{01}=0}^k\sum_{k_0=0}^{k_{01}}\sum_{k_+=0}^{k-k_{01}}\bigl[P(k_{01}k_0k_+\vert \rho,k)P_0(c_0\vert \rho k k_{01}k_0k_+) \times\nonumber\\
&&\qquad\qquad\times P_1(c_1\vert \rho k k_{01}k_0k_+)  P_+(c_+\vert \rho k k_{01}k_0k_+) \times\nonumber\\
&&\qquad\qquad\qquad\times P_-(c_-\vert \rho k k_{01}k_0k_+)\bigr],\nonumber\\
\end{eqnarray}
for $c_0,c_1,c_+,c_-\in\{0,1\}$, where
\begin{eqnarray}
\label{detections2}
P(k_{01}k_0k_+\vert \rho, k)&\!=\!&\begin{pmatrix}
k \\ k_{01}
\end{pmatrix}\begin{pmatrix}
k_{01} \\ k_{0}
\end{pmatrix}\begin{pmatrix}
k-k_{01} \\ k_+
\end{pmatrix}\Bigl(\frac{1}{2}\Bigr)^k \times\nonumber\\
&&\!\!\!\!\!\quad\times (q_0)^{k_0}(1-q_0)^{k_{01}-k_0}\!
(q_+)^{k_+}\times\nonumber\\
&&\!\!\!\!\!\quad\quad\times(1-q_+)^{k-k_{01}-k_+}\!,\nonumber\\
P_0(0\vert \rho kk_{01}k_0k_+)&=&(1-d_0)(1-\eta_0)^{k_0},\nonumber\\
P_1(0\vert \rho kk_{01}k_0k_+)&=&(1-d_1)(1-\eta_1)^{k_{01}-k_0},\nonumber\\
P_+(0\vert \rho kk_{01}k_0k_+)&=&(1-d_+)(1-\eta_+)^{k_+},\nonumber\\
P_-(0\vert \rho kk_{01}k_0k_+)&=&(1-d_-)(1-\eta_-)^{k-k_{01}-k_+},\nonumber\\
P_i(1\vert \rho kk_{01}k_0k_+)&=&1-P_i(0\vert \rho k k_{01}k_0k_+),
\end{eqnarray}
for $i\in\{0,1,+,-\}$. Let $P_i$ be the probability that only the detector D$_i$ clicks, for $i\in\{0,1,+,-\}$. We have that
\begin{eqnarray}
\label{detections3}
P_0&=&P_\text{det}(1000\vert\rho,k),\nonumber\\
P_1&=&P_\text{det}(0100\vert\rho,k),\nonumber\\
P_+&=&P_\text{det}(0010\vert\rho,k),\nonumber\\
P_-&=&P_\text{det}(0001\vert\rho,k).
\end{eqnarray}

The probability $P_\text{report}(1,0\vert\rho,k)$ that Bob sends the message $m=1$ to Alice and assigns a measurement outcome in the basis $\mathcal{B}_0$ is the probability that Bob assigns a measurement outcome corresponding to the state $\lvert 0\rangle$ (detector D$_0$) plus the probability that Bob assigns a measurement outcome corresponding to the state $\lvert 1\rangle$ (detector D$_1$). Similarly, the probability $P_\text{report}(1,1\vert\rho,k)$ that Bob sends the message $m=1$ to Alice and assigns a measurement outcome in the basis $\mathcal{B}_1$ is the probability that Bob assigns a measurement outcome corresponding to the state $\lvert \tilde{+}\rangle$ (detector D$_+$) plus the probability that Bob assigns a measurement outcome corresponding to the state $\lvert \tilde{-}\rangle$ (detector D$_-$). Since in the \hyperref[symmetrizationsetup2]{SLII} reporting strategy used by Bob, Bob assigns a measurement outcome to the detector D$_i$ with probability $S_i$ if only the detector D$_i$ clicks then we have that $P_\text{report}(1,0\vert\rho,k)$ and $P_\text{report}(1,1\vert\rho,k)$ are given by (\ref{bases}). Finally, from (\ref{detections}) -- (\ref{detections3}), it follows straightforwardly using the binomial theorem that $P_i$ is given by (\ref{clicks}), for $i\in\{0,1,+,-\}$, as claimed.
\end{proof}

The following lemma states Alice's guessing probability in the \hyperref[MPAII]{MPAII} attack when Bob uses the \hyperref[symmetrizationsetup2]{SLII} reporting strategy, given some specific experimental parameters.

\begin{lemma}
\label{guessingprobability}
Suppose that Alice applies the \hyperref[MPAII]{MPAII} attack with $\rho=\rho_\text{qubit}^{\otimes k}$ and that Bob uses the \hyperref[symmetrizationsetup2]{SLII} reporting strategy with $S_i=\frac{\eta_\text{min}}{\eta_i}$, for some qubit state $\rho_\text{qubit}$, for some nonnegative integer $k$, and for $i\in\{0,1,+,-\}$. Consider the following parameters: $\rho_\text{qubit}=\lvert0 \rangle\langle 0\rvert$, $\theta=0$, $0\leq d_i=d<<1$ and $\eta_i=\eta\in(0,1)$, for $i\in\{0,1,+,-\}$. We assume that Alice knows these values. Alice computes the probability $P_\text{report}(1,\beta\vert \rho,k)$ that Bob sends the message $m=1$ to Alice and assigns a measurement outcome in the basis $\mathcal{B}_\beta$, for $\beta\in\{0,1\}$. Then $S_i=1$, for $i\in\{0,1,+,-\}$ and 
\begin{eqnarray}
\label{basesequation}
P_\text{report}(1,0\vert\rho,k)&=&P_0+P_1,\nonumber\\
P_\text{report}(1,1\vert\rho,k)&=&P_++P_-,
\end{eqnarray}
where
\begin{eqnarray}
\label{probabilitiesequation}
P_0&=&(1-d)^3\biggl[\Bigl(1-\frac{\eta}{2}\Bigr)^k-(1-d)(1-\eta)^k\biggr],\nonumber\\
P_1&=&d(1-d)^3(1-\eta)^k,\nonumber\\
P_+&=&P_-=(1-d)^3\biggl[\Bigl(1-\frac{3\eta}{4}\Bigr)^k-(1-d)(1-\eta)^k\biggr],\nonumber\\
\end{eqnarray}
for any nonnegative integer $k$. It also holds that $P_\text{report}(1,0\vert\rho,k)=P_\text{report}(1,1\vert \rho,k)$ if $k\in\{0,1\}$ and that $P_\text{report}(1,0\vert \rho,k)>P_\text{report}(1,1\vert \rho,k)$ if $k\geq 2$. Alice guesses that Bob assigns a measurement outcome in the basis $\mathcal{B}_0$ if Bob sends her the message $m=1$. Alice's probability to guess the measurement basis assigned by Bob is given by
\begin{equation}
\label{guessingequation}
P_\text{guess}=\frac{P_\text{report}(1,0\vert \rho,k)}{P_\text{report}(1,0\vert \rho,k)+P_\text{report}(1,1\vert \rho,k)},
\end{equation}
which satisfies $P_\text{guess}=\frac{1}{2}$ if $k\in\{0,1\}$, $P_\text{guess}>\frac{1}{2}$ if $k\geq 2$ and $P_\text{guess}\rightarrow 1$ as $k\rightarrow \infty$. 
\end{lemma}
\begin{proof}
Since $\eta_i=\eta\in(0,1)$ and $S_i=\frac{\eta_\text{min}}{\eta_i}$, we see from Bob's \hyperref[symmetrizationsetup2]{SLII} reporting strategy that $S_i=1$, for $i\in\{0,1,+,-\}$. Since in Alice's \hyperref[MPAII]{MPAII} considered attack, $\rho= \rho_\text{qubit}^{\otimes k}$ with $\rho_\text{qubit}=\lvert 0\rangle\langle 0\rvert$, the Bloch vector $\vec{r}$ of the qubit state $\rho_\text{qubit}$ is $\vec{r}=(0,0,1)$. From this, and since $\theta=0$, we have $q_0=1$ and $q_+=\frac{1}{2}$ in Equation (\ref{qs}) of Lemma \ref{lemmasetup2}. By substituting $q_0=1$, $q_+=\frac{1}{2}$, $S_i=1$, $d_i=d$ and $\eta_i=\eta$ in Equations (\ref{bases}) and (\ref{clicks}) of Lemma \ref{lemmasetup2}, it is straightforward to obtain (\ref{basesequation}) and (\ref{probabilitiesequation}). 

It is easy to see from (\ref{basesequation}) and (\ref{probabilitiesequation}) that $P_\text{report}(1,0\vert \rho,k)=P_\text{report}(1,1\vert \rho,k)$ if $k\in\{0,1\}$. To see that $P_\text{report}(1,0\vert \rho,k)>P_\text{report}(1,1\vert \rho,k)$ if $k\geq 2$, we note from (\ref{basesequation}) and (\ref{probabilitiesequation}) that
\begin{equation}
\label{difference}
P_\text{report}(1,0\vert \rho,k)-P_\text{report}(1,1\vert \rho,k)=(1-d)^3f(k),
\end{equation}
where 
\begin{equation}
\label{function}
f(k)=\Bigl(1-\frac{\eta}{2}\Bigr)^k+(1-\eta)^k -2\Bigl(1-\frac{3\eta}{4}\Bigr)^k,
\end{equation}
for any nonnegative integer $k$. The function $g(x)=x^k$ is convex for $x>0$ and $k\geq 2$, since $g''(x)=k(k-1)x^{k-2}>0$ for $x>0$ and $k\geq 2$. Thus, we have
\begin{eqnarray}
\label{function2}
f(k)&=&2\biggl[\frac{1}{2}\Bigl(1-\frac{\eta}{2}\Bigr)^k+\frac{1}{2}(1-\eta)^k -\Bigl(1-\frac{3\eta}{4}\Bigr)^k\biggr]\nonumber\\
&>&2\Biggl[\biggl(\frac{1}{2}\Bigl(1-\frac{\eta}{2}\Bigr)+\frac{1}{2}(1-\eta)\biggr)^k -\Bigl(1-\frac{3\eta}{4}\Bigr)^k\Biggr]\nonumber\\
&=&2\biggl[\Bigl(1-\frac{3\eta}{4}\Bigr)^k-\Bigl(1-\frac{3\eta}{4}\Bigr)^k\biggr]\nonumber\\
&=&0,
\end{eqnarray}
for $k\geq 2$, where in the second line we used that $\eta\in(0,1)$ and that the function $g(x)=x^k$ is convex for $x>0$ and $k\geq 2$. Thus, from (\ref{difference}) -- (\ref{function2}), we see that $P_\text{report}(1,0\vert\rho,k)>P_\text{report}(1,1\vert\rho,k)$ if $k\geq 2$. It follows from 
(\ref{guessingequation}) that $P_\text{guess}>\frac{1}{2}$ if $k\geq 2$, as claimed.

Now we show that $P_\text{guess}\rightarrow 1$ if $k\rightarrow \infty$. It follows from (\ref{basesequation}) -- (\ref{guessingequation}) that 
\begin{eqnarray}
\label{guessingequation2}
&&P_\text{guess}\nonumber\\
&&\quad=1-\frac{P_++P_-}{P_0+P_1+P_++P_-}\nonumber\\
&&\quad =1-\frac{2\Bigl[\bigl(1-\frac{3\eta}{4}\bigr)^k-(1-d)(1-\eta)^k\Bigr]}{\bigl(1-\frac{\eta}{2}\bigr)^k+(1-4(1-d))(1-\eta)^k+2\bigl(1-\frac{3\eta}{4}\bigr)^k}\nonumber\\
&&\quad =1-\frac{2\biggl[1-(1-d)\Bigl(\frac{1-\eta}{1-\frac{3\eta}{4}}\Bigr)^k\biggr]}{2+\Bigl(\frac{1-\frac{\eta}{2}}{1-\frac{3\eta}{4}}\Bigr)^k+(1-4(1-d))\Bigl(\frac{1-\eta}{1-\frac{3\eta}{4}}\Bigr)^k}.\nonumber\\
\end{eqnarray}
We see that $\Bigl(\frac{1-\frac{\eta}{2}}{1-\frac{3\eta}{4}}\Bigr)^k\rightarrow \infty$ as $k\rightarrow \infty$ and $\Bigl(\frac{1-\eta}{1-\frac{3\eta}{4}}\Bigr)^k\rightarrow 0$ as $k\rightarrow \infty$. Thus, we see from (\ref{guessingequation2}) that $P_\text{guess}\rightarrow 1$ as $k\rightarrow \infty$, as claimed.
\end{proof}

It is important to interpret correctly the value of $P_\text{guess}$ in Lemma \ref{guessingprobability}. As explained above this probability is conditioned on Bob setting $m=1$, i.e. on Bob assigning a valid measurement outcome in any of the two bases. Thus, although $P_\text{guess}$ can be very close to unity for $k$ large enough, this does not mean that Alice's dishonest pulse allows her to guess Bob's assigned measurement basis with great probability. In fact, it is easy to see from Equations (\ref{basesequation}) and (\ref{probabilitiesequation}) that $P_\text{report}(1,0\vert \rho,k)+P_\text{report}(1,1\vert \rho,k)\rightarrow 0$ as $k\rightarrow \infty$. This means that the probability that Bob assigns $m=1$ is very small if $k$ is large. However, Alice can implement the generalization of the \hyperref[MPAII]{MPAII} attack described above in which Alice sends Bob various dishonest pulses. In the extreme case that each of the $N$ pulses that Alice sends Bob is a dishonest pulse with large of number of photons $k$, for a small fraction of the pulses Bob will send Alice the message $m=1$. But for each of these pulses, the probability that Alice guesses the measurement basis assigned by Bob is very close to unity.

A possible countermeasure by Bob against the \hyperref[MPAII]{MPAII} attack is that Bob aborts if a fraction $f<\delta_\text{report}$ of the pulses sent by Alice produces the value $m=1$, for some predetermined $\delta_\text{report}\in(0,1)$. Investigating this countermeasure in detail is left as an open problem. However, neither this nor other countermeasures were proposed in Ref. \cite{NJMKW12}. We must then conclude that the protocol of Ref. \cite{NJMKW12} is insecure.

In Fig. \ref{plots}, we plot Alice's guessing probability $P_\text{guess}$ in the \hyperref[MPAII]{MPAII} attack when Bob uses the \hyperref[symmetrizationsetup2]{SLII} reporting strategy, with the parameters of Lemma \ref{guessingprobability}, as a function of the number of photons $k$ of Alice's dishonest pulse. We consider the case that Bob's dark count probabilities are $d=10^{-5}$ and consider various values for the efficiencies $\eta\in(0,1)$ of Bob's detectors.

\begin{figure}
\includegraphics[scale=0.57]{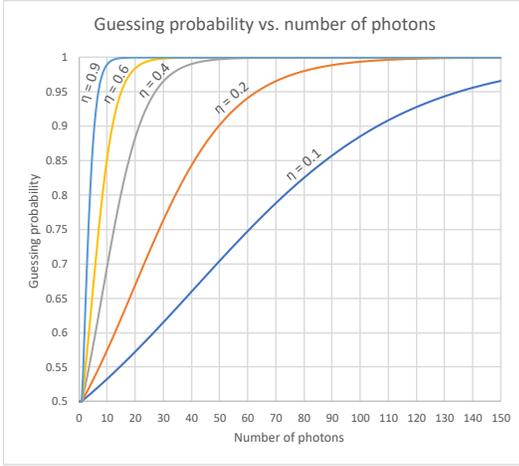}
\caption{\label{plots} \textbf{Alice's guessing probability in multiphoton attack I.} We plot Alice's guessing probability $P_\text{guess}$ in the \hyperref[MPAII]{MPAII} attack when Bob uses the \hyperref[symmetrizationsetup2]{SLII} reporting strategy, with the parameters of Lemma \ref{guessingprobability}, as a function of the number of photons $k$, for $k\in\{0,1,\ldots,150\}$. We consider the case $d=10^{-5}$ and consider five different values for the detection efficiencies $\eta$.}
\end{figure}

\subsubsection{Proof of Lemma \ref{lemmaSLII}}
\begin{proof}
We can suppose that Alice implements the \hyperref[MPAII]{MPAII} attack in the particular case $k\in\{0,1\}$, that Alice sends Bob an empty pulse if $k=0$, and that $\rho=\rho_\text{qubit}$ is a qubit state with Bloch vector $\vec{r}=(r_x,r_y,r_z)$ if $k=1$. Thus, the probability  $P_{\text{report}}(1,\beta\lvert  \rho, k)$ that Bob sets $m=1$ and assigns a valid measurement outcome in the basis $\mathcal{B}_\beta$ is given by Eq. (\ref{bases}) of Lemma \ref{lemmasetup2}, for $\beta,k\in\{0,1\}$.

We consider the case $k=0$. From Eqs. (\ref{bases}) -- (\ref{qs}) of Lemma \ref{lemmasetup2}, we have
\begin{eqnarray}
\label{repprobsdiff}
&&\bigl\lvert P_{\text{report}}(1,0\lvert  \rho, k)- P_{\text{report}}(1,1\lvert  \rho, k)\bigr\rvert\nonumber\\
&&\quad=\eta_\text{min}\biggl\lvert\frac{d_0(1-d_1)(1-d_+)(1-d_-)}{\eta_0}\nonumber\\
&&\qquad\qquad\quad+\frac{d_1(1-d_0)(1-d_+)(1-d_-)}{\eta_1}\nonumber\\
&&\qquad\qquad\quad-\frac{d_+(1-d_0)(1-d_1)(1-d_-)}{\eta_+}\nonumber\\
&&\qquad\qquad\quad-\frac{d_-(1-d_0)(1-d_1)(1-d_+)}{\eta_-}
\biggr\rvert\nonumber\\
&&\quad \leq 2\delta\nonumber\\
&&\quad \leq 6\delta ,
\end{eqnarray}
where in the second line we used that $0<\eta_\text{min}\leq \eta_i$ and $0\leq d_i\leq \delta<1$, for $i\in\{0,1,+,-\}$.

We consider the case $k=1$. From Eqs. (\ref{bases}) -- (\ref{qs}) of Lemma \ref{lemmasetup2}, we have
\begin{eqnarray}
\label{repprobsdiff2}
&&\bigl\lvert P_{\text{report}}(1,0\lvert  \rho, k)- P_{\text{report}}(1,1\lvert \rho, k)\bigr\rvert\nonumber\\
&&\quad=\eta_\text{min}\Biggl\lvert\frac{(1-d_1)(1-d_+)(1-d_-)q_0\eta_0}{\eta_0}\nonumber\\
&&\qquad\qquad\quad+\frac{(1-d_0)(1-d_+)(1-d_-)(1-q_0)\eta_1}{\eta_1}\nonumber\\
&&\qquad\qquad\quad-\frac{(1-d_0)(1-d_1)(1-d_-)q_+\eta_+}{\eta_+}\nonumber\\
&&\qquad\qquad\quad-\frac{(1-d_0)(1-d_1)(1-d_+)(1-q_+)\eta_-}{\eta_-}
\nonumber\\
&&\qquad\qquad\quad+\biggl[\frac{d_0(1-d_1)(1-d_+)(1-d_-)}{\eta_0}\nonumber\\
&&\qquad\qquad\qquad\quad+\frac{d_1(1-d_0)(1-d_+)(1-d_-)}{\eta_1}\nonumber\\
&&\qquad\qquad\qquad\quad-\frac{d_+(1-d_0)(1-d_1)(1-d_-)}{\eta_+}\nonumber\\
&&\qquad\qquad\qquad\quad-\frac{d_-(1-d_0)(1-d_1)(1-d_+)}{\eta_-}\biggr]\times\nonumber\\
&&\qquad\qquad\qquad\qquad\times \biggl[1-\frac{1}{2}\Bigl(q_0\eta_0+(1-q_0)\eta_1\nonumber\\
&&\qquad\qquad\qquad\qquad\qquad\quad+q_+\eta_++(1-q_+)\eta_-\Bigr)\biggr]\Biggr\rvert\nonumber\\
&&\quad =\eta_\text{min}\Biggl\lvert(1-d_+)(1-d_-)\bigl(1-d_0(1-q_0)-d_1q_0\bigr)\nonumber\\
&&\qquad\qquad\quad-(1-d_0)(1-d_1)\bigl(1-d_+(1-q_+)-d_-q_+\bigr)\nonumber\\
&&\qquad\qquad +\biggl[\frac{d_0(1-d_1)(1-d_+)(1-d_-)}{\eta_0}\nonumber\\
&&\qquad\qquad\qquad +\frac{d_1(1-d_0)(1-d_+)(1-d_-)}{\eta_1}\nonumber\\
&&\qquad\qquad\qquad -\frac{d_+(1-d_0)(1-d_1)(1-d_-)}{\eta_+}\nonumber\\
&&\qquad\qquad\qquad -\frac{d_-(1-d_0)(1-d_1)(1-d_+)}{\eta_-}\biggr]\times\nonumber\\
&&\qquad\qquad\qquad\quad\times\biggl[1-\frac{1}{2}\Bigl(q_0\eta_0+(1-q_0)\eta_1\nonumber\\
&&\qquad\qquad\qquad\qquad\qquad\qquad\quad+q_+\eta_++(1-q_+)\eta_-\Bigr)\biggr]\Biggr\rvert\nonumber\\
&&\quad \leq 1-(1-\delta)^2(1-2\delta)+2\delta\nonumber\\
&&\quad= 6\delta -5\delta^2+2\delta^3\nonumber\\
&&\quad \leq 6\delta,
\end{eqnarray}
where in the third line we used that $0\leq d_i\leq \delta<1$, $0\leq q_0\leq 1$, $0\leq q_+\leq 1$ and $0<\eta_\text{min}\leq \eta_i<1$, for $i\in\{0,1,+,-\}$; and in the last line we used that $\delta^3\leq \delta^2$ since $0\leq \delta<1$.

The claimed result (\ref{a3new}) follows from (\ref{repprobsdiff}) and (\ref{repprobsdiff2}).
\end{proof}

\subsubsection{Reporting single and double clicks}

As in setup I discussed in the main text, a better reporting strategy in setup II is that Bob sets $m=1$ if one or two detectors click. More precisely, we consider the following reporting strategy.

\begin{definition}[Reporting single and double clicks with setup II (RSDCII)]
\label{RLODCII}
Bob sends Alice the message $m=1$ if at least one detector from the pair D$_0$, D$_1$ (D$_+$, D$_-$) clicks and none detector from the pair D$_+$, D$_-$ (D$_0$, D$_1$) click. If only the detector D$_i$ clicks, Bob sets $m=1$ and assigns the measurement outcome corresponding to the detector D$_i$, for $i\in\{0,1,+,-\}$. If D$_0$ and D$_{1}$ click, but D$_+$ and D$_-$ do not click, Bob sets $m=1$ and randomly assigns a measurement outcome corresponding to the detector D$_0$ and D$_{1}$. Similarly, if D$_+$ and D$_{-}$ click, but D$_0$ and D$_1$ do not click, Bob sets $m=1$ and randomly assigns a measurement outcome corresponding to the detector D$_+$ and D$_{-}$. In all other cases Bob sets $m=0$.
\end{definition}

We see that with this reporting strategy, Bob assigns a measurement outcome in the basis $\mathcal{B}_0$ with unit probability if at least one of the detectors D$_0$ and D$_1$ click and none of the detectors D$_+$ and D$_-$ click. Similarly, Bob assigns a measurement outcome in the basis $\mathcal{B}_1$ with unit probability if at least one of the detectors D$_+$ and D$_-$ click and none of the detectors D$_0$ and D$_1$ click.

As the following lemma shows, if Bob applies the \hyperref[RLODCII]{RSDCII} reporting strategy, with Bob's detectors having exactly
equal efficiencies and their dark count probabilities being independent
of his measurement basis, Alice cannot learn any information about
$\beta$ from the message $m$.

\begin{lemma}
\label{setup2equaldet}
Consider setup II of Fig. \ref{setupII} in which Bob uses the \hyperref[RLODCII]{RSDCII} reporting strategy. Suppose that Alice sends Bob a dishonest pulse of $k$ photons, for some nonnegative integer $k$ chosen by Alice. Let $\rho$ be an arbitrary quantum entangled state of the $k$ qubits encoded in the polarization of the $k$ photons, which is in an arbitrary entangled state with an ancilla held by Alice, and which is chosen by Alice. Let $P_\text{report}(1,\beta\vert \rho,k)$ be the probability that Bob sends the message $m=1$ to Alice and assigns a valid measurement outcome in the basis $\mathcal{B}_\beta$, for $\beta\in\{0,1\}$. Suppose that $\eta_0=\eta_1\in(0,1)$ and $\eta_+=\eta_-\in(0,1)$, for $i\in\{0,1,+,-\}$. 
Then 
\begin{eqnarray}
\label{report}
&&P_\text{report}(1,0\vert\rho,k)=(1-d_+)(1-d_-)\Bigl(1-\frac{\eta_+}{2}\Bigr)^k \nonumber\\
&&\qquad-(1-d_0)(1-d_1)(1-d_+)(1-d_-)\Bigl(1-\frac{\eta_0+\eta_+}{2}\Bigr)^k,\nonumber\\
&&P_\text{report}(1,1\vert\rho,k)=(1-d_0)(1-d_1)\Bigl(1-\frac{\eta_0}{2}\Bigr)^k \nonumber\\
&&\qquad-(1-d_0)(1-d_1)(1-d_+)(1-d_-)\Bigl(1-\frac{\eta_0+\eta_+}{2}\Bigr)^k,\nonumber\\
\end{eqnarray}
for any quantum state $\rho$ of $k$ qubits, for any nonnegative integer $k$, and for any qubit orthogonal bases $\mathcal{B}_0$ and $\mathcal{B}_1$.
\end{lemma}

\begin{corollary}
\label{corollary1}
In Lemma \ref{setup2equaldet}, if $(1-d_0)(1-d_1)=(1-d_+)(1-d_-)$ and $\eta_i=\eta$, for $i\in\{0,1,+,-\}$, then $P_\text{report}(1,0\vert \rho,k)=P_\text{report}(1,1\vert\rho,k)$, for any quantum state $\rho$ of $k$ qubits and for any nonnegative integer $k$ chosen by Alice, and for any qubit orthogonal bases $\mathcal{B}_0$ and $\mathcal{B}_1$. Thus, Alice cannot obtain any information about Bob's assigned measurement basis in this case.
\end{corollary}

Exactly equal detection efficiencies cannot be guaranteed in practice. However, as discussed in the main text, attenuators can be used to make the detector efficiencies approximately equal. Furthermore, Bob can effectively make his detectors to have
approximately equal dark count probabilities  by simulating dark counts in the detectors with lower dark count probabilities so that they approximate the dark count probability of the detector with highest dark count probability.

\begin{proof}[Proof of Lemma \ref{setup2equaldet}]
  The probability that Bob sets $m=1$ and assigns a measurement
  outcome in the basis $\mathcal{B}_0$ is the probability that he
  assigns a valid measurement outcome to the state corresponding to
  the detector D$_0$ plus the probability that he assigns a valid
  measurement outcome to the state corresponding to the detector
  D$_1$. Similarly, the probability that Bob sets $m=1$ and assigns a
  measurement outcome in the basis $\mathcal{B}_1$ is the probability
  that he assigns a valid measurement outcome to the state
  corresponding to the detector D$_+$ plus the probability that he
  assigns a valid measurement outcome to the state corresponding to
  the detector D$_-$. Thus, it follows from Bob's
  \hyperref[RLODCII]{RSDCII} reporting strategy that
  \begin{eqnarray}
\label{reportprobs}
&&P_\text{report}(1,0\vert\rho,k)\nonumber\\
&&\quad=P_{\text{det}}(1000\vert\rho,k)+P_{\text{det}}(0100\vert\rho,k)+P_{\text{det}}(1100\vert\rho,k),\nonumber\\
&&P_\text{report}(1,1\vert \rho,k)\nonumber\\
&&\quad=P_{\text{det}}(0010\vert\rho,k)+P_{\text{det}}(0001\vert\rho,k)+P_{\text{det}}(0011\vert\rho,k).\nonumber\\
\end{eqnarray}

Consider Bob's setup in Fig. \ref{setupII}. Alice sends Bob a pulse of
$k$ photons, encoding a $k-$qubit state $\rho$ in the polarization
degrees of freedom, for some nonnegative integer $k$. Let $k_{01}$ be
the number of photons that are transmitted through the 50:50 beam
splitter towards the polarizing beam splitter PBS$_{01}$ and let
$k_{+-}=k-k_{01}$ be the number of photons that are reflected from the
50:50 beam splitter towards the polarizing beam splitter PBS$_{+-}$,
for $k_{01}\in\{0,1,\ldots,k\}$. Let $k_0$ and $k_1=k_{01}-k_0$ be the
number of photons that go towards the detectors D$_0$ and D$_1$,
respectively, for $k_0\in\{0,1,\ldots,k_{01}\}$. Let $k_+$ and
$k_-=k-k_{01}-k_+$ be the number of photons that go towards the
detectors D$_+$ and D$_-$, respectively, for
$k_+\in\{0,1,\ldots,k-k_{01}\}$.

The probabilities that a photon is transmitted through the 50:50 beam
splitter towards the polarizing beam splitter PBS$_{01}$ and reflected
from the 50:50 beam splitter towards the polarizing beam splitter
PBS$_{+-}$ are both $\frac{1}{2}$. We note that we do not lose
generality by considering that the 50:50 beam splitter has
transmission and reflection probabilities exactly equal to
$\frac{1}{2}$. If these probabilities were different, these values
could be absorbed in the efficiencies of the detectors, leaving the
equivalent transmission and reflection probabilities of the 50:50 beam
splitter effectively equal to $\frac{1}{2}$.

Let $P(k_{01}k_0k_+\vert \rho,k)$ be the probability of the values $k_{01}$, $k_0$ and $k_+$, given $\rho$ and $k$. We have that
\begin{equation}
\label{Born}
P(k_{01}k_0k_+\vert \rho,k)=\Bigl(\frac{1}{2}\Bigr)^k \begin{pmatrix}
k \\ k_{01}
\end{pmatrix}P(k_{0}k_+\vert k_{01},\rho,k),
\end{equation}
where $P(k_{0}k_+\vert k_{01},\rho,k)$ is given by the Born rule and satisfies
\begin{equation}
\label{sumprobs}
\sum_{k_0=0}^{k_{01}}\sum_{k_+=0}^{k-k_{01}}P(k_{0}k_+\vert k_{01},\rho,k)=1,
\end{equation}
for any $k_{01}\in\{0,1,\ldots,k\}$, for any $k-$qubit state $\rho$ and for any nonnegative integer $k$.

Let $P_{\text{det}}(c_0c_1c_+c_-\vert k_{01},k_0,k_+,k)$ be the probability of the detection event $(c_0,c_1,c_+,c_-)$, given the values $k_{01}$, $k_0$ and $k_+$, which is independent of $\rho$, for $c_0,c_1,c_+,c_-\in\{0,1\}$. We have that
\begin{eqnarray}
\label{probs}
&&P_{\text{det}}(c_0c_1c_+c_-\vert \rho,k)\nonumber\\
&&\qquad=\sum_{k_{01}=0}^k\sum_{k_0=0}^{k_{01}}\sum_{k_+=0}^{k-k_{01}}\bigl[P_{\text{det}}(c_0c_1c_+c_-\vert k_{01},k_0,k_+,k)\times\nonumber\\&&\qquad\qquad\qquad\qquad\qquad\quad \times P(k_{01}k_0k_+\vert \rho,k)\bigr],
\end{eqnarray}
and that
\begin{eqnarray}
\label{probs2}
&&P_{\text{det}}(c_0c_1c_+c_-\vert k_{01},k_0,k_+,k)\nonumber\\
&&\quad=P_0(c_0\vert k_0)P_1(c_1\vert k_{01},k_0) P_+(c_+\vert k_+)P_-(c_-\vert k_{01},k_+,k),\nonumber\\
\end{eqnarray}
where $P_i(c_i\vert k_i)$ is the probability that the detector D$_i$ clicks if $c_i=1$ and does not click if $c_i=0$, given the value of $k_i$, which satisfies
\begin{eqnarray}
\label{probs3}
P_i(0\vert k_i)&=&(1-d_i)(1-\eta_i)^{k_i},\nonumber\\
P_i(1\vert k_i)&=&1-P_i(0\vert k_i),
\end{eqnarray}
for $i\in\{0,1,+,-\}$, $c_0,c_1,c_+,c_-\in\{0,1\}$, any $k-$qubit state $\rho$ and any nonnegative integer $k$; where we recall that $k_1=k_{01}-k_0$ and $k_-=k-k_{01}-k_+$. 

From (\ref{reportprobs}), (\ref{Born}), (\ref{probs}),  (\ref{probs2}) and (\ref{probs3}), we have
\begin{eqnarray}
\label{difference4}
&&P_\text{report}(1,0\vert\rho,k)\nonumber\\&&\quad=\sum_{k_{01}=0}^k\sum_{k_0=0}^{k_{01}}\sum_{k_+=0}^{k-k_{01}} \Bigl(\frac{1}{2}\Bigr)^k \begin{pmatrix}
k \\ k_{01}
\end{pmatrix}P(k_{0}k_+\vert k_{01},\rho,k) \times\nonumber\\
&&\qquad\times\Bigl[\bigl(1-P_0(0\vert k_0)P_1(0\vert k_{01},k_0)\bigr) \times\nonumber\\
&&\qquad\qquad\qquad\times P_+(0\vert k_+)P_-(0\vert k_{01},k_+,k)\Bigr]\nonumber
\\
&&\quad=\sum_{k_{01}=0}^k\sum_{k_0=0}^{k_{01}}\sum_{k_+=0}^{k-k_{01}} \Bigl(\frac{1}{2}\Bigr)^k \begin{pmatrix}
k \\ k_{01}
\end{pmatrix}P(k_{0}k_+\vert k_{01},\rho,k) \times\nonumber\\
&&\quad\qquad\times\bigl[(1-d_+)(1-d_-)(1-\eta_+)^{k-k_{01}}\nonumber\\
&&\qquad\qquad\quad-(1-d_0)(1-d_1)(1-d_+)(1-d_-)\times\nonumber\\
&&\quad\qquad\qquad\qquad\times(1-\eta_0)^{k_{01}}(1-\eta_+)^{k-k_{01}}\bigr]\nonumber\\
&&\quad=\Bigl(\frac{1}{2}\Bigr)^k\sum_{k_{01}=0}^k\begin{pmatrix}
k \\ k_{01}
\end{pmatrix}   \bigl[(1-d_+)(1-d_-)(1-\eta_+)^{k-k_{01}}\nonumber\\
&&\qquad\qquad\quad-(1-d_0)(1-d_1)(1-d_+)(1-d_-)\times\nonumber\\
&&\quad\qquad\qquad\qquad\times(1-\eta_0)^{k_{01}}(1-\eta_+)^{k-k_{01}}\bigr]\nonumber\\
&&\quad=   (1-d_+)(1-d_-)\Bigl(1-\frac{\eta_+}{2}\Bigr)^{k}\nonumber\\
&&\quad\quad-(1-d_0)(1-d_1)(1-d_+)(1-d_-)\Bigl(1-\frac{\eta_0+\eta_+}{2}\Bigr)^{k},\nonumber\\
\end{eqnarray}
where in the second line we used (\ref{probs3}), $\eta_+=\eta_-$ and $\eta_0=\eta_1$, in the third line we used (\ref{sumprobs}), and in the last line we used the Binomial theorem. Similarly, from (\ref{reportprobs}), (\ref{Born}), (\ref{probs}),  (\ref{probs2}) and (\ref{probs3}), we have
\begin{eqnarray}
\label{difference5}
&&P_\text{report}(1,1\vert\rho,k)\nonumber\\&&\quad=\sum_{k_{01}=0}^k\sum_{k_0=0}^{k_{01}}\sum_{k_+=0}^{k-k_{01}} \Bigl(\frac{1}{2}\Bigr)^k \begin{pmatrix}
k \\ k_{01}
\end{pmatrix}P(k_{0}k_+\vert k_{01},\rho,k) \times\nonumber\\
&&\qquad\times\Bigl[\bigl(1-P_+(0\vert k_+)P_-(0\vert k_{01},k_+,k)\bigr) \times\nonumber\\
&&\qquad\qquad\qquad\times P_0(0\vert k_0)P_1(0\vert k_{01},k_0)\Bigr]\nonumber
\\
&&\quad=\sum_{k_{01}=0}^k\sum_{k_0=0}^{k_{01}}\sum_{k_+=0}^{k-k_{01}} \Bigl(\frac{1}{2}\Bigr)^k \begin{pmatrix}
k \\ k_{01}
\end{pmatrix}P(k_{0}k_+\vert k_{01},\rho,k) \times\nonumber\\
&&\quad\qquad\times\bigl[(1-d_0)(1-d_1)(1-\eta_0)^{k_{01}}\nonumber\\
&&\qquad\qquad\quad-(1-d_0)(1-d_1)(1-d_+)(1-d_-)\times\nonumber\\
&&\quad\qquad\qquad\qquad\times(1-\eta_0)^{k_{01}}(1-\eta_+)^{k-k_{01}}\bigr]\nonumber\\
&&\quad=\Bigl(\frac{1}{2}\Bigr)^k\sum_{k_{01}=0}^k\begin{pmatrix}
k \\ k_{01}
\end{pmatrix}   \bigl[(1-d_0)(1-d_1)(1-\eta_0)^{k_{01}}\nonumber\\
&&\qquad\qquad\quad-(1-d_0)(1-d_1)(1-d_+)(1-d_-)\times\nonumber\\
&&\quad\qquad\qquad\qquad\times(1-\eta_0)^{k_{01}}(1-\eta_+)^{k-k_{01}}\bigr]\nonumber\\
&&\quad=   (1-d_0)(1-d_1)\Bigl(1-\frac{\eta_0}{2}\Bigr)^{k}\nonumber\\
&&\quad\quad-(1-d_0)(1-d_1)(1-d_+)(1-d_-)\Bigl(1-\frac{\eta_0+\eta_+}{2}\Bigr)^{k},\nonumber\\
\end{eqnarray}
where in the second line we used (\ref{probs3}), $\eta_+=\eta_-$ and $\eta_0=\eta_1$, in the third line we used (\ref{sumprobs}), and in the last line we used the Binomial theorem. Thus, the claimed result (\ref{report}) follows from (\ref{difference4}) and (\ref{difference5}).

\end{proof}

\subsubsection{multiphoton attack II}
\label{lastsec}
We recall from Lemma \ref{lemma0} in the main text that, in setup I, if Bob's detectors have exactly equal efficiencies and Bob reports double clicks with unit probability, as in reporting strategy II, then Alice cannot obtain any information about Bob's measurement basis. This motivated our definition of multiphoton attack II for setup I, given in the main text, as any strategy implemented by Alice allowing her to exploit the difference of Bob's detection efficiencies to obtain information about Bob's measurement basis, when Bob reports double clicks with unit probability, as in reporting strategy II, or with high probability. 

Similarly, in setup II, Corollary \ref{corollary1} says that if Bob applies the \hyperref[RLODCII]{RSDCII} reporting strategy, and if $(1-d_0)(1-d_1)=(1-d_+)(1-d_-)$ and $\eta_i=\eta$ for $i\in\{0,1,+,-\}$, then Alice cannot learn any information about Bob's assigned measurement basis, for a pulse that Bob reports as being successfully measured. This motivates us to define \emph{multiphoton attack II in setup II} as any strategy implemented by Alice allowing her to exploit the difference of Bob's detection efficiencies to obtain information about Bob's assigned measurement basis, for a pulse that Bob reports as successfully measured, when Bob applies the \hyperref[RLODCII]{RSDCII} reporting strategy.

%Providing detailed analyses for the multiphoton attack II in setup II is considerably more challenging than for the multiphoton attack II in setup I, because Bob has four detectors in setup II, while he only has two detectors in setup I, which makes the analysis of the detection probabilities for setup II more complicated than for setup I. However, 
It is not difficult to see that the multiphoton attack II in setup II allows Alice to obtain some information about Bob's assigned measurement basis, for any photon pulse with $k\geq 1$ photons. We show this below for the particular case in which $(1-d_0)(1-d_1)=(1-d_+)(1-d_-)$, $\eta_0=\eta_1$ and $\eta_+=\eta_-$, with $\eta_0\neq\eta_+$. In this case, we show in particular that Alice guesses $\beta$ with probability tending to unity if $k\rightarrow \infty$, conditioned on Bob reporting $m=1$. In Fig. \ref{lastfig}, we present plots for Alice's guessing probability $P_\text{guess}^{\text{II}}$ versus the number of photons $k$ of her dishonest pulse in the multiphoton attack II in setup II of Lemma \ref{lastlemma} below.

\begin{figure}
\includegraphics[scale=0.65]{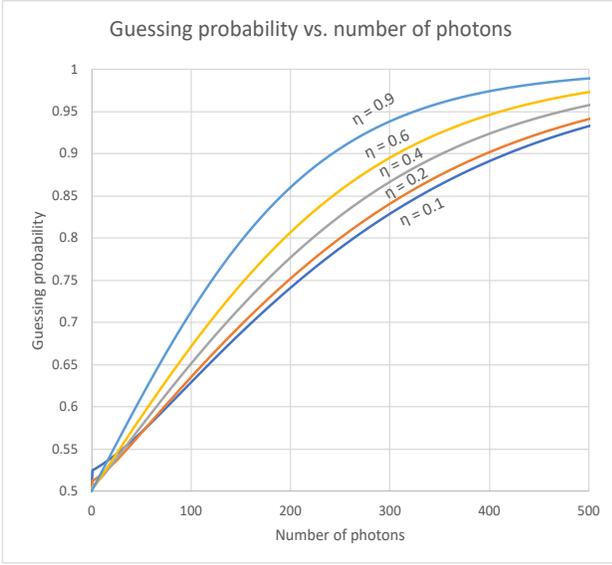}
\caption{\label{lastfig} \textbf{Alice's guessing probability in multiphoton attack II.} We plot Alice's guessing probability $P_\text{guess}^{\text{II}}$ in the multiphoton attack II in setup II of Lemma \ref{lastlemma} when Bob uses the \hyperref[RLODCII]{RSDCII} reporting strategy, as a function of the number of photons $k$, for $k\in\{0,1,\ldots,500\}$. We consider the case $d_0=d_1=d_+=d_-=10^{-5}$, $\eta_0=\eta_1=\eta+0.005$ and $\eta_+=\eta_-=\eta-0.005$, for five different values of $\eta$.}
\end{figure}

\begin{lemma}
\label{lastlemma}
Consider setup II of Fig. \ref{setupII} in which Bob uses the \hyperref[RLODCII]{RSDCII} reporting strategy. Suppose that $(1-d_0)(1-d_1)=(1-d_+)(1-d_-)\equiv a\in(0,1]$, $\eta_0=\eta_1\in(0,1)$ and $\eta_+=\eta_-\in(0,1)$, with $\eta_0\neq\eta_+$. Suppose that Alice applies the following multiphoton attack II. Alice prepares a dishonest pulse of $k$ photons encoding a $k$-qubit state $\rho$ in the polarization degrees of freedom, for some nonnegative integer $k$ chosen by Alice. Let $P_\text{report}(1,\beta\vert \rho, k)$ be the probability that Bob reports the message $m=1$ to Alice and that Bob assigns a valid measurement outcome in the basis $\mathcal{B}_\beta$, for $\beta\in\{0,1\}$. We assume that Alice knows the value of $P_\text{report}(1,\beta\vert \rho, k)$, for $\beta\in\{0,1\}$. If $P_\text{report}(1,0\vert \rho, k)\geq P_\text{report}(1,1\vert \rho, k)$, and if Bob reports to Alice the message $m=1$, Alice guesses that Bob's assigned measurement basis is $\mathcal{B}_0$. On the other hand, if $P_\text{report}(1,1\vert \rho, k) > P_\text{report}(1,0\vert \rho, k)$, and if Bob reports to Alice the message $m=1$, Alice guesses that Bob's assigned measurement basis is $\mathcal{B}_1$. Thus, the probability that Alice guesses Bob's assigned measurement basis $\mathcal{B}_\beta$ for the dishonest pulse, conditioned on Bob assigning a measurement basis (i.e. conditioned on $m=1$), is given by
\begin{equation}
\label{lasteq1}
P_\text{guess}^{\text{II}}=\frac{\max_{\beta\in\{0,1\}}\{P_\text{report}(1,\beta\vert \rho, k)\}}{P_\text{report}(1,0\vert \rho, k)+P_\text{report}(1,1\vert \rho, k)}.
\end{equation}
Then, for any $k-$qubit state $\rho$, we have $P_\text{guess}^{\text{II}}=\frac{1}{2}$ if $k=0$, $P_\text{guess}^{\text{II}}>\frac{1}{2}$ if $k\geq 1$, and $\lim_{k\rightarrow\infty}P_\text{guess}^{\text{II}}=1$.
\end{lemma}

\begin{proof}
From Lemma \ref{setup2equaldet}, we have
\begin{eqnarray}
\label{lasteq2}
P_\text{report}(1,0\vert \rho, k)&=&a\Bigl(1-\frac{\eta_+}{2}\Bigr)^k-a^2\Bigl(1-\frac{\eta_0+\eta_1}{2}\Bigr)^k,\nonumber\\
P_\text{report}(1,1\vert \rho, k)&=&a\Bigl(1-\frac{\eta_0}{2}\Bigr)^k-a^2\Bigl(1-\frac{\eta_0+\eta_1}{2}\Bigr)^k,\nonumber\\
\end{eqnarray}
for any $k-$qubit state $\rho$, for any nonnegative integer $k$, and for any qubit orthogonal bases $\mathcal{B}_0$ and $\mathcal{B}_1$.

From (\ref{lasteq2}), we have $P_\text{report}(1,0\vert \rho, 0)=P_\text{report}(1,1\vert \rho, 0)$. Thus, $P_\text{guess}^{\text{II}}=\frac{1}{2}$ if $k=0$, as claimed. We also have $P_\text{report}(1,0\vert \rho, 1)>P_\text{report}(1,1\vert \rho, 1)$ if $\eta_0>\eta_+$, and $P_\text{report}(1,1\vert \rho, 1)>P_\text{report}(1,0\vert \rho, 1)$ if $\eta_+>\eta_0$, for $k\geq 1$. Thus, from (\ref{lasteq1}), we have $P_\text{guess}^{\text{II}}>\frac{1}{2}$, for $k\geq 1$, as claimed.

We show that
\begin{equation}
\label{lasteq3}
\lim_{k\rightarrow \infty} P_\text{guess}^{\text{II}}=1.
\end{equation}
Let $\eta_0>\eta_+$. From (\ref{lasteq1}) and (\ref{lasteq2}), we have
\begin{eqnarray}
\label{lasteq4}
&&2P_\text{guess}^{\text{II}}-1\nonumber\\
&&\qquad=\frac{2P_\text{report}(1,0\vert \rho, k)}{P_\text{report}(1,0\vert \rho, k)+P_\text{report}(1,1\vert \rho, k)}-1\nonumber\\
&&\qquad=\frac{P_\text{report}(1,0\vert \rho, k)-P_\text{report}(1,1\vert \rho, k)}{P_\text{report}(1,0\vert \rho, k)+P_\text{report}(1,1\vert \rho, k)}\nonumber\\
&&\qquad= \frac{\Bigl(1-\frac{\eta_+}{2}\Bigr)^k-\Bigl(1-\frac{\eta_0}{2}\Bigr)^k}{\Bigl(1-\frac{\eta_+}{2}\Bigr)^k+\Bigl(1-\frac{\eta_0}{2}\Bigr)^k-2a\Bigl(1-\frac{\eta_0+\eta_+}{2}\Bigr)^k}\nonumber\\
&&\qquad= \frac{1-\biggl(\frac{1-\frac{\eta_0}{2}}{1-\frac{\eta_+}{2}}\biggr)^k}
{1+\biggl(\frac{1-\frac{\eta_0}{2}}{1-\frac{\eta_+}{2}}\biggr)^k-2a\biggl(\frac{1-\frac{\eta_0+\eta_+}{2}}{1-\frac{\eta_+}{2}}\biggr)^k}.
\end{eqnarray}
From (\ref{lasteq4}), we obtain (\ref{lasteq3}), as claimed. It is straightforward to see that (\ref{lasteq3}) also holds if $\eta_+>\eta_0$, by interchanging $\eta_0$ and $\eta_+$ in (\ref{lasteq4}).
\end{proof}

As we discussed previously for multiphoton attack I in setup II, it is important to interpret correctly the value of $P_\text{guess}^{\text{II}}$ in Lemma \ref{lastlemma}. This probability is conditioned on Bob setting $m=1$, i.e. on Bob assigning a valid measurement outcome in any of the two bases. Therefore, although $P_\text{guess}^{\text{II}}$ can be very close to unity for $k$ large enough, the probability that Bob assigns $m=1$ decreases with increasing $k$, as from (\ref{lasteq2}) we have $P_\text{report}(1,0\vert \rho, k)+P_\text{report}(1,1\vert \rho, k)\rightarrow 0$ as $k\rightarrow\infty$. However, Alice can extend her attack by sending a large number $N$ of dishonest pulses with large numbers of photons $k$. In this case, Bob assigns $m=1$ for a small fraction of the pulses. But for each of these pulses, Alice guesses Bob's assigned measurement basis with probability close to unity.

A possible countermeasure by Bob against this attack is that Bob aborts if a fraction $f<\delta_\text{report}$ of the pulses sent by Alice produces the value $m=1$, for some predetermined $\delta_\text{report}\in(0,1)$.
More broadly, Bob may use the statistics of his detection events across his four detectors to infer probabilistically if the number of multiphoton pulses is above a threshold, and abort in this case. In this way, dishonest Alice must limit her number of multiphoton pulses to avoid Bob aborting. Although this countermeasure can be helpful, it cannot guarantee perfect protection against arbitrary multiphoton attacks. As discussed in the main text, any photon source emits multiphoton pulses with nonzero probability. Thus, dishonest Alice can send Bob multiphoton pulses with the statistics of the source agreed for the protocol. In this way, Bob cannot detect that Alice is cheating, but Alice can still  learn Bob's measurement basis with nonzero probability. Investigating whether this countermeasure allows Bob in practice to bound Alice's cheating probability below a sufficiently small value is left as an open problem.

\subsection{A multiphoton attack on the quantum oblivious transfer protocol of Ref. \cite{ENGLWW14}}

Ref. \cite{ENGLWW14} demonstrated quantum $1$-out-of-$2$ oblivious transfer in the noisy storage model. The protocol of Ref. \cite{ENGLWW14} uses setup II, which is illustrated in Fig. \ref{setupII}. A subroutine of this protocol is equivalent to the subroutine discussed in section \ref{NJMKW12} and implemented in the protocol of
Ref. \cite{NJMKW12}. The \hyperref[symmetrizationsetup2]{SLII} reporting strategy is used by Bob, as it is claimed by Ref. \cite{ENGLWW14} that this guarantees security against
Alice. However, Ref. \cite{ENGLWW14} does not say whether Bob
only reports single clicks as valid measurement
outcomes, or whether Bob reports multiple clicks as valid measurement outcomes. If Bob only reports single clicks as valid measurement outcomes, then the multiphoton attack I in setup II (\hyperref[MPAII]{MPAII} attack) and the results presented in
section \ref{NJMKW12} apply too. In particular, as illustrated in Fig. \ref{plots} and discussed in section \ref{NJMKW12}, Alice can guess Bob's measurement bases with
high probability by implementing the \hyperref[MPAII]{MPAII}
attack with a large number of photons $k$. The results and analyses of section \ref{lastsec} apply too. %if Bob uses the \hyperref[RLODCII]{RSDCII} reporting strategy.}

%\section{Proofs of Lemmas \ref{lemma0} and \ref{lemma1}, and Theorem \ref{lemma3}}
\section{Proof of Theorem \ref{lemma3}}
\label{appE}

\begin{theorem*}
\label{lemma3repeated}
Suppose that 
\begin{eqnarray}
\label{a14}
d_{i\beta}&=&0,\\
\label{a15}
\eta_{i\beta}&=&\eta_i\in(0,1),\\
\label{new2}
\eta_0&\neq& \eta_1,
\end{eqnarray}
for $i,\beta\in\{0,1\}$, and $\mathcal{B}_0$ and $\mathcal{B}_1$ are
arbitrary distinct qubit orthogonal bases.
If Alice sends Bob a pulse
of $k$ photons encoding a state $\rho$, with $k\in\{0,1,2\}$
chosen by Alice and unknown to Bob, then the only
  probabilistic reporting strategy that guarantees to Bob that Alice
  cannot obtain any information about $\beta$ from his message $m$ is
  the trivial strategy (\ref{new6}).
\end{theorem*}

\begin{proof}

We first show Theorem 1 analytically for the case that $\mathcal{B}_0$ and $\mathcal{B}_1$ are the computational and Hadamard bases, respectively, with $\lvert \psi_{00}\rangle=\lvert 0\rangle$,
$\lvert \psi_{10}\rangle=\lvert 1\rangle$,
$\lvert \psi_{01}\rangle=\lvert +\rangle$ and $\lvert \psi_{11}\rangle=\lvert -\rangle$. We then show Theorem 1 numerically for the case that $\mathcal{B}_0$ and $\mathcal{B}_1$ are arbitrary distinct qubit orthogonal bases.

\subsection{Analytical proof for the case that $\mathcal{B}_0$ and $\mathcal{B}_1$ are the computational and Hadamard bases, respectively.}

We consider the case that $\mathcal{B}_0$ and $\mathcal{B}_1$ are the computational and Hadamard bases, respectively.
Alice does not obtain any information about $\beta$ from Bob's message $m$ if and only if
\begin{equation}
\label{a12}
P_{\text{report}}(1\lvert 1, \rho, k)=P_{\text{report}}(1\lvert 0, \rho, k),
\end{equation}
for any $k-$qubit state $\rho$ and for $k\in\{0,1,2\}$.

We first note that if (\ref{new6}) of the main text holds then, from (\ref{new7}) of the main text, (\ref{a12}) holds too, for any $k-$qubit state $\rho$ and for any $k\in\{0,1,2,\ldots\}$. Thus, Alice does not obtain any information about $\beta$ from Bob's message $m$ in this case.

Now we show that Bob is guaranteed that Alice cannot obtain any information about $\beta$ from his message $m$ only if (\ref{new6}) of the main text holds. From (\ref{a4}) of the main text and (\ref{a12}), it follows that
\begin{equation}
\label{a13}
\sum_{i=0}^1\sum_{j=0}^1\bigl(S_{ij1}P_{\text{det}}(i,j\lvert 1, \rho, k)- S_{ij0}P_{\text{det}}(i,j\lvert 0, \rho, k) \bigr)=0,
\end{equation}
for any $k-$qubit state $\rho$ encoded in a pulse of $k$ photons, and for $k\in\{0,1,2\}$.

From (\ref{a13}), we establish a set of eight linear equations, by considering the case $k=0$, the case $k=1$ with $\rho\in\bigl\{\lvert 0\rangle\langle 0\rvert,\lvert 1\rangle\langle 1\rvert,\lvert +\rangle\langle +\rvert,\lvert -\rangle\langle -\rvert\bigr\}$, and the case $k=2$ with $\rho\in\bigl\{\lvert 00\rangle\langle 00\rvert,\lvert 11\rangle\langle 11\rvert,\lvert ++\rangle\langle ++\rvert\bigr\}$. Considering the state $\rho=\lvert --\rangle\langle --\rvert$ does not give any extra information. It does not simplify the proof either. As we show below, the only solution for this set of equations is given by (\ref{new6}) of the main text.

From (\ref{xab2}) of the main text, (\ref{a14}), (\ref{a15}) and (\ref{a13}), we obtain in the case $k=0$ that
\begin{equation}
\label{a16}
S_{001}=S_{000}.
\end{equation}
In the case $k=1$ with $\rho=\lvert 0\rangle\langle 0\rvert$, we obtain from (\ref{xab1main}) of the main text that $q_0=\lvert \langle0\vert0\rangle\rvert^2=1$ and $q_1=\lvert \langle+\vert0\rangle\rvert^2=\frac{1}{2}$. Thus, from (\ref{xab2}) of the main text, (\ref{a14}), (\ref{a15}) and (\ref{a13}), we obtain 
\begin{eqnarray}
\label{a17}
&&\Bigl[1-\Bigl(\frac{\eta_0+\eta_1}{2}\Bigr)\Bigr]S_{001}+\frac{\eta_0}{2}S_{101}+\frac{\eta_1}{2}S_{011}\nonumber\\
&&\qquad\qquad\qquad\quad-(1-\eta_0)S_{000}-\eta_0S_{100}=0.
\end{eqnarray}
In the case $k=1$ with $\rho=\lvert 1\rangle\langle 1\rvert$, we obtain from (\ref{xab1main}) of the main text that $q_0=\lvert \langle0\vert1\rangle\rvert^2=0$ and $q_1=\lvert \langle+\vert1\rangle\rvert^2=\frac{1}{2}$. Thus, from (\ref{xab2}) of the main text, (\ref{a14}), (\ref{a15}) and (\ref{a13}), we obtain 
\begin{eqnarray}
\label{a18}
&&\Bigl[1-\Bigl(\frac{\eta_0+\eta_1}{2}\Bigr)\Bigr]S_{001}+\frac{\eta_0}{2}S_{101}+\frac{\eta_1}{2}S_{011}\nonumber\\
&&\qquad\qquad\qquad\quad-(1-\eta_1)S_{000}-\eta_1S_{010}=0.
\end{eqnarray}
In the case $k=1$ with $\rho=\lvert +\rangle\langle +\rvert$, we obtain from (\ref{xab1main}) of the main text that $q_0=\lvert \langle0\vert+\rangle\rvert^2=\frac{1}{2}$ and $q_1=\lvert \langle+\vert+\rangle\rvert^2=1$. Thus, from (\ref{xab2}) of the main text, (\ref{a14}), (\ref{a15}) and (\ref{a13}), we obtain 
\begin{eqnarray}
\label{a19}
&&(1-\eta_0)S_{001}+\eta_0S_{101}-\Bigl[1-\Bigl(\frac{\eta_0+\eta_1}{2}\Bigr)\Bigr]S_{000}\nonumber\\
&&\qquad\qquad\qquad\qquad\quad-\frac{\eta_1}{2}S_{010}-\frac{\eta_0}{2}S_{100}=0.
\end{eqnarray}
In the case $k=1$ with $\rho=\lvert -\rangle\langle -\rvert$, we obtain from (\ref{xab1main}) of the main text that $q_0=\lvert \langle0\vert-\rangle\rvert^2=\frac{1}{2}$ and $q_1=\lvert \langle+\vert-\rangle\rvert^2=0$. Thus, from (\ref{xab2}) of the main text, (\ref{a14}), (\ref{a15}) and (\ref{a13}), we obtain 
\begin{eqnarray}
\label{a20}
&&(1-\eta_1)S_{001}+\eta_1S_{011}-\Bigl[1-\Bigl(\frac{\eta_0+\eta_1}{2}\Bigr)\Bigr]S_{000}\nonumber\\
&&\qquad\quad\qquad\qquad\qquad-\frac{\eta_1}{2}S_{010}-\frac{\eta_0}{2}S_{100}=0.
\end{eqnarray}
In the case $k=2$ with $\rho=\lvert 00\rangle\langle 00\rvert$, we obtain from (\ref{xab1main}) of the main text that $q_0=\lvert \langle0\vert0\rangle\rvert^2=1$ and $q_1=\lvert \langle+\vert0\rangle\rvert^2=\frac{1}{2}$. Thus, from (\ref{xab2}) of the main text, (\ref{a14}), (\ref{a15}) and (\ref{a13}), we obtain
\begin{eqnarray}
\label{a21}
&&\biggl(1\!-\!\frac{\eta_0\!+\!\eta_1}{2}\biggr)^2S_{001}
+\biggl(\Bigl(1\!-\!\frac{\eta_0}{2}\Bigr)^2\!-\!\Bigl(1\!-\!\frac{\eta_0\!+\!\eta_1}{2}\Bigr)^2\biggr)S_{011}\nonumber\\
&&\qquad+\biggl(\Bigl(1\!-\!\frac{\eta_1}{2}\Bigr)^2\!-\!\Bigl(1\!-\!\frac{\eta_0 \!+\!\eta_1}{2}\Bigr)^2\biggr)S_{101}+\Bigl(\frac{\eta_0\eta_1}{2}\Bigr)S_{111}\nonumber\\
&&\qquad-(1\!-\!\eta_0)^2S_{000}-\bigl(1\!-\!(1\!-\!\eta_0)^2\bigr)S_{100}=0.
\end{eqnarray}
In the case $k=2$ with $\rho=\lvert ++\rangle\langle ++\rvert$, we obtain from (\ref{xab1main}) of the main text that $q_0=\lvert \langle0\vert+\rangle\rvert^2=\frac{1}{2}$ and $q_1=\lvert \langle+\vert+\rangle\rvert^2=1$. Thus, from (\ref{xab2}) of the main text, (\ref{a14}), (\ref{a15}) and (\ref{a13}), we obtain
\begin{eqnarray}
\label{a22}
&&(1\!-\!\eta_0)^2S_{001}+\bigl(1\!-\!(1\!-\!\eta_0)^2\bigr)S_{101}-\biggl(1\!-\!\frac{\eta_0\!+\!\eta_1}{2}\biggr)^2S_{000}
\nonumber\\
&&\qquad-\biggl(\Bigl(1\!-\!\frac{\eta_0}{2}\Bigr)^2\!-\!\Bigl(1\!-\!\frac{\eta_0\!+\!\eta_1}{2}\Bigr)^2\biggr)S_{010}-\frac{\eta_0\eta_1}{2}S_{110}\nonumber\\
&&\qquad-\biggl(\Bigl(1\!-\!\frac{\eta_1}{2}\Bigr)^2\!-\!\Bigl(1\!-\!\frac{\eta_0\!+\!\eta_1}{2}\Bigr)^2\biggr)S_{100}=0.
\end{eqnarray}
In the case $k=2$ with $\rho=\lvert 11\rangle\langle 11\rvert$, we obtain from (\ref{xab1main}) of the main text that $q_0=\lvert \langle0\vert1\rangle\rvert^2=0$ and $q_1=\lvert \langle+\vert1\rangle\rvert^2=\frac{1}{2}$. Thus, from (\ref{xab2}) of the main text, (\ref{a14}), (\ref{a15}) and (\ref{a13}), we obtain
\begin{eqnarray}
\label{a23}
&&\biggl(1\!-\!\frac{\eta_0\!+\!\eta_1}{2}\biggr)^2S_{001}
+\biggl(\Bigl(1\!-\!\frac{\eta_0}{2}\Bigr)^2\!-\!\Bigl(1\!-\!\frac{\eta_0\!+\!\eta_1}{2}\Bigr)^2\biggr)S_{011}\nonumber\\
&&\qquad+\biggl(\Bigl(1\!-\!\frac{\eta_1}{2}\Bigr)^2\!-\!\Bigl(1\!-\!\frac{\eta_0 \!+\!\eta_1}{2}\Bigr)^2\biggr)S_{101}+\Bigl(\frac{\eta_0\eta_1}{2}\Bigr)S_{111}\nonumber\\
&&\qquad-(1\!-\!\eta_1)^2S_{000}-\bigl(1\!-\!(1\!-\!\eta_1)^2\bigr)S_{010}=0.
\end{eqnarray}

We solve the system of equations (\ref{a16}) -- (\ref{a23}) in two parts. First, we solve the system of equations (\ref{a16}) -- (\ref{a20}). Then, using the obtained solutions, we solve the remaining system of equations (\ref{a21}) -- (\ref{a23}).

We solve the system of equations (\ref{a16}) -- (\ref{a20}). We first eliminate the variable $S_{001}$ by substituting (\ref{a16}) in (\ref{a17}) -- (\ref{a20}). We obtain
\begin{eqnarray}
\label{a24}
\Bigl(\!\frac{\eta_0-\eta_1}{2}\!\Bigr)S_{000}\!+\!\frac{\eta_0}{2}S_{101}\!+\!\frac{\eta_1}{2}S_{011}-\eta_0S_{100}&\!=\!&0,\\
\label{a25}
\Bigl(\!\frac{\eta_1-\eta_0}{2}\!\Bigr)S_{000}\!+\!\frac{\eta_0}{2}S_{101}\!+\!\frac{\eta_1}{2}S_{011}\!-\!\eta_1S_{010}&\!=\!&0,\\
\label{a26}
\Bigl(\!\frac{\eta_1-\eta_0}{2}\!\Bigr)S_{000}\!+\!\eta_0S_{101}\!-\!\frac{\eta_1}{2}S_{010}\!-\!\frac{\eta_0}{2}S_{100}&\!=\!&0,\\
\label{a27}
\Bigl(\!\frac{\eta_0-\eta_1}{2}\!\Bigr)S_{000}\!+\!\eta_1S_{011}\!-\!\frac{\eta_1}{2}S_{010}\!-\!\frac{\eta_0}{2}S_{100}&\!=\!&0.
\end{eqnarray}
We use (\ref{a24}) to eliminate $S_{000}$ in (\ref{a25}) -- (\ref{a27}). Adding (\ref{a24}) to (\ref{a25}) and (\ref{a26}), and subtracting (\ref{a24}) from (\ref{a27}), and rearranging terms, we obtain respectively
\begin{eqnarray}
\label{a28}
\eta_1S_{011}+\eta_0S_{101}-\eta_1S_{010}-\eta_0S_{100}&=&0,\\
\label{a29}
\eta_1S_{011}+3\eta_0S_{101}-\eta_1S_{010}-3\eta_0S_{100}&=&0,\\
\label{a30}
\eta_1S_{011}-\eta_0S_{101}-\eta_1S_{010}+\eta_0S_{100}&=&0.
\end{eqnarray}
Subtracting (\ref{a28}) from (\ref{a29}), or from (\ref{a30}), and rearranging terms, we obtain 
\begin{equation}
\label{a31}
S_{101}=S_{100}.
\end{equation}
From (\ref{a28}) and (\ref{a31}) we obtain
\begin{equation}
\label{a32}
S_{011}=S_{010}.
\end{equation}
From (\ref{a24}), (\ref{a31}) and (\ref{a32}), we obtain
\begin{equation}
\label{a33}
(\eta_0-\eta_1)S_{000}+\eta_1S_{010}-\eta_0S_{100}=0.
\end{equation}
Thus, the solutions to the system of equations (\ref{a16}) -- (\ref{a20}) consist in (\ref{a16}), and (\ref{a31}) -- (\ref{a33}).

Now we use the obtained solutions, (\ref{a16}), and (\ref{a31}) -- (\ref{a33}), to solve the remaining system of equations (\ref{a21}) -- (\ref{a23}). Substituting (\ref{a16}), (\ref{a31}) and (\ref{a32}) in (\ref{a21}), and multiplying by four, we obtain
\begin{eqnarray}
\label{a34}
&&\bigl((2\!-\!(\eta_0\!+\!\eta_1))^2 - 4(1\!-\!\eta_0)^2\bigr)S_{000}\nonumber\\
&&\!\!\quad +\bigl((2\!-\!\eta_0)^2\!-\!(2\!-\!(\eta_0\!+\!\eta_1))^2\bigr)S_{010}+2\eta_0\eta_1S_{111}\nonumber\\
&&\!\!\quad+\bigl((2\!-\!\eta_1)^2\!-\!(2\!-\!(\eta_0 \!+\!\eta_1))^2 -4(1\!-\!(1\!-\!\eta_0)^2)\bigr)S_{100}=0.\nonumber\\
\end{eqnarray}
Substituting (\ref{a16}) and (\ref{a31}) in (\ref{a22}), and multiplying by four, we obtain
\begin{eqnarray}
\label{a35}
&&\bigl(4(1\!-\!\eta_0)^2-(2\!-\!(\eta_0\!+\!\eta_1))^2\bigr)S_{000}\nonumber\\
&&\quad+\bigl(4(1\!-\!(1\!-\!\eta_0)^2)-(2\!-\!\eta_1)^2\!+\!(2\!-\!(\eta_0\!+\!\eta_1))^2\bigr)S_{100}\nonumber\\
&&\quad-\bigl((2\!-\!\eta_0)^2\!-\!(2\!-\!(\eta_0\!+\!\eta_1))^2\bigr)S_{010}-2\eta_0\eta_1S_{110}=0.\nonumber\\
\end{eqnarray}
Substituting (\ref{a16}), (\ref{a31}) and (\ref{a32}) in (\ref{a23}), and multiplying by four, we obtain
\begin{eqnarray}
\label{a36}
&&\bigl((2\!-\!(\eta_0\!+\!\eta_1))^2 -4(1\!-\!\eta_1)^2\bigr)S_{000}\nonumber\\
&&\quad+\bigl((2\!-\!\eta_0)^2\!-\!(2\!-\!(\eta_0\!+\!\eta_1))^2-4(1\!-\!(1\!-\!\eta_1)^2)\bigr)S_{010}\nonumber\\
&&\quad+\bigl((2\!-\!\eta_1)^2\!-\!(2\!-\!(\eta_0 \!+\!\eta_1))^2\bigr)S_{100}+2\eta_0\eta_1S_{111}=0.\nonumber\\
\end{eqnarray}
Thus, now we have a system of four equations, (\ref{a33}) -- (\ref{a36}), with five variables, $S_{000}$, $S_{100}$, $S_{010}$, $S_{110}$ and $S_{111}$. By subtracting (\ref{a34}) from (\ref{a36}), we eliminate $S_{111}$, and we obtain
\begin{eqnarray}
\label{a37}
&&4\bigl((1\!-\!\eta_0)^2-(1\!-\!\eta_1)^2\bigr)S_{000}-4(1\!-\!(1\!-\!\eta_1)^2)S_{010}\nonumber\\
&&\qquad\qquad+4(1\!-\!(1\!-\!\eta_0)^2)S_{100}=0.
\end{eqnarray}
Thus, now we have a system of three equations, (\ref{a33}), (\ref{a35})  and (\ref{a37}), with four variables, $S_{000}$, $S_{100}$, $S_{010}$ and $S_{110}$. We use (\ref{a33}) to eliminate $S_{010}$ from (\ref{a35})  and (\ref{a37}). We multiply (\ref{a33}) by $-\bigl((2-\eta_0)^2-(2-(\eta_0+\eta_1))^2\bigr)/\eta_1$ and subtract from (\ref{a35}), and we arrange terms to obtain
\begin{equation}
\label{a38}
(\eta_1-\eta_0)S_{000}+(\eta_0-3\eta_1)S_{100}+2\eta_1S_{110}=0.
\end{equation}
We multiply (\ref{a33}) by $-4(1\!-\!(1\!-\!\eta_1)^2)/\eta_1$ and subtract from (\ref{a37}), and we arrange terms to obtain
\begin{equation}
\label{new3}
(\eta_0-\eta_1)(S_{000}-S_{100}) = 0.
\end{equation}
Thus, from (\ref{new2}) and (\ref{new3}), we get
\begin{equation}
\label{a39}
S_{000}=S_{100}.
\end{equation}
We substitute (\ref{a39}) in (\ref{a38}) and obtain
\begin{equation}
\label{a40}
S_{100}=S_{110}.
\end{equation}
Since $\eta_1>0$, from (\ref{a33}) and (\ref{a39}), we obtain
\begin{equation}
\label{a41}
S_{100}=S_{010}.
\end{equation}
Since $\eta_0>0$ and $\eta_1>0$, it is straightforward to obtain from (\ref{a34}), (\ref{a39}) and (\ref{a41}) that
\begin{equation}
\label{a42}
S_{111}=S_{100}.
\end{equation}
Thus, from (\ref{a16}), (\ref{a31}), (\ref{a32}) and (\ref{a39}) -- (\ref{a42}), we obtain the claimed result (\ref{new6}) of the main text.

We note from (\ref{new3}) that if $\eta_0=\eta_1$ then (\ref{a39}) does not follow. In this case, from (\ref{a15}), (\ref{a16}) and (\ref{a31}) -- (\ref{a35}), it follows straightforwardly that $S_{001}=S_{000}$ and $S_{c_0c_1\beta}=S$ for $(c_0,c_1)\in\{(0,1),(1,0),(1,1)\}$, $\beta\in\{0,1\}$ and any $S_{000},S\in[0,1]$.

\subsection{Computational proof for the case that $\mathcal{B}_0$ and $\mathcal{B}_1$ are arbitrary distinct qubit orthogonal bases}

\subsubsection{Obtaining a system of linear equations}

We suppose that 
\begin{eqnarray}
\label{newa14}
d_{i\beta}&=&0,\\
\label{newa15}
\eta_{i\beta}&=&\eta_i\in(0,1),\\
\label{newnew2}
\eta_0&\neq& \eta_1,
\end{eqnarray}
for $i,\beta\in\{0,1\}$, and $\mathcal{B}_0$ and $\mathcal{B}_1$ are arbitrary qubit orthogonal bases. Below we show numerically that if Alice sends Bob a pulse of $k$ photons encoding a state $\rho$, with $k\in\{0,1,2\}$ chosen by Alice and unknown to Bob, then the only probabilistic reporting strategy that guarantees to Bob that Alice cannot obtain any information about $\beta$ from his message $m$ is the trivial strategy (\ref{new6}) of the main text.

The bases $\mathcal{B}_0$ and $\mathcal{B}_1$ define a plane in the Bloch sphere. Without loss of generality, this plane can be taken as the $x-z$ plane, and $\mathcal{B}_0$ can be taken as the computational basis, with $\lvert\psi_{00}\rangle=\lvert 0\rangle$ and  $\lvert\psi_{10}\rangle=\lvert 1\rangle$. Thus, in general, the states of the basis $\mathcal{B}_1$ are given by
\begin{eqnarray}
\label{newx1}
\lvert\psi_{01}\rangle&=&\cos(a)\lvert 0\rangle+\sin(a)\lvert 1\rangle,\nonumber\\
\lvert\psi_{11}\rangle&=&\sin(a)\lvert 0\rangle-\cos(a)\lvert 1\rangle,
\end{eqnarray}
where $a=\frac{\pi}{4}-\frac{\theta}{2}$ is half the angle in the Bloch sphere between the states $\lvert\psi_{00}\rangle$ and $\lvert\psi_{01}\rangle$, and $\theta$ is the angle in the Bloch sphere between the states $\lvert\psi_{01}\rangle$ and $\lvert +\rangle=\frac{1}{\sqrt{2}}\bigl(\lvert 0\rangle+\lvert 1\rangle\bigr)$, for $\theta\in[0,\frac{\pi}{2})$ and $a\in(0,\frac{\pi}{4}]$.

We proceed as in the analytical proof of Theorem 1. From the condition that Alice cannot obtain any information about Bob's message $m$, we obtain a set of eight linear equations. We solve these equations numerically with a Mathematica program, provided as supplementary material, and we obtain that the only solution corresponds to the trivial strategy (\ref{new6}) of the main text.

Alice does not obtain any information about $\beta$ from Bob's message $m$ if and only if
\begin{equation}
\label{newa12}
P_{\text{report}}(1\lvert 1, \rho, k)=P_{\text{report}}(1\lvert 0, \rho, k),
\end{equation}
for any $k-$qubit state $\rho$ and for $k\in\{0,1,2\}$.

We note that if (\ref{new6}) of the main text holds then, from (\ref{new7}) of the main text, (\ref{newa12}) holds too, for any $k-$qubit state $\rho$ and for any $k\in\{0,1,2,\ldots\}$. Thus, Alice does not obtain any information about $\beta$ from Bob's message $m$ in this case.

Now we show, numerically, that Bob is guaranteed that Alice cannot obtain any information about $\beta$ from his message $m$ only if (\ref{new6}) of the main text holds. From (\ref{a4}) of the main text and (\ref{newa12}), it follows that
\begin{eqnarray}
\label{newa13}
&&\sum_{c_0=0}^1\sum_{c_1=0}^1\bigl[S_{c_0c_11}P_{\text{det}}(c_0,c_1\lvert 1, \rho, k)\nonumber\\
&&\quad\qquad\qquad- S_{c_0c_10}P_{\text{det}}(c_0,c_1\lvert 0, \rho, k) \bigr]=0,
\end{eqnarray}
for any $k-$qubit state $\rho$ encoded in a pulse of $k$ photons, and for $k\in\{0,1,2\}$.

From (\ref{newa13}), we establish a set of eight linear equations, by considering the case $k=0$, the case $k=1$ with $\rho\in\bigl\{\lvert 0\rangle\langle 0\rvert,\lvert 1\rangle\langle 1\rvert,\lvert\psi_{01}\rangle\langle\psi_{01}\rvert,\lvert \psi_{11}\rangle\langle \psi_{11}\rvert\bigr\}$, and the case $k=2$ with $\rho\in\bigl\{\lvert 00\rangle\langle 00\rvert,\lvert 11\rangle\langle 11\rvert,(\lvert \psi_{01}\rangle\langle \psi_{01}\rvert)\otimes(\lvert \psi_{01}\rangle\langle \psi_{01}\rvert)\bigr\}$. As we show below, the only solution for this set of equations is given by (\ref{new6}) of the main text.

From (\ref{xab2}) of the main text, (\ref{newa14}), (\ref{newa15}) and (\ref{newa13}), we obtain in the case $k=0$ that
\begin{equation}
\label{newa16}
S_{001}=S_{000}.
\end{equation}
In the case $k=1$ with $\rho=\lvert 0\rangle\langle 0\rvert$, we obtain from (\ref{xab1main}) of the main text that $q_0=\lvert \langle0\vert0\rangle\rvert^2=1$ and $q_1=\lvert \langle\psi_{01}\vert0\rangle\rvert^2=\cos^2(a)$. Thus, from (\ref{xab2}) of the main text, (\ref{newa14}), (\ref{newa15}) and (\ref{newa13}), we obtain 
\begin{eqnarray}
\label{newa17}
&&\Bigl[1-\eta_0\cos^2(a)-\eta_1\bigl(1-\cos^2(a)\bigr)\Bigr]S_{001}+\eta_0\cos^2(a)S_{101}\nonumber\\
&&\qquad+\eta_1\bigl(1-\cos^2(a)\bigr)S_{011}-(1-\eta_0)S_{000}-\eta_0S_{100}=0.\nonumber\\
\end{eqnarray}
In the case $k=1$ with $\rho=\lvert 1\rangle\langle 1\rvert$, we obtain from (\ref{xab1main}) of the main text that $q_0=\lvert \langle0\vert1\rangle\rvert^2=0$ and $q_1=\lvert \langle\psi_{01}\vert1\rangle\rvert^2=\sin^2(a)=1-\cos^2(a)$. We express $\sin^2(a)$ in terms of $\cos^2(a)$, as this makes the numerical calculation by the Mathematica program easier. Thus, from (\ref{xab2}) of the main text, (\ref{newa14}), (\ref{newa15}) and (\ref{newa13}), we obtain 
\begin{eqnarray}
\label{newa18}
&&\Bigl[1-\eta_0\bigl(1-\cos^2(a)\bigr)-\eta_1\cos^2(a)\Bigr]S_{001}\nonumber\\
&&\quad\qquad+\eta_0\bigl(1-\cos^2(a)\bigr)S_{101}+\eta_1\cos^2(a)S_{011}\nonumber\\
&&\quad\quad\qquad\qquad-(1-\eta_1)S_{000}-\eta_1S_{010}=0.
\end{eqnarray}
In the case $k=1$ with $\rho=\lvert \psi_{01}\rangle\langle \psi_{01}\rvert$, we obtain from (\ref{xab1main}) of the main text that $q_0=\lvert \langle0\vert\psi_{01}\rangle\rvert^2=\cos^2(a)$ and $q_1=\lvert \langle\psi_{01}\vert\psi_{01}\rangle\rvert^2=1$. Thus, from (\ref{xab2}) of the main text, (\ref{newa14}), (\ref{newa15}) and (\ref{newa13}), we obtain 
\begin{eqnarray}
\label{newa19}
&&(1-\eta_0)S_{001}+\eta_0S_{101}-\eta_1\bigl(1-\cos^2(a)\bigr)S_{010}\nonumber\\
&&\quad-\Bigl[1-\eta_0\cos^2(a)-\eta_1\bigl(1-\cos^2(a)\bigr)\Bigr]S_{000}\nonumber\\&&\qquad\qquad-\eta_0\cos^2(a)S_{100}=0.
\end{eqnarray}
In the case $k=1$ with $\rho=\lvert \psi_{11}\rangle\langle \psi_{11}\rvert$, we obtain from (\ref{xab1main}) of the main text that $q_0=\lvert \langle0\vert\psi_{11}\rangle\rvert^2=1-\cos^2(a)$ and $q_1=\lvert \langle\psi_{01}\vert\psi_{11}\rangle\rvert^2=0$. Thus, from (\ref{xab2}) of the main text, (\ref{newa14}), (\ref{newa15}) and (\ref{newa13}), we obtain 
\begin{eqnarray}
\label{newa20}
&&(1-\eta_1)S_{001}+\eta_1S_{011}-\eta_1\cos^2(a)S_{010}\nonumber\\
&&\qquad-\Bigl[1-\eta_0\bigl(1-\cos^2(a)\bigr)-\eta_1\cos^2(a)\Bigr]S_{000}\nonumber\\
&&\qquad\qquad-\eta_0\bigl(1-\cos^2(a)\bigr)S_{100}=0.
\end{eqnarray}
In the case $k=2$ with $\rho=\lvert 00\rangle\langle 00\rvert$, we obtain from (\ref{xab1main}) of the main text that $q_0=\lvert \langle0\vert0\rangle\rvert^2=1$ and $q_1=\lvert \langle\psi_{01}\vert0\rangle\rvert^2=\cos^2(a)$. Thus, from (\ref{xab2}) of the main text, (\ref{newa14}), (\ref{newa15}) and (\ref{newa13}), we obtain
\begin{eqnarray}
\label{newa21}
&&\biggl[1\!-\!\eta_0\cos^2(a)\!-\!\eta_1\bigl(1-\cos^2(a)\bigr)\biggr]^2S_{001}\nonumber\\
&&\qquad+\biggl[\Bigl(1\!-\!\eta_0\cos^2(a)\Bigr)^2\nonumber\\
&&\qquad\qquad\!-\!\Bigl(1\!-\!\eta_0\cos^2(a)\!-\!\eta_1\bigl(1-\cos^2(a)\bigr)\Bigr)^2\biggr]S_{011}\nonumber\\
&&\qquad+\biggl[\Bigl(1\!-\!\eta_1\bigl(1-\cos^2(a)\bigr)\Bigr)^2\nonumber\\
&&\qquad\qquad\!-\!\Bigl(1\!-\!\eta_0\cos^2(a) \!-\!\eta_1\bigl(1-\cos^2(a)\bigr)\Bigr)^2\biggr]S_{101}\nonumber\\
&&\qquad+2\eta_0\eta_1\cos^2(a)\bigl(1-\cos^2(a)\bigr)S_{111}\nonumber\\
&&\qquad-(1\!-\!\eta_0)^2S_{000}-\bigl(1\!-\!(1\!-\!\eta_0)^2\bigr)S_{100}=0.
\end{eqnarray}
In the case $k=2$ with $\rho=(\lvert \psi_{01}\rangle\langle \psi_{01}\rvert)\otimes (\lvert \psi_{01}\rangle\langle \psi_{01}\rvert)$, we obtain from (\ref{xab1main}) of the main text that $q_0=\lvert \langle0\vert\psi_{01}\rangle\rvert^2=\cos^2(a)$ and $q_1=\lvert \langle\psi_{01}\vert\psi_{01}\rangle\rvert^2=1$. Thus, from (\ref{xab2}) of the main text, (\ref{newa14}), (\ref{newa15}) and (\ref{newa13}), we obtain
\begin{eqnarray}
\label{newa22}
&&(1\!-\!\eta_0)^2S_{001}+\bigl(1\!-\!(1\!-\!\eta_0)^2\bigr)S_{101}\nonumber\\
&&\qquad-\Bigl[1\!-\!\eta_0\cos^2(a)\!-\!\eta_1\bigl(1-\cos^2(a)\bigr)\Bigr]^2S_{000}
\nonumber\\
&&\qquad-\Bigl[\Bigl(1\!-\!\eta_0\cos^2(a)\Bigr)^2\nonumber\\
&&\qquad\qquad\!-\!\Bigl(1\!-\!\eta_0\cos^2(a)\!-\!\eta_1\bigl(1-\cos^2(a)\bigr)\Bigr)^2\Bigr]S_{010}\nonumber\\
&&\qquad-2\eta_0\eta_1\cos^2(a)\bigl(1-\cos^2(a)\bigr)S_{110}\nonumber\\
&&\qquad-\biggl[\Bigl(1\!-\!\eta_1\bigl(1-\cos^2(a)\bigr)\Bigr)^2\nonumber\\
&&\qquad\qquad\!-\!\Bigl(1\!-\!\eta_0\cos^2(a)\!-\!\eta_1\bigl(1-\cos^2(a)\bigr)\Bigr)^2\biggr]S_{100}=0.\nonumber\\
\end{eqnarray}
In the case $k=2$ with $\rho=\lvert 11\rangle\langle 11\rvert$, we obtain from (\ref{xab1main}) of the main text that $q_0=\lvert \langle0\vert1\rangle\rvert^2=0$ and $q_1=\lvert \langle\psi_{01}\vert1\rangle\rvert^2=1-\cos^2(a)$. Thus, from (\ref{xab2}) of the main text, (\ref{newa14}), (\ref{newa15}) and (\ref{newa13}), we obtain
\begin{eqnarray}
\label{newa23}
&&\Bigl[1\!-\!\eta_0\bigl(1-\cos^2(a)\bigr)\!-\!\eta_1\cos^2(a)\Bigr]^2S_{001}\nonumber\\
&&\qquad+\Bigl[\Bigl(1\!-\!\eta_0\bigl(1-\cos^2(a)\bigr)\Bigr)^2\nonumber\\
&&\qquad\qquad\!-\!\Bigl(1\!-\!\eta_0\bigl(1-\cos^2(a)\bigr)\!-\!\eta_1\cos^2(a)\Bigr)^2\Bigr]S_{011}\nonumber\\
&&\qquad+\Bigl[\Bigl(1\!-\!\eta_1\cos^2(a)\Bigr)^2\nonumber\\
&&\qquad\qquad\!-\!\Bigl(1\!-\!\eta_0\bigl(1-\cos^2(a)\bigr)^2 \!-\!\eta_1\cos^2(a)\Bigr)^2\Bigr]S_{101}\nonumber\\
&&\qquad+2\eta_0\eta_1\cos^2(a)\bigl(1-\cos^2(a)\bigr)S_{111}\nonumber\\
&&\qquad-(1\!-\!\eta_1)^2S_{000}-\bigl(1\!-\!(1\!-\!\eta_1)^2\bigr)S_{010}=0.
\end{eqnarray}

\subsubsection{Solving numerically the system of linear equations}

We solve numerically the system of linear equations given by (\ref{newa16}) -- (\ref{newa23}), using a Mathematica program, which we attach as supplementary material. We use the code `Reduce', with the system of equations (\ref{newa16}) -- (\ref{newa23}). The program gives as output various solutions to the system of equations, according to different values of $\eta_0$, $\eta_1$ and $\cos^2(a)$, also output by the program. In particular, in the last line, when simplified, the program outputs the conditions $\eta_0\neq 0$, $\eta_1\neq 0$, $\eta_0\neq \eta_1$ and $\cos^2(a)\neq 1$, and the solution corresponding to the trivial probabilistic reporting strategy (\ref{new6}) of the main text, as claimed.

Furthermore, in the penultimate line, when simplified, the program outputs the conditions $\eta_0=\eta_1\neq 0$ and $\cos^2(a)\neq 1$, and the solution $S_{001}=S_{000}$ and $S_{c_0c_1\beta}=S$, for $(c_0,c_1)\in\{(0,1),(1,0),(1,1)\}$ and for any $S_{000},S$; which is the solution that we got in the proof of Theorem 1 under the conditions $\eta_0=\eta_1\neq 0$ for the particular case in which $\mathcal{B}_0$ and $\mathcal{B}_1$ are the computational and Hadamard bases. The program also outputs, in the eigth line, when simplified, the conditions $\cos^2(a)=1$, $\eta_0\neq 0$ and $\eta_1\neq 0$, and the solution $S_{c_0c_10}=S_{c_0c_11}$ for $(c_0,c_1)\neq (1,1)$; which is straightforward to deduce analytically, and which corresponds to the case $\mathcal{B}_0=\mathcal{B}_1$.

The output lines can be seen in the Mathematica program, which we attach as supplementary material. In the program, we use the variables `$N0$' and `$N1$' instead of `$\eta_0$' and `$\eta_1$', and `$Sc_0c_1\beta$' instead of `$S_{c_0c_1\beta}$', respectively, for $c_0,c_1,\beta\in\{0,1\}$.
\end{proof}


%\bibliography{postdocbiblio}
\end{document}